%% file: draft.tex
\documentclass{article}
\usepackage{shortcuts}

\input{stdinc}

\title{Exponentially Improved Dimensionality Reduction for $\ell_1$: Subspace Embeddings and Independence Testing}
\author{
Yi Li\\NTU\\\texttt{yili@ntu.edu.sg} \and 
David P. Woodruff\\CMU\\\texttt{dwoodruf@cs.cmu.edu} \and 
Taisuke Yasuda\\CMU\\\texttt{taisukey@cs.cmu.edu}
}
\begin{document}

\begin{titlepage}
\maketitle
\thispagestyle{empty}

\begin{abstract}
Despite many applications, dimensionality reduction in the $\ell_1$-norm is much less understood than in the Euclidean norm. We give two new oblivious dimensionality reduction techniques for the $\ell_1$-norm which improve {\it exponentially} over prior ones:
\begin{enumerate}
\item We design a distribution over random matrices $\bfS \in \mathbb{R}^{r \times n}$, where $r = 2^{\tilde O(d/(\eps \delta))}$, such that given any matrix $\bfA \in \mathbb{R}^{n \times d}$, with probability at least $1-\delta$, simultaneously for all $\bfx$, $\|\bfS\bfA\bfx\|_1 = (1 \pm \epsilon)\|\bfA\bfx\|_1$. Note that $\bfS$ is linear, does not depend on $\bfA$, and maps $\ell_1$ into $\ell_1$. Our distribution provides an exponential improvement on the previous best known map of Wang and Woodruff (SODA, 2019), which required $r = 2^{2^{\Omega(d)}}$, even for constant $\eps$ and $\delta$. Our bound is optimal, up to a polynomial factor in the exponent, given a known $2^{\sqrt d}$ lower bound for constant $\eps$ and $\delta$. 

\item We design a distribution over matrices $\bfS \in \mathbb{R}^{k \times n}$, where $k = 2^{O(q^2)}(\eps^{-1} q \log d)^{O(q)}$, such that given any $q$-mode tensor $\bfA \in (\mathbb{R}^{d})^{\otimes q}$, one can estimate the entrywise $\ell_1$-norm $\|\bfA\|_1$ from $\bfS(\bfA)$. Moreover, $\bfS = \bfS^1 \otimes \bfS^2 \otimes \cdots \otimes \bfS^q$ and so given vectors $\bfu_1, \ldots, \bfu_q \in \mathbb{R}^d$, one can compute $\bfS(\bfu_1 \otimes \bfu_2 \otimes \cdots \otimes \bfu_q)$ in time $2^{O(q^2)}(\epsilon^{-1} q \log d)^{O(q)}$, which is much faster than the $d^q$ time required to form $\bfu_1 \otimes \bfu_2 \otimes \cdots \otimes \bfu_q$. Our linear map gives a streaming algorithm for independence testing using space $2^{O(q^2)}(\eps^{-1} q \log d)^{O(q)}$, improving the previous doubly exponential $(\eps^{-1} \log d)^{q^{O(q)}}$ space bound of Braverman and Ostrovsky (STOC, 2010). 
\end{enumerate}
For subspace embeddings, we also study the setting when $\bfA$ is itself drawn from distributions with independent entries, and obtain a polynomial embedding dimension. For independence testing, we also give algorithms for any distance measure with a polylogarithmic-sized sketch and satisfying an approximate triangle inequality. 
\end{abstract}
\end{titlepage}

\tableofcontents

\newpage

\section{Introduction}\label{sec:intro}
Dimensionality reduction refers to mapping a set of high-dimensional vectors to a set of low-dimensional vectors while preserving their lengths and pairwise distances. A celebrated result is the Johnson-Lindenstrauss embedding, which asserts that for a random linear map $\bfS : \mathbb{R}^n \rightarrow \mathbb{R}^r$, for any fixed $\bfx \in \mathbb{R}^n$, we have $\|\bfS\bfx\|_2 = (1 \pm \eps) \|\bfx\|_2$ with probability $1-\delta$. It is necessary and sufficient for the sketching dimension $r$ to be $\Theta(\epsilon^{-2} \log(1/\delta))$ \cite{johnson1984extensions, DBLP:conf/focs/LarsenN17}. A key property of $\bfS$ is that it is {\it linear} and {\it oblivious}, meaning that it is a linear map that does not depend on the point set. This makes it applicable in settings such as the widely used streaming model, where one sees coordinates or updates to coordinates one at a time (see, e.g., \cite{muthukrishnan2005data,CGHJ12} for surveys) and the distributed model where points are shared across servers (see, e.g., \cite{boutsidis2016optimal}, for a discussion of different models). Here it is crucial that for points $\bfx$ and $\bfy$, $\bfS(\bfx+\bfy) = \bfS\bfx + \bfS\bfy$, and $\bfS$ does not depend on $\bfx$ or $\bfy$. In this case, if one receives a new point $\bfz$ chosen independently of $\bfS$, then $\bfS$ still has a good probability of preserving the length of $\bfz$, whereas data-dependent linear maps $\bfS$ may change with the addition of $\bfz$, and are often slower \cite{indyk2000identifying}. For these reasons, our focus is  on linear oblivious dimensionality reduction, often referred to as ``sketching".  

For many problems, the $1$-norm $\|\bfx\|_1 = \sum_{i=1}^n |\bfx_i|$ is more appropriate than the Euclidean norm. Indeed, this norm is used in applications demanding robustness since it is less sensitive to changes in individual coordinates. As the $1$-norm is twice the variation distance between distributions, it is often the metric of choice for comparing distributions \cite{IM08,BCLMO10,braverman2010measuring,MV15}. A sample of applications involving the $1$-norm includes clustering \cite{DBLP:conf/soda/FeldmanMSW10,labib2005hardware}, regression \cite{clarkson2005subgradient, DBLP:conf/stoc/SohlerW11,DBLP:conf/soda/ClarksonDMMMW13,meng2013low,DBLP:conf/colt/WoodruffZ13,DBLP:conf/soda/ClarksonW15,clarkson2017low,woodruff2014sketching}, time series analysis \cite{dodge1992l1,lawrence2019robust}, internet traffic monitoring \cite{feigenbaum2002approximate}, multimodal and similarity search \cite{aggarwal2001surprising,lin1995fast}. As stated in \cite{aggarwal2001surprising}, ``the Manhattan distance metric is consistently more preferable than the Euclidean distance metric for high dimensional data mining applications". 

While useful for the Euclidean norm, the Johnson-Lindenstrauss embedding completely fails if one wants for a vector $\bfx$, that $\|\bfS\bfx\|_1 = (1 \pm \eps) \|\bfx\|_1$ with probability $1-\delta$. Indeed, the results of Wang and Woodruff \cite{DBLP:conf/soda/WangW19} imply nearly tight bounds: for constant $\eps$ and $\delta$, a sketching dimension of $2^{\textrm{poly}(n)}$ is necessary and sufficient\footnote{Their bounds are stated for subspaces, but when applied to $n$ arbitrary points result in this as both an upper and a lower bound, with differing polynomial factors in the exponent. We give more details in Remark \ref{rem:sketch-arbitrary-vectors} of Section \ref{sec:prelim}.}. Indyk \cite{indyk2006stable} shows that if instead one embeds $\bfx$ into a non-normed space, namely, performs a ``median of absolute values" estimator of $\bfS\bfx$, then the dimension can be reduced to $O((\log n)/\eps^2)$. Such a mapping is still linear and oblivious, and this estimator is useful if one desires to approximate the norm of a single vector, for which the dimension becomes $O((\log(1/\delta))/\eps^2)$ and the failure probability is $\delta$. However, this estimator is less useful in optimization problems as it requires solving a non-convex problem after sketching. Thus, there is a huge difference in dimensionality reduction for the Euclidean and $1$-norms. 

This work is motivated by our poor understanding of dimensionality reduction in the $1$-norm, as exemplified by two existing doubly exponential bounds for important problems: preserving a subspace of points and preserving a sum of tensor products, both of which are well-understood for the Euclidean norm. 

\paragraph{Subspace Embeddings.} 
In this problem, one would like a distribution on linear maps $\bfS \in \mathbb{R}^{r \times n}$, for which with constant probability over the choice of $\bfS$, for any matrix $\bfA \in \mathbb{R}^{n \times d}$, simultaneously for all $\bfx \in \mathbb{R}^d$, $\|\bfS\bfA\bfx\|_1 = (1 \pm \eps)\|\bfA\bfx\|_1$. Note that $\bfS$ preserves the lengths of an infinite number of vectors, namely, the entire column span of $\bfA$. Subspace embeddings arise in least absolute deviation regression \cite{DBLP:conf/stoc/SohlerW11,DBLP:conf/soda/ClarksonDMMMW13,meng2013low,DBLP:conf/colt/WoodruffZ13,clarkson2017low} and entrywise $\ell_1$-low rank approximation \cite{song2017low,ban2019ptas,DBLP:conf/soda/MahankaliW21}, among other places. Since a subspace embedding maps the entire subspace into a lower dimensional subspace of $\ell_1$, one can impose arbitrary constraints on $\bfx$, e.g., non-negativity, manifold constraints, regularization, and so on, after computing $\bfS\bfA$. The resulting problem in the sketch space is convex if the constraints are convex. 

For the analogous problem in the Euclidean norm, there is a linear oblivious sketching matrix $S$ with $O((d + \log(1/\delta))/\epsilon^2)$ rows, which is best possible \cite{clarkson2009numerical,nelson2014lower,woodruff2014sketching}. 

For the $1$-norm, we understand much less. The best upper bound \cite{DBLP:conf/soda/WangW19} for an oblivious subspace embedding is for constant $\epsilon$ and $\delta$ and gives a sketching dimension of $2^{2^{O(d)}}$. This bound is obtained by instantiating the $2^{\textrm{poly}(n)}$ bound above with $n = d^{O(d)}$, and union bounding over the points in a net of a subspace. The lower bound on the sketching dimension is, however, only $2^{\Omega(\sqrt{d})}$, representing an exponential gap in our understanding for this fundamental problem \cite{DBLP:conf/soda/WangW19}. 

\paragraph{Independence Testing.}
Another important problem using dimensionality reduction for $\ell_1$ is testing independence in a stream. This problem was introduced by Indyk and McGregor \cite{IM08} and is the following: letting $[d] = \{1, 2, \ldots, d\}$, suppose you are given a stream of items $(i_1, \ldots, i_q) \in [d]^q$. These define an empirical joint distribution $P$ on the $q$ modes defined as follows: if $f(i_1, ..., i_q)$ is the number of occurrences of $(i_1, \ldots, i_q)$ in a stream of length $m$, then $P(i_1, \ldots, i_q) = \frac{1}{m} f(i_1, \ldots, i_q)$. One can also define the marginal distributions $P_j$, for $j = 1, 2, \ldots, q$, where for $i \in [d]$ we have $P_j(i) = \frac{1}{m}\sum_{i_1, \ldots, i_{j-1}, i_{j+1},\dots i_q} f(i_1, \ldots, i_{j-1}, i, i_{j+1},\dots i_q)$. The goal is to compute $\|P - Q\|_1$ with $Q=P_1 \otimes P_2 \otimes \cdots \otimes P_q$, that is, the $1$-norm of the difference of the joint distribution and the product of marginals. If the $q$ modes were independent, then this difference would be $0$, as $P$ would be a product distribution. In general this measures the distance to independence. 

It is important to note that if one were given $P_1, \ldots, P_q$ and $P$, then one could explicitly compute $P_1 \otimes P_2 \otimes \cdots \otimes P_q$, and then compute the median-based sketch $\bfS(P - P_1 \otimes P_2 \otimes \cdots \otimes P_q)$ of Indyk \cite{indyk2006stable} above. The issue is that in the data stream model, the vectors $P, P_1, \ldots, P_q$ are too large to store, and while it is easy to update $\bfS(P)$ given a new tuple in the stream (namely, $\bfS(P) \leftarrow \bfS(P) + \bfS_{i_1, \ldots, i_q}$, where $\bfS_{i_1, \ldots, i_q}$ is the column of $\bfS$ indexed by the new stream element $(i_1, \ldots, i_q)$), it is not clear how to update $\bfS(P_1 \otimes P_2 \otimes \cdots \otimes P_q)$ in a stream. Consequently, a natural approach is to maintain sketches $\bfS^1 P_1, \bfS^2 P_2, \ldots, \bfS^q P_q$ as well as $\bfS P$, and combine these at the end of the stream. A natural way to combine them is to let $\bfS = (\bfS^1) \otimes (\bfS^2) \otimes \cdots \otimes (\bfS^q)$ be the tensor product of the sketches on each mode. 

For the corresponding problem of estimating the Euclidean distance $\|P - P_1 \otimes P_2 \otimes \cdots \otimes P_q\|_2$, recent work \cite{ahle2020oblivious} implies that this can be done with a very small sketching dimension of $O(n/\epsilon^2)$, though such work makes use of the Johnson Lindenstrauss lemma and completely fails for the $1$-norm. 

Despite a number of works on independence testing for the $1$-norm in a stream \cite{IM08,BCLMO10,braverman2010measuring,MV15}, the best upper bound is due to Braverman and Ostrovsky \cite{braverman2010measuring} with a sketching dimension of $(\epsilon^{-1} \log d)^{q^{O(q)}}$, which, while logarithmic in $d$, is doubly exponential in $q$. A natural question is whether this can be improved. 

\subsection{Our Results}
We give exponential improvements in the sketching dimension of linear oblivious maps for both $\ell_1$-subspace embeddings and $\ell_1$-independence testing. 
\\\\
{\bf Subspace Embeddings:} We design a distribution over random matrices $\bfS \in \mathbb{R}^{r \times n}$, where $r = 2^{\textrm{poly}(d/(\epsilon \delta))}$, so that given any matrix $A \in \mathbb{R}^{n \times d}$, with probability at least $1-\delta$, simultaneously for all $\bfx$, $\|\bfS\bfA\bfx\|_1 = (1 \pm \eps)\|\bfA\bfx\|_1$. We present both a sparse embedding which has a dependence of $\log n$ in the base of the exponent, as well as a dense embedding which removes this dependence on $n$ entirely. 

\begin{thm}[Sparse embedding, restatement of Theorem \ref{thm:1+eps-sparse-embedding}]
Let $\eps\in(0,1)$ and $\delta\in (0,1)$. Then there exists a sparse oblivious $\ell_1$ subspace embedding $\bfS$ into $k$ dimensions with
\[
    k = \poly(d,\eps^{-1},\delta^{-1},\log n)^{d/\delta\eps}
\]
such that for any $\bfA\in\mathbb R^{n\times d}$,
\[
    \Pr\braces*{(1-\eps)\norm{\bfA\bfx}_1 \leq \norm{\bfS\bfA\bfx}_1\leq (1+\eps)\norm{\bfA\bfx}_1}\geq 1-\delta. 
\]
\end{thm}

\begin{cor}[Dense embedding, restatement of Corollary \ref{cor:dense-embedding}]
Let $\eps\in(0,1)$ and $\delta\in (0,1)$. Then there exists an oblivious $\ell_1$ subspace embedding $\bfS$ into $k$ dimensions with
\[
    k = \exp\parens*{\tilde O(d/\delta\eps))}
\]
such that for any $\bfA\in\mathbb R^{n\times d}$,
\[
    \Pr\braces*{(1-\eps)\norm{\bfA\bfx}_1 \leq \norm{\bfS\bfA\bfx}_1\leq (1+\eps)\norm{\bfA\bfx}_1}\geq 1-\delta. 
\]
\end{cor}

This is an exponential improvement over the previous bound of $r = 2^{2^{\Omega(d)}}$ \cite{DBLP:conf/soda/WangW19}, which held for constant $\eps$ and $\delta$. Our bound is optimal, up to a polynomial factor in the exponent, given the $2^{\sqrt d}$ lower bound for constant $\eps$ and $\delta$ \cite{DBLP:conf/soda/WangW19}. In fact, this lower bound also implies a lower bound of $2^{\sqrt{1/\delta}}$ as well, so an exponential dependence on $\delta$ is necessary as well. An important feature of $\bfS$ is that $\bfS \cdot \bfA$ can be computed in an expected $O(\mathrm{nnz}(\bfA))$ time, where $\mathrm{nnz}(\bfA)$ denotes the number of non-zero entries of $\bfA$. This is in contrast to the embedding of \cite{DBLP:conf/soda/WangW19}, which requires $2^{2^{\Omega(d)}} \cdot \mathrm{nnz}(\bfA)$ time. 
\\\\
{\bf Independence Testing:} We design a distribution over matrices $\bfS \in \mathbb{R}^{k \times n}$, where $k = \poly(q \epsilon^{-1} \log d)$, so that given any $q$-mode tensor $\bfA \in (\mathbb{R}^{d})^{\otimes q}$, one can estimate the entrywise $1$-norm $\|\bfA\|_1$ from $\bfS(\bfA)$. Moreover, $\bfS = \bfT^{\otimes q}$ and so given vectors $\bfu_1, \ldots, \bfu_q \in \mathbb{R}^d$, one can compute $\bfS(\bfu_1 \otimes \bfu_2 \otimes \cdots \otimes \bfu_q)$ in time $2^{O(q^2)}(\eps^{-1} q \log d)^{O(q)}$, which is much faster than the $d^q$ time required to form $\bfu_1 \otimes \bfu_2 \otimes \cdots \otimes \bfu_q$. Our linear map can be applied in a stream since we can sketch each marginal and then take the tensor product of sketches, yielding a streaming algorithm for independence testing using $2^{O(q^2)}(\eps^{-1} q \log d)^{O(q)}$ bits of space. 

\begin{thm}[Restatement of Theorem \ref{thm:independence-testing}]
Suppose that the stream length $m = \poly(d^q)$. There is a randomized sketching algorithm which outputs a $(1\pm\eps)$-approximation to $\|P-Q\|_1$ with probability at least $0.9$, using $\exp(O(q^2 + q\log(q/\eps)+q\log\log d))$ bits of space. The update time is $\exp(O(q^2 + q\log(q/\eps)+q\log\log d))$.
\end{thm}

This improves the previous doubly exponential $(\epsilon^{-1} \log d)^{q^{O(q)}}$ space bound \cite{braverman2010measuring}. 
\\\\
For subspace embeddings, we also study the setting when $\bfA$ is itself drawn from distributions with certain properties, and obtain a polynomial embedding dimension. This captures natural statistical problems when the design matrix $\bfA$ for regression, is itself random. Our various results here are discussed in Section \ref{section:subspace_embeddings_random}. 

A byproduct of our sketch is the ability to preserve the $1$-norm of a matrix $\bfP$ by left and right multiplying by independent draws $\bfS^1$ and $\bfS^2$ of our sketch, where we show that $\Theta(\|\bfP\|_1) \leq \|\bfS^1\bfP\bfS^2\|_1 = O(1/\alpha^2) \|\bfP\|_1$ where $\bfS^1\bfP\bfS^2$ is a $d^{\alpha} \times d^{\alpha}$ matrix. Here $\alpha \in (0,1)$ can be any constant; previously, no such trade-off was known. 

\begin{thm}[Restatement of Theorem \ref{thm:l1-entrywise-embedding}]
Let $\delta\in (0,1)$ and $\alpha\in(0,1)$. Then there exists a sparse oblivious $\ell_1$ entrywise embedding $\bfS$ into $k$ dimensions with
\[
    k = \parens*{\frac{d}{\delta}\log n}^\alpha \poly(\delta^{-1},\log n)
\]
such that for any $\bfA\in\mathbb R^{n\times d}$,
\[
    \Pr\braces*{\Omega(1)\norm{\bfA}_1 \leq \norm{\bfS\bfA}_1 \leq O\parens*{\frac1{\delta\alpha}}\norm{\bfA}_1} \geq 1 - \delta.
\]
\end{thm}

We also give a matching lower bound showing that for any oblivious sketch $\bfS^1$ with $r$ rows, the distortion between $\|\bfS^1\bfP\|_1$ and $\|\bfP\|_1$ is $\Omega \left (\frac{\log d}{\log r} \right )$. Thus, with $r = d^\alpha$ dimensions, the distortion must be at least
\[
    \frac{\log d}{\log r} = \frac{\log d}{\log d^\alpha} = \frac1\alpha. 
\]

\begin{thm}[Restatement of Theorem \ref{thm:l1-entrywise-embedding-lower-bound}]
Let $\bfS$ be a fixed $r\times d$ matrix. Then there is a distribution $\mu$ over $d\times d$ matrices such that if
\[
    \Pr_{\bfA\sim\mu}\parens*{\norm{\bfA}_1\leq \norm{\bfS\bfA}_1\leq \kappa \norm{\bfA}_1}\geq \frac23
\]
then $\kappa = \Omega((\log d)/(\log r))$. 
\end{thm}

For independence testing, we also give algorithms for any distance measure with a polylogarithmic-sized sketch and satisfying an approximate triangle inequality; these include many functions in \cite{braverman2010zero}. For example, we handle the robust Huber loss and $\ell_p$-measures for $0 < p < 2$. 

\subsection{Our Techniques}\label{sec:techniques}
We begin by explaining our techniques for subspace embeddings, and then transition to independence testing. 

\subsubsection{Subspace Embeddings} 
The linear oblivious sketch we use is a twist, both algorithmically and analytically, to a methodology originating from the data stream literature for approximating frequency moments \cite{DBLP:conf/stoc/IndykW05,DBLP:conf/soda/BhuvanagiriGKS06}. These methods involve sketches which subsample the coordinates of a vector at geometrically decreasing rates $1, 1/2, 1/4, 1/8, \ldots, 1/n$, and apply an independent \textsf{CountSketch} matrix \cite{CCF02} (see Definition \ref{def:countsketch}) to the surviving coordinates at each scale. Analyses of this sketch for data streams does not apply here, since it involves nonlinear median operations, but here we must embed $\ell_1$ into $\ell_1$. These sketches {\it have} been used for embedding single vectors or matrices in $\ell_1$ into $\ell_1$, called the {\it Rademacher sketch} in \cite{verbin2012rademacher}, and the {\it $M$-sketch} in \cite{DBLP:conf/soda/ClarksonW15}.  However the approximation guarantees in these works are significantly worse than what we achieve, and we improve them by (1) changing the actual sketch to ``randomized boundaries'' and (2) changing the analysis of the sketch to track the behavior of the $\ell_1$-leverage score vector, which captures the entire subspace, and tracking it via a new mix of expected and high probability events. 

We now explain these ideas in more detail. To motivate our sketch, we first explain the pitfalls of previous sketches. 

\paragraph{Cauchy Sketches \cite{DBLP:conf/stoc/SohlerW11,DBLP:conf/soda/WangW19}.}
The previous best $O(1)$ distortion $\ell_1$ oblivious subspace embedding of \cite{DBLP:conf/soda/WangW19}, which achieved a sketching dimension of $2^{2^{O(d)}}$, was based on analyzing a sketch $\bfS$ of i.i.d.\  Cauchy random variables. The only analyses of such random variables we are aware of, in the context of subspace embeddings, works by truncating the random variables so that they have a finite expectation, and then analyzing the behavior of the random variable $\|\bfS\bfy\|_1$, for an input vector $\bfy$ in expectation. It turns out that the expectation of this random variable can be much larger than the value it takes with constant probability, as it is very heavy-tailed. Namely, the expected value of $\|\bfS\bfy\|_1$ after truncation is $\Theta(\log n) \|\bfy\|_1$, which makes it unsuitable for the sketching dimension that we seek. 

\paragraph{Rademacher and $M$ Sketches \cite{verbin2012rademacher,DBLP:conf/soda/ClarksonW15}.}
Using techniques from the data stream literature, the {\it Rademacher sketch} of \cite{verbin2012rademacher} and the {\it $M$-sketch} of \cite{DBLP:conf/soda/ClarksonW15} achieve an $O(1)$-approximation for a single vector by subsampling rows of $\bfy$ with probability $p$ and rescaling by $1/p$ at $O(\log n)$ scales $p = 1,1/2,1/4,1/8,\dots,1/n$. This approach allows us to more finely track the random variables in our sketch, and serves as the starting point of our sketch. Note that for a single scale $p$ and a single coordinate $\bfy_i$, the expected contribution of the subsampled and rescaled coordinate is
\[
    \frac1p \cdot p\cdot \abs{\bfy_i} = \abs{\bfy_i}.
\]
Then in expectation, the $O(\log n)$ subsampling levels give a $O(\log n)$ factor approximation, which is the same as that of a Cauchy sketch. However, due to the geometrically decreasing sampling rates, we are able to argue that with good probability the coordinate does not survive more than $O(1)$ levels. Thus we effectively ``beat the expectation'', showing that the random variable is much less than what its expectation would predict, with good probability. We illustrate this with an example.

Suppose the first $\sqrt{n}$ coordinates of $\bfy$ equal $\frac{1}{\sqrt{n}}$, and remaining $n-\sqrt{n}$ coordinates equal $\frac{1}{n}$. Then $\|\bfy\|_1 = 2(1-o(1))$. If we subsample at geometric rates $1, 1/2, 1/4,\ldots, 1/n$ and use $t = O(1)$ hash buckets in \textsf{CountSketch} in each scale, then for rates larger than $1/\sqrt{n}$, the random signs in each \textsf{CountSketch} bucket cancel out and the absolute value of the bucket concentrates to its Euclidean norm, which is much smaller than its $1$-norm. At the rate $p = 1/\sqrt{n}$, we expect a single survivor from the first $\sqrt{n}$ coordinates of $\bfy$. We call this the {\it ideal rate} for the first $\sqrt{n}$ coordinates of $\bfy$. There are also about $\sqrt{n}$ survivors from the remaining $n-\sqrt{n}$ coordinates of $\bfy$ at this ideal rate, but these $\sqrt{n}$ survivors concentrate to their Euclidean norm in each \textsf{CountSketch} bucket, which will be about $1/n^{3/4}$, and negligible compared to the value $1/\sqrt{n}$. This lone survivor will be scaled up by $\sqrt{n}$, giving a contribution of $1$ to the overall $1$-norm. Similarly, at the subsampling rate of $1/n$, we expect one surviving coordinate of $\bfy$, it is scaled up by $n$, and it gives an additional contribution of about $1$ to the overall $1$-norm. Overall, this gives a good approximation to $\|\bfy\|_1$, which is $2(1-o(1))$. 

While the above gives a good approximation, the expected value of the $1$-norm of $\bfS\bfy$ is a much larger $\Theta(\log n)$. Indeed, consider subsampling rates $1/(2 \sqrt{n}), 1/(4 \sqrt{n}), 1/(8 \sqrt{n}), \ldots$. For each of these, the single survivor of the first $\sqrt{n}$ coordinates of $y$ has probability $1/2, 1/4, 1/8, \ldots,$ of surviving each successive level. If it survives, it is scaled up by $2, 4, 8, \ldots,$ giving an overall expectation of $\Theta(\log n)$. Thus, the expectation is not what we should be looking at, but rather we should be conditioning on the event that no items among the first $\sqrt{n}$ surviving beyond the rate $1/\sqrt{n}$. 

\paragraph{Ingredient 1: Aggressive Subsampling and Randomized Boundaries.}
So far, this is standard. Indeed, the {\it Rademacher sketch} in \cite{verbin2012rademacher} and the {\it M-Sketch} in \cite{DBLP:conf/soda/ClarksonW15} achieve an $O(1)$-approximation for a single vector and argue this way. But these works cannot achieve $(1+\eps)$-approximation with good probability, since it is already problematic if the single survivor of the first $\sqrt{n}$ coordinates of $\bfy$ survives one additional subsampling rate beyond its ideal rate, and this happens with constant probability. This motivates our first fix: instead of subsampling at rates $1/2^i$, for $i = 0, 1, 2, \ldots, O(\log n)$, we subsample at a much more aggressive $\exp(\eps^{-1}\polylog(n))^i$ for $i = 0,1,2,\dots,O(\log n)$, and furthermore, randomly shift these subsampling rates as well. 

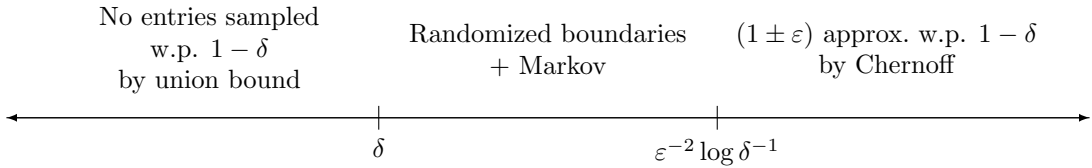
\begin{figure}[ht]
    \centering
    \begin{tikzpicture}[scale=0.9]
        \draw[latex-] (-8,0) -- (8,0);
        \draw[-latex] (-8,0) -- (8,0);
        \draw[shift={(-2.5,0)},color=black] (0pt,5pt) -- (0pt,-5pt) node[below] 
        {$\delta$};
        \draw[shift={(2.5,0)},color=black] (0pt,5pt) -- (0pt,-5pt) node[below] 
        {$\eps^{-2}\log\delta^{-1}$};
        \node[align=center] at (-5,1) {No entries sampled \\ w.p. $1-\delta$ \\ by union bound};
        \node[align=center] at (-0,1) {Randomized boundaries \\ + Markov};
        \node[align=center] at (5,1) {$(1\pm\eps)$ approx.\ w.p. $1-\delta$ \\by Chernoff};
    \end{tikzpicture}
    \caption{Casework on $pm$}
    \label{fig:casework-pm}
\end{figure}

To see why this is a good idea, consider a level set of weight $w$, which is the multiset of coordinates of $\bfy$ with absolute value $\Theta(w)$ (think of $w$ as $[2^j, 2^{j+1})$ for some $j$) that is subsampled at rate $p$ and rescaled by $1/p$. Let the size of the level set be $m$. We case on $pm$ (see Figure \ref{fig:casework-pm}). If $pm \geq \varepsilon^{-2}\log \frac1\delta$, then Chernoff bounds imply that this concentrates to the expected mass of $pm$ with probability at least $1 - \delta$. On the other hand, if $pm < \delta$, then by a union bound, there is a $\delta$ probability that any of the $m$ elements in the level set are sampled. By taking $\delta = 1/\log^2 n$, we see that by a union bound over the at most $\log n$ level sets and $\log n$ sampling rates $p$, any level set with size $m$ and subsampling rate $p$ with $pm \notin [\delta, \varepsilon^{-2}\log \delta^{-1}]$ either samples $(1\pm\varepsilon)$ of the expected mass, or doesn't sample the level at all, with constant probability. Then, for these levels, our earlier analyses involving \textsf{CountSketch} apply and in fact give us a $(1\pm O(\eps))$ approximation. However, for the level sets and the sampling rates with $pm\in [\delta, \varepsilon^{-2}\log \delta^{-1}]$, we cannot make any meaningful statements about these levels with high accuracy and probability. To remedy this situation, we randomize our choice of the sampling rates $p$ themselves and bound the contribution from these levels with a Markov expectation bound. To this end, we let $W = \varepsilon^{-2}\delta^{-1}\log\delta^{-1}$ be the size of this bad window, we let $B = \exp(\varepsilon^{-1}\log W)$ be our branching factor, and we choose our sampling rates to be $p_i = B^{-u}B^{-i}$ for a uniformly random $u\sim[0,1]$. Note then that the probability that a given sampling level $p_i$ falls in the window $p_im\in[\delta, \varepsilon^{-2}\log\delta^{-1}]$ is at most $\varepsilon$, since after taking logarithms, the bad window is an $\varepsilon$ fraction of the range of the uniformly random shift $u$. Now note that for each level set of size $m$ and weight $w$, there are only $O(1)$ sampling levels $p_i$ that have a nonzero probability such that $p_im\in[\delta, \varepsilon^{-2}\log\delta^{-1}]$, and these levels contribute an expected $\varepsilon \cdot p\cdot p^{-1}\cdot m\cdot w = \varepsilon mw$ amount of $\ell_1$ mass, so summing over all level sets, the \emph{expected} contribution from these bad sampling rates is at most an $\varepsilon$ fraction of the total $\ell_1$ mass $\norm{\bfy}_1$. 

This is an example of how subsampling gives us more flexibility than sketches using Cauchy random variables - even though the expectation is large, we can argue with arbitrarily large constant probability we obtain a $(1+O(\eps))$-approximation by separating the analysis into an expectation for some levels and a union bound for others. One also needs to argue that no vector has its $1$-norm shrink by more than a $(1-\eps)$-factor, which is simpler and similar to previous work \cite{DBLP:conf/soda/ClarksonW15}. Here the idea is that for every level set of coordinates of $\bfy$, by Chernoff bounds, there are enough survivors in a level set at its ideal rate and that the noise in \textsf{CountSketch} buckets will be small. Our analysis so far is novel, and we note that prior analyses of subsampling \cite{verbin2012rademacher,DBLP:conf/soda/ClarksonW15} could not obtain a $1+O(\eps)$-approximation even for a fixed vector. 

However, we are still in trouble - the above analysis gives a $(1+O(\eps))$-approximation, but only a constant probability of success due to the Markov bound applied to the bad sampling rates. We could more aggressively subsample, namely, at rate roughly $1/2^{2^{O(d)}}$ and with $2^{2^{O(d)}}$ buckets, and then we could make the failure probability $(\eps)^{O(d)}$ for a fixed vector, which is now small enough to union bound over an $\eps$-net of vectors in a $d$-dimensional subspace. This is enough to recover the same sketching dimension as the sketch in \cite{DBLP:conf/soda/WangW19}, which instead consisted of an $r \times n$ matrix of i.i.d.\  Cauchy random variables. There it was shown that with probability $1- O \left (\frac{\log \log r}{\log r} \right )$, for any fixed vector $\bfy$, $\|\bfS\bfy\|_1 = \Theta(1) \|\bfy\|_1$. The idea was then to take a union bound over $2^{O(d)}$ vectors in a net for the subspace, which constrains $\frac{\log \log r}{\log r} \leq 2^{-\Theta(d)}$, resulting in an $r = 2^{2^{O(d)}}$ overall dependence. With minor modifications, one can achieve $\|\bfS\bfA\bfx\|_1 = (1 \pm \eps) \|\bfA\bfx\|_1$ for all $\bfx$ by setting $r = 2^{2^{O(d/\epsilon^2)}}$. This is the best one can achieve for an arbitrary set of $2^{O(d)}$ vectors, as can be deduced from the lower bound in \cite{DBLP:conf/soda/WangW19}; see Section \ref{sec:prelim} for details. 

\paragraph{Ingredient 2: $\ell_1$ Leverage Scores.}
One might suspect that the above approach is optimal, since union bounding over $2^{O(d)}$ arbitrary points does give an optimal sketching dimension for subspace embeddings for the Euclidean norm. It turns out though that for the $1$-norm this is not the case, and one can do exponentially better by using the fact that these $2^{O(d)}$ points all live in the same $d$-dimensional subspace. Indeed, instead of making a net argument, our analysis proceeds through the {\it $\ell_1$-leverage score vector} (see Definition \ref{def:l1-leverage-scores}), which provides a nonuniform importance sampling distribution that is analogous to the standard leverage scores for $\ell_2$. 

With these $\ell_1$ leverage scores in hand, we proceed as discussed previously, choosing a uniformly random shift $u \in [0,1]$ and subsampling at rates $1/((\log n)^{\textrm{poly}((d/\eps) (i + u)})$ for $i = 0, 1, 2, \ldots, O(\log n)$, and also increasing our number of \textsf{CountSketch} buckets in each subsampling level to $(\log n)^{\poly(d/\eps)}$. Now we can show that the expected $\ell_1$-norm of the $\ell_1$ leverage score vector $\bflambda$ that survives an additional level is only $\eps \|\bflambda\|_1/d$. Noting that $\|\bflambda\|_1 = d$, this bound is $O(\eps)$ with constant probability. But the entries of $\bflambda$ uniformly bound the corresponding entries of any vector $\bfy$ in the subspace with $\|\bfy\|_1 = 1$, and thus we obtain that for all vectors in the subspace, the total expected $\ell_1$-contribution from level sets that are one subsampling rate beyond their ideal rate is $O(\eps)\norm{\bfy}_1$. Since the subsampling rate is $(\log n)^{-\poly(d/\epsilon)}$, the expected number of survivors two or more levels out is small enough to union bound over all net vectors. Finally, to remove the $\log n$ factor in our sketch, making it independent of the original dimension $n$, we can compose our embedding with the $2^{2^{O(d)}}$ $\ell_1$ oblivious subspace embedding of \cite{DBLP:conf/soda/WangW19}; we are able to adapt their $O(1)$-approximation to achieve a $(1+\eps)$-approximation with $2^{2^{O(d/\eps^2)}}$ dimensions, and consequently in our sketch, $\log n = 2^{O(d/\eps^2)}$. Our full discussion is in Section \ref{section:1+eps-l1-subspace-embedding}. 

\subsubsection{A Transition to Tensors}
One could hope to use our techniques for subspaces to obtain sketches for the sum of $q$-mode tensors, which could then be used for independence testing in a stream. Consider the simple example of a $2$-mode tensor, i.e., a $d \times d$ matrix $\bfP$. As described above, a streaming-amenable way of sketching this would be to find a sketch $\bfS: \mathbb{R}^{d^2} \rightarrow \mathbb{R}^{k^2}$ of the form $\bfS = \bfS^1 \otimes \bfS^2$, where $\bfS^1, \bfS^2$ are maps from $\mathbb{R}^d$ to $\mathbb{R}^k$. In this case, we have that $\bfS(\bfP) = \bfS^1 \cdot \bfP \cdot (\bfS^2)^\top$, where $\cdot$ denotes matrix multiplication. 

One aspect of our sketch above is that we can achieve a tradeoff: instead of looking at one subsampling rate beyond the ideal rate for a given level set of a vector, we can look at $1/\alpha$ rates for $\alpha \in (0,1)$. Then if we look at $\|\bfS\bfy\|_1$ for a column $\bfy$ of $\bfP$, its expected cost for these $1/\alpha$ rates is $O(1/\alpha)\|\bfS\bfy\|_1$. If we use roughly $(d\log d)^{\alpha}$ buckets in each \textsf{CountSketch}, together with subsampling rate roughly $(d \log d)^{-\alpha}$, then after $O(1/\alpha)$ rates beyond the ideal rate for a given level set of a vector, the probability the level set survives is at most $\left (\frac{1}{(d \log d)^{\alpha}} \right )^{O(1/\alpha)} \ll O \left (\frac{1}{d \log d} \right )$, which is so small that we can union bound over all columns of $\bfP$ and all level sets in each column. Consequently, we can condition on this event, and take an expectation over the $O(1/\alpha)$ rates nearest to the ideal rate of each level set in each column to obtain an overall $O(1/\alpha)$ approximation with roughly $(d \log d)^{\alpha}$ memory. One can also show that with constant probability, the $1$-norm does not decrease by more than a constant factor, and thus, with constant overall probability, $\Omega(\|\bfP\|_1) \leq \|\bfS^1 \bfP\|_1 = O(1/\alpha) \|\bfP\|_1$. Applying $\bfS^2$ to the matrix $\bfS^1\bfP$ we can conclude that with constant overall probability, $\Omega(\|\bfP\|_1) \leq \|\bfS^1\bfP\bfS^2\|_1 = O(1/\alpha^2) \|\bfP\|_1$. Our overall sketching dimension is $d^{2\alpha} \ll d$ if $\alpha \ll 1$. Thus, the memory we achieve is a significant improvement over the trivial $d^2$ bound, our sketch $\bfS = \bfS^1 \otimes \bfS^2$ is a tensor product, and we achieve an $O(1/\alpha^2)$-approximation. Ours is the first sketch to achieve a tradeoff, as the Rademacher sketch of \cite{verbin2012rademacher} does not apply in this case\footnote{The notion of the Rademacher dimension in \cite{verbin2012rademacher} is at least $\sqrt{d}$, and their sketch size is at least the Rademacher dimension to the $5$-th power.}. 

Unfortunately, if we want constant distortion, our single-mode sketch size $k$ will be $d^{2 \alpha}$, which means for constant $\alpha$, it is not strong enough to obtain a
polylogarithmic dependence on $d$. In fact, we show that for any $d \times d$ matrix $\bfP$, if you compute $\bfS\bfP$ for an oblivious sketch $\bfS$ with $t$ rows, the estimator $\|\bfS\bfP\|_1$ is at best an $O \left (\frac{\log d}{\log t} \right )$-approximation to $\|\bfP\|_1$. Indeed, one can show this already for the distribution in which with probability $1/2$, $\bfP \in \mathbb{R}^{d \times d}$ is an i.i.d.\  Cauchy matrix, and with probability $1/2$, $\bfP$ has its first $t$ columns being i.i.d.\  Cauchy random variables, scaled by $d/t$, and remaining columns equal to $0$. In both cases $\|\bfP\|_1 = \Theta(d^2 \log d),$ but in the first case $\|\bfS\bfP\|_1 = O(d \log t \|\bfS\|_1)$, while in the second case $\|\bfS\bfP\|_1 = \Omega(d \log d \|\bfS\|_1)$, both with constant probability. These algorithms and lower bounds are discussed in Section \ref{section:entrywise-embeddings}.

Fortunately, for independence testing, we only need to approximate the $1$-norm of a single tensor, and so our estimator can be a non-convex median-based estimator, which we now show how to utilize. 

\subsubsection{Independence Testing}
Our sketch $\bfS = \bfS^1 \otimes \bfS^2 \otimes \cdots \otimes \bfS^q$ is a tensor product of $q$ sketches, each itself being a sketch for estimating the $1$-norm of a $d$-dimensional vector with a $\log(1/\delta)$ dependence. We must choose the $\bfS^i$ carefully, and cannot take the $\bfS^i$ to be an arbitrary black box sketch for estimating the $1$-norm, even with a non-linear high probability estimator. As an illustration, suppose $q = 2$ and we have a $d \times d$ matrix $\bfP$ and we compute $\bfS^1 \bfP \bfS^2$, where $\bfS^1$ and $\bfS^2$ are i.i.d.\ Cauchy matrices with $r = O(\eps^{-1} \log d)$ small dimension with corresponding median of absolute values estimator, i.e., the sketch of \cite{DBLP:journals/jacm/Indyk06} above. Then, applying the estimator of $\bfS^2$ to each row of $\bfS^1 \bfP$, we would have that our overall estimate is $(1 \pm \eps)\norm{\bfS^1 \bfP}_1$ with probability $1-1/\poly(d)$. The issue is that, for constant $\eps$, if $\bfP = (1, 1, 1, \ldots, 1) \otimes (1, 0, 0, \ldots, 0)$, then $\|\bfS^1\bfP\|_1 = \Theta(d \log r)$ with large probability, while if $\bfP = \bfI_d$, the $d \times d$ identity matrix, then $\|\bfS^1\bfP\|_1 = \Theta(d \log d)$ with large probability. To see this, if $\bfP = (1, 1, 1, \ldots, 1) \otimes (1, 0, 0, \ldots, 0)$, note that the $i$-th row of $\bfS^1\bfP = d \cdot (C^i, 0, \ldots, 0)$, where $C^i$ is a standard Cauchy, and the $C^1, \ldots, C^r$ are independent. About a $\Theta(2^{-j})$ fraction of the $|C^i|$ will be $2^{j}$, and so with constant probability $\|\bfS^1\bfP\|_1 = \Theta(d \log r)$. On the other hand, if $\bfP = \bfI_d$, then $\bfS^1\bfP = \bfS^1$, which is an $r \times d$ matrix of i.i.d.\  Cauchy random variables, and the same reasoning shows with constant probability that $\|\bfS^1\bfP\|_1 = \Theta(d \log(rd))$, which is almost a $\log d$ factor larger than the other case. Thus, we cannot decode mode by mode with a generic high probability sketch for the $1$-norm. 

Perhaps surprisingly, we show that a different choice of $\bfS^i$, which is itself an existing sketch for estimating the $1$-norm of a $d$-dimensional vector with a $\log(1/\delta)$ dependence, {\it does work}. 
In more detail, the sketch of \cite{DBLP:conf/stoc/IndykW05} works by defining level sets of coordinates of $\bfx$ according to their magnitudes and subsamples the coordinates at different rates. For each level set, if it contributes a non-negligible fraction to $\norm{\bfx}_1$, there is a subsampling level for which (1) there are sufficiently many survivors from the level set in this subsampling level and (2) these survivors are so-called $\ell_2$-heavy hitters (see, e.g., \cite{CCF02}) among all the survivors in this subsampling level. Hence, recovering the heavy hitters at each subsampling rate allows us to estimate the contribution of each level set to $\norm{\bfx}_1$. Here a median is used when applying \textsf{CountSketch} to ensure that we succeed with high probability. This single mode sketch has been applied to $\ell_1$-estimation in various places \cite{earth_mover,DBLP:conf/nips/LevinSW18}. We refer to this as a \textit{\textsf{SubsamplingHeavyHitters} sketch} in the following discussion. 

Our overall sketch $\bfS = \bfS^1 \otimes \bfS^2 \otimes \cdots \otimes \bfS^q$, where each $\bfS^i$ is a \textsf{SubsamplingHeavyHitters} sketch.
Moreover, $\bfS = \bfS^1 \otimes \cdots \otimes \bfS^q$, and so given vectors $P^1, \ldots, P^q \in \mathbb{R}^d$ in a stream, one can maintain $\bfS^i P^i$ for $i = 1, \ldots, q$, as well as $\bfS P$ for any vector $P\in \R^{d^q}$. In particular, in the context of independence testing, the $P^i$ could be the empirical marginal distributions and $P$ the empirical joint distribution. 
We show that $\bfS$ can be used to estimate the $\ell_1$-norm of an underlying arbitrary vector $x \in \mathbb{R}^{d^q}$ (which will be taken to be $P-P^1 \otimes \cdots \otimes P^q$). We do this by viewing $\bfS^q$ as being applied to each row of a flattened $t^{q-1} \times d$ matrix, where $t$ is the common sketching dimension of the $\bfS^i$. This matrix is defined as follows. We flatten $x$ to a $d^{q-1} \times d$ matrix $X$. We then consider the ``partially sketched" $d^{q-1} \times d$ matrix, where the $i$-th column is $\bfS^1 \otimes \bfS^2 \otimes \cdots \otimes \bfS^{q-1}$ applied to the $i$-th column $X_{*,i}$ of $X$. This gives us a $t^{q-1} \times d$ matrix $Y$, and this is the matrix whose rows we apply $\bfS^q$ to. Now $\bfS^q$ is a \textsf{SubsamplingHeavyHitters} sketch, but instead of having a signed sum of single coordinates in each \textsf{CountSketch} bucket, we have a signed sum of columns of $Y$ in each bucket, which are themselves sketches of $d^{q-1}$-dimensional vectors, where the sketching matrix is itself a tensor product of smaller sketching matrices.

The problem is that $\bfS^q$ estimates the number of columns of a matrix in a level set (here the level sets are groups of columns with approximately the same $1$-norm) by hashing columns together and estimating the size of each level set, where columns are in the same level set if they have approximately the same $1$-norm. Fortunately, since $\bfS^1 \otimes \cdots \otimes \bfS^{q-1}$ is still a linear map, hashing the sketched columns (sketched by $\bfS^1 \otimes \cdots \otimes \bfS^{q-1}$) together is the same as taking the sketch (by $\bfS^1 \otimes \cdots \otimes \bfS^{q-1}$) of the hashed columns together. However, it is still unclear what the $1$-norm of the sketch of the hashed columns is. In fact, it cannot be concentrated with high probability by the above discussion. Fortunately, for each bucket in a \textsf{CountSketch} associated with a subsampling rate in $\bfS^q$, we can use our knowledge of $\bfS^1 \otimes \cdots \otimes \bfS^{q-1}$ to {\it recursively estimate} the $1$-norm inside of that bucket. This recursive estimation involves applying $\bfS^{q-1}$ to the rows of a $t^{q-2} \times d$ matrix $Z$, computing recursive estimates, and so on. Finally, we use these recursive estimates to estimate the level sets of columns of the matrix $X$, and ultimately build and output the estimator provided by $\bfS^q$. 

The main issue we still face is how to handle the blowup in approximation ratio and error probability in each recursive step. In each $\bfS^i$ we would like to randomize boundaries to avoid overcounting when estimating level set sizes in the estimator. However, the approximation error grows as we decode more modes. The most natural approach, if the error after decoding the $i$-th mode is $(1 + \eta)$, is to randomize boundaries so that the probability is $O(\eta)$ of landing near a boundary, and consequently not being included in the estimator, when decoding $\bfS^{i+1}$. However, this blows up the approximation to $(1+\eta)^2$. Unfolding the recursion, we get a $(1+\eps)^{\tilde{O}(2^q)}$ overall approximation. Setting our initial $\eps$ to $\epsilon/2^{\tilde{O}(q)}$, we can make the overall approximation $1+\eps$.
This yields a $2^{O(q)}$ factor in the sketching dimension on each mode and thus a $2^{O(q^2)}$ factor in the sketching dimension in the overall tensor product.

\begin{figure}[t]
\begin{minipage}{0.4\textwidth}
    \centering
    \includegraphics[width=\textwidth]{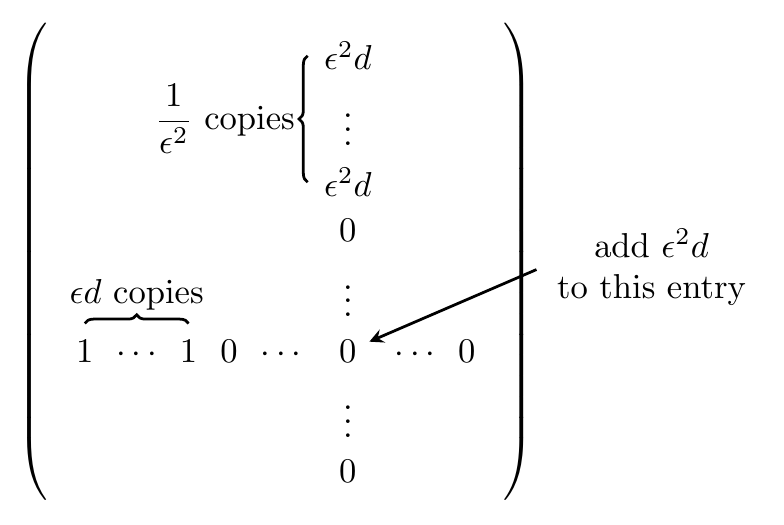}
    \caption{Hard instance for the attempted improvement when $q=2$. The algorithm first hashes rows into buckets.}
    \label{fig:hard_2d}
\end{minipage}
\hfill
\begin{minipage}{0.55\textwidth}
    \centering
    \includegraphics[width=0.8\textwidth]{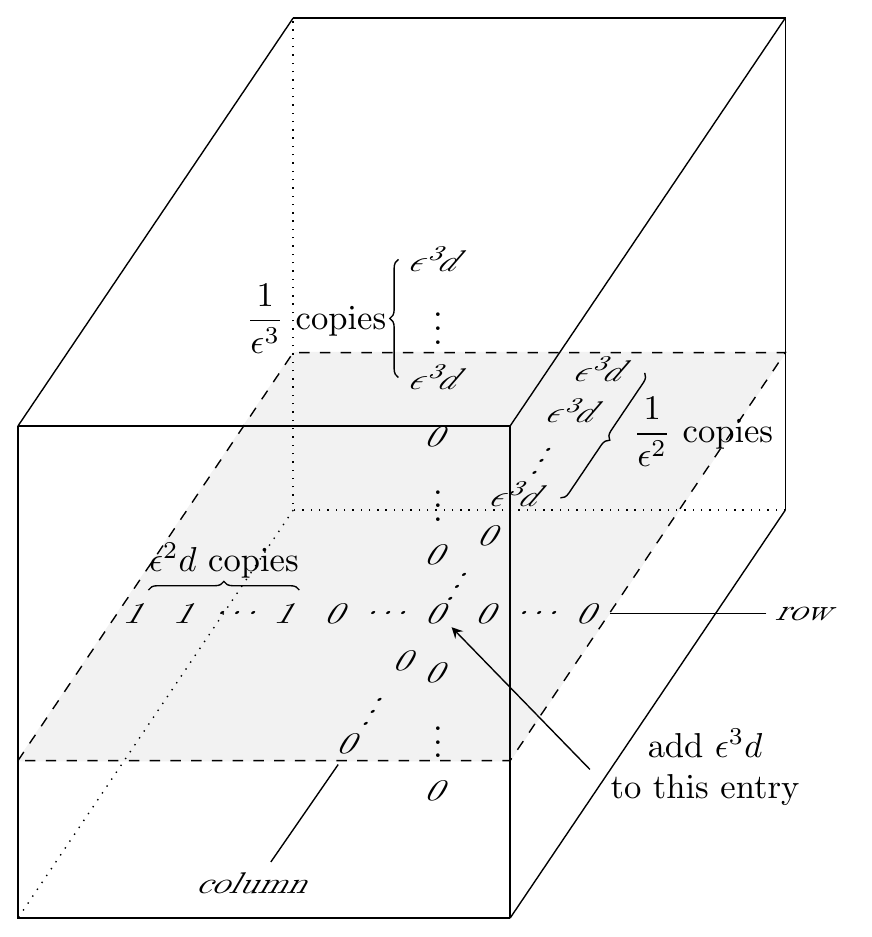}
    \caption{Hard instance for the attempted improvement when $q=3$. The algorithm first hashes horizontal slices into buckets (parallel to the shaded slice), then the sub-algorithm for each bucket (which contains a linear combination of horizontal slices) hashes rows into buckets.}
    \label{fig:hard_3d}
\end{minipage}
\end{figure}

It seems difficult to improve the $2^{O(q^2)}$ bound. To improve this bound, we need to make the error smaller than $(1+\eta)^2$ in the $(i+1)$-st mode after obtaining a multiplicative error of $(1+\eta)$ factor in the $i$-th mode. Imagine that we flatten the first $(i+1)$-modes as a $d\times d^i$ matrix. It is tempting to view one's estimate in the $(i+1)$-st mode as providing an approximation to the $1$-norm of the {\it vector of estimates} of rows produced by $\bfS^i$. 
Since we hash the rows (the first $i$ modes) into buckets as in a \textsf{CountSketch} structure, a heavy row in a bucket is perturbed by some small noise and we need to claim that this small perturbation only incurs a small error in the estimate of the row by $\bfS^i$. 
An issue arises that a small perturbation in $1$-norm on the first $i$ modes may appear larger for a heavy row on the first $(i-1)$ modes, or, equivalently, the first $(i+1)$ modes can tolerate a constant-factor smaller perturbation under $\bfS^{i+1}$ than the first $i$ modes under $\bfS^i$, and thus $\bfS^{i+1}$ needs to use a constant-factor more number of buckets than $\bfS^i$ to reduce the error in each bucket, resulting in the same $2^{O(q^2)}$ factor in the overall sketching dimension. 
To see that the shrinking perturbation on higher modes is indeed possible, see Figures~\ref{fig:hard_2d} and~\ref{fig:hard_3d} for example. In Figure~\ref{fig:hard_2d}, the $d\times d$ matrix has norm $\Theta(d)$ and exactly one $\eps$-heavy row. To recover the heavy row, the rows are hashed into $1/\eps^2$ buckets and the heavy row is combined with exactly one value of $\eps^2 d$ at the specified entry in some bucket. Note that the entry is an $\eps$-heavy hitter in the combined row. Adding a value of $\eps^2 d$ to the specified entry is only an $\eps^2$-factor perturbation to the overall matrix but an $\eps$-factor perturbation to the bucket and a constant-factor perturbation to that entry. Similarly, in Figure~\ref{fig:hard_3d}, adding a value of $\eps^3 d$ to the specified entry is only an $\eps^3$-factor perturbation to the overall $d\times d\times d$ cube but an $\eps^2$-factor perturbation to the only $\eps$-heavy slice (shaded) and an $\eps$-factor perturbation to the only $\eps$-heavy row on that slice.

It is important to note that the work of Braverman and Ostrovsky~\cite{braverman2010measuring} also applies $\ell_1$-sketches in the context of tensor products. However, the subroutines used in \cite{braverman2010measuring} define both level sets and subsampling rates in power of $1+\eps$, and $\eps$ can be shown to become polynomially smaller in each recursive step, and  
consequently, when iterating this process for a general tensor of order $q$, at the base level it requires a $(1+\eps^{2^q})$-approximation to the relevant quantities, resulting in a doubly exponential $\Omega(1/\eps^{2^q})$ amount of memory. Removing the $1/\eps^{2^q}$ term from their space complexity does not appear to be straightforward~\cite{B20}. In contrast, our algorithm is a more direct analogue of {\sc TensorSketch} \cite{DBLP:journals/toct/Pagh13,DBLP:conf/kdd/PhamP13,DBLP:conf/nips/AvronNW14,ahle2020oblivious} but for the $1$-norm, and admits a simpler analysis, leading to a singly exponential sketching dimension as well as a singly exponential memory bound in a data stream.

Given the simplicity and modular components of our algorithm, we can extend it to any distance measure with a (1) small so-called {\it Rademacher dimension}, a (2) black box sketching algorithm, and (3) an approximate triangle inequality. 

\subsubsection{Polynomial-Sized Subspace Embeddings}
In order to obtain even better oblivious subspace embeddings into $\ell_1$, we consider the case when the input matrix $\bfA$ itself has i.i.d.\  entries. This models settings in statistics with random design matrices for regression, and our results can be viewed from the lens of average-case complexity. The important property from the distribution on each entry of $\bfA$ is its tail. 

We give the intuition for our improved upper bounds when $\bfA$ is a matrix of i.i.d.\  Cauchy random variables. We obtain an $O((\log n)/\log d)$-approximation by simply using a \textsf{CountSketch} matrix $\bfS$ with $\poly(d)$ rows. When $n$ is at most a polynomial in $d$, this gives an $O(1)$-approximation, bypassing the $\Omega(d/\log^2 d)$ lower bound of \cite{DBLP:conf/soda/WangW19} for arbitrary input matrices $\bfA$. The idea is that by looking at the rows of $\bfA$ containing the largest $\poly(d)$ entries in $\bfA$ - call this submatrix of rows $\bfA_{top}$ - then we can show $\|\bfA_{top}\bfx\|_1 \geq n (\log d) \|\bfx\|_1$ for all $\bfx$. On the other hand, one can show that for any $x$, $\|\bfA\bfx\|_1 \leq \|\bfA_{top} \bfx\|_1 + (n \log n) \|\bfx\|_1,$ by concentration bounds applied to the rows not containing a large entry. Finally, we use that (1) \textsf{CountSketch} does not increase the $1$-norm of any vector it is applied to, and (2) it perfectly hashes the rows in $\bfA_{top}$. Putting these statements together gives us an $O((\log n)/\log d)$-approximation. 

We also give a number of lower bounds, showing that our algorithms for random $\bfA$ are also nearly optimal in their sketching dimension. These results are presented in Section \ref{section:subspace_embeddings_random}. 

\subsection{Additional Related Work}
Our focus is on linear oblivious maps. Besides being a fundamental mathematical object, such maps are essential for the data stream and distributed models above, allowing for very fast update time under updates. There are other, non-oblivious embeddings for $n$ points in $\ell_1$, achieving $O(n/\epsilon^2)$ dimensions \cite{Newman2010OnCD,schechtman1987more,talagrand1990embedding}, which is nearly optimal \cite{charikar2002dimension,brinkman2005impossibility,DBLP:conf/focs/AndoniCNN11}. See also \cite{cohen2015lp,talagrand1990embedding} for non-oblivious subspace embeddings based on Lewis weights. 

For oblivious subspace embeddings, one can achieve $O(d \log d)$ distortion with a sketching dimension of $O(d \log d)$ using a matrix of Cauchy random variables \cite{DBLP:conf/stoc/SohlerW11}. This is a significantly larger distortion than the distortion we seek here. It does not contradict the lower bound of \cite{DBLP:conf/soda/WangW19} which grows roughly as $\Omega(d/\log^2 r)$, where $r$ is the sketching dimension. 

\input{preliminaries}

\input{subspace_embedding}

\input{entrywise_embedding}

\input{prod_test_ell1}
\input{prod_test_ell1_thm}
\input{prod_test_ell1_proof}

\input{subspace_embedding_random}

\section*{Acknowledgements}
We thank anonymous reviewers for their feedback, and T.\ Yasuda thanks Manuel Fernandez for useful discussions. D.\ Woodruff and T.\ Yasuda thank partial support from a Simons Investigator Award. Y.\ Li was supported in part by Singapore Ministry of Education (AcRF) Tier 2 grant MOE2018-T2-1-013.

\bibliographystyle{alpha}
\bibliography{citations}

\appendix

\input{preliminaries_appendix}
\input{subspace_embedding_appendix}
\input{entrywise_embedding_appendix}
\input{subspace_embedding_random_appendix}

\end{document}

%% file: stdinc.tex
\usepackage{xspace}
\usepackage{verbatim}
\usepackage[dvipsnames]{xcolor}
\usepackage{graphicx}
\usepackage{tabularx}

\usepackage{amsmath}
\usepackage{bbm}
\usepackage[showonlyrefs]{mathtools}
\usepackage{empheq}
\usepackage{amstext,amssymb,amsfonts}
\usepackage{fullpage}
\usepackage{bm}

\usepackage{enumitem}

\usepackage{todonotes}

\newcounter{mynotes}
\DeclarePairedDelimiter\floor{\lfloor}{\rfloor}

\usepackage{textcomp,setspace}

\usepackage{nameref}
\usepackage[pagebackref,colorlinks,linkcolor=blue,filecolor = blue,citecolor = ForestGreen, urlcolor = blue, hyperfootnotes=false]{hyperref}
\usepackage{color}
\usepackage[capitalize, noabbrev, nameinlink]{cleveref}

\usepackage[linesnumbered,lined,algo2e,ruled]{algorithm2e}
\SetKwComment{cmt}{$\triangleright$\ }{}
\SetCommentSty{}

\usepackage{amsthm}
\usepackage{thmtools}
\usepackage{thm-restate}

\declaretheorem[within=section]{theorem}

\declaretheorem[sibling=theorem]{lemma}

\crefname{claim}{Claim}{Claims}
\newcounter{termcounter}
\renewcommand{\thetermcounter}{\Alph{termcounter}}
\crefname{term}{term}{terms}
\creflabelformat{term}{\textup{(#2#1#3)}}

\makeatletter
\def\term{\@ifnextchar[\term@optarg\term@noarg}
\def\term@optarg[#1]#2{%
  \textup{(#1)}%
  \def\@currentlabel{#1}%
  \def\cref@currentlabel{[][2147483647][]#1}%
  \cref@label[term]{#2}}
\def\term@noarg#1{%
  \refstepcounter{termcounter}%
  \textup{(\thetermcounter)}%
  \cref@label[term]{#1}}
\makeatother

\newcommand{\wtilde}[1]{\widetilde{#1}}

\newcommand{\ignore}[1]{}

\DeclareMathOperator{\poly}{poly}

\DeclareMathOperator{\polylog}{polylog}
\DeclareMathOperator*{\median}{median}

\newcommand{\abs}[1]{\lvert#1\rvert}
\newcommand{\Abs}[1]{\left\lvert#1\right\rvert}

\newcommand{\norm}[1]{\lVert#1\rVert}
\newcommand{\Norm}[1]{\left\lVert#1\right\rVert}

\definecolor{DSred}{rgb}{1,0,0}

\renewcommand{\leq}{\leqslant}
\renewcommand{\geq}{\geqslant}

\renewcommand{\epsilon}{\varepsilon}
\newcommand{\eps}{\epsilon}

\newcommand{\R}{\mathbb{R}}

\newcommand{\Z}{\mathbb{Z}}

\newcommand{\cA}{\mathcal A}
\newcommand{\cB}{\mathcal B}

\newcommand{\cD}{\mathcal D}

\newcommand{\cJ}{\mathcal J}

\newcommand{\Esymb}{{\bf E}}

\newcommand{\Psymb}{{\bf Pr}}

\DeclareMathOperator*{\E}{\Esymb}
\DeclareMathOperator*{\Var}{{\bf Var}}
\DeclareMathOperator*{\ProbOp}{\Psymb}

\renewcommand{\Pr}{\ProbOp}



%% file: preliminaries.tex
\section{Preliminaries}\label{sec:prelim}

\subsection{Subspace embeddings}

We record some results in the literature that are standard ingredients in the construction and analysis of subspace embeddings. We first recall the \textsf{CountSketch} construction.
\begin{dfn}[\textsf{CountSketch} \cite{CCF02}]\label{def:countsketch}
\textsf{CountSketch} is a distribution over $r\times n$ matrices that samples a random matrix $\bfS$ as follows. 
\begin{itemize}
    \item Let $H: [n]\to [r]$ be a random hash function, so that $H(i) = r'$ for $r'\in [r]$ with probability $1/r$.
    \item For each $i\in[n]$, let $\Lambda_i \sim \{\pm1\}$. 
    \item $\bfS$ is an $r\times n$ matrix taking values in $\{-1,0,1\}$ such that $\bfS_{H(i), i} = \Lambda_i$ for each $i\in[n]$ and $0$s everywhere else.
\end{itemize}
\end{dfn}
\begin{rem}
The \textsf{CountSketch} construction originated in the data stream literature \cite{CCF02} and has been successfully applied to problems in numerical linear algebra in works such as \cite{DBLP:conf/stoc/DasguptaKS10,clarkson2017low,DBLP:conf/soda/ClarksonW15}. 
\end{rem}

The next lemma is useful for net arguments:
\begin{lem}[Net argument]\label{lem:net-argument-ingredients}
Let $\bfA\in\mathbb R^{n\times d}$ and let $\mathcal S\coloneqq \braces*{\bfA\bfx : \bfx\in\mathbb R^d, \norm{\bfA\bfx} = 1}$. Let $\eps\in(0,1/2)$. 
\begin{itemize}
    \item There exists an $\ell_1$ $\eps$-net $\mathcal N$ of size at most $(3/\eps)^d = \exp(d\log(3/\eps))$ over $\mathcal S$, that is, for every $\bfy\in\mathcal S$ there exists a $\bfy'\in\mathcal N$ such that $\norm{\bfy-\bfy'}_1\leq \eps$ \cite{bourgain1989approximation}.
    \item Let $\bfy\in\mathcal S$. Then, $\bfy = \sum_{i=0}^\infty \bfy^{(i)}$ where each nonzero $\bfy^{(i)}$ has $\bfy^{(i)}/\norm{\bfy^{(i)}}_1\in \mathcal N$ and $\norm{\bfy^{(i)}}_1\leq \eps^i$ \cite[implicit in Theorem 3.5]{DBLP:conf/soda/WangW19}. 
\end{itemize}
\end{lem}

The next lemma uses a standard balls and bins martingale argument (e.g., \cite{lee2016lecture}) to show concentration for uniquely hashed items. This is used in \cite{DBLP:conf/soda/ClarksonW15} to analyze the $M$-sketch. 
\begin{lem}[Concentration for unique hashing]\label{lem:concentration-unique-hashing}
Let $h: [n]\to [r]$ be a random hash function. Let $S\subseteq T\subseteq [n]$, $p\in(0,1]$, and $\eps\in (0,1)$ with $\eps r\geq p\abs{T}$. Consider the process that samples each element $i\in [n]$ with probability $p$ and hashes it to a bucket in $[r]$ if it was sampled. Let $X$ be the number of elements $i\in S$ that are sampled and hashed to a bucket containing no other member of $T$. Then,
\[
    \Pr\parens*{X \geq (1-\eps)^2 p\abs{S}}\leq 2\exp\parens*{-\frac{\eps^2}{12}p\abs{S}}.
\]
\end{lem}
\begin{proof}
The proof is deferred to Appendix \ref{sec:appendix:prelim}. 
\end{proof}

\begin{thm}[Improvement of Theorem 3.5, \cite{DBLP:conf/soda/WangW19}]\label{thm:1+eps-dense-cauchy}
Let $\eps\in(0,1)$, $r = \exp(\exp(O(d\eps^{-2}\log\eps^{-1} + \eps^{-2}\log\delta^{-1})))$, and let $\bfS$ be an $r\times n$ matrix of i.i.d.\ Cauchys. Then for any $\bfA\in\mathbb R^{n\times d}$,
\[
    \Pr\braces*{(1-\eps)\norm{\bfA\bfx}_1 \leq \norm{\bfS\bfA\bfx}_1\leq (1+\eps)\norm{\bfA\bfx}_1}\geq 1-\delta. 
\]
\end{thm}
\begin{proof}
The proof is deferred to Appendix \ref{sec:appendix:prelim}. 
\end{proof}

\begin{rem}\label{rem:sketch-arbitrary-vectors}
Note that the above dense sketch preserves an arbitrary fixed vector with probability at least $1-\delta$ using a sketching dimension of $2^{1/\delta}$. Thus, for preserving the $1$-norm of $n$ arbitrary vectors, it suffices to set $\delta = O(1/n)$. On the other hand, the lower bound argument of \cite[Theorem 1.1]{DBLP:conf/soda/WangW19} proves a distortion lower bound for sketching matrices that preserve even just the columns of the input matrix $\bfA$. Thus, we can place our $n$ vectors along the columns of a matrix, so that for constant distortion, a sketch needs $r$ dimensions, for
\[
    \frac{n}{\log^2 r} = O(1)\implies r = \Omega(2^{\sqrt{n}}).
\]
\end{rem}

%% file: subspace_embedding.tex
\section{Singly Exponential \texorpdfstring{$(1+\eps)$}{(1+eps)} \texorpdfstring{$\ell_1$}{l1} Subspace Embeddings}\label{section:1+eps-l1-subspace-embedding}

In this section, we prove the following theorem:
\begin{thm}\label{thm:1+eps-sparse-embedding}
Let $\eps\in(0,1)$ and $\delta\in (0,1)$. Then there exists a sparse oblivious $\ell_1$ subspace embedding $\bfS$ into $r$ dimensions with
\[
    r = \poly(d,\eps^{-1},\delta^{-1},\log n)^{d/\delta\eps}
\]
such that for any $\bfA\in\mathbb R^{n\times d}$,
\[
    \Pr_\bfS\braces*{\forall \bfx\in\mathbb R^d, (1-\eps)\norm{\bfA\bfx}_1 \leq \norm{\bfS\bfA\bfx}_1\leq (1+\eps)\norm{\bfA\bfx}_1, }\geq 1-\delta. 
\]
\end{thm}

Our main contribution towards proving this result is in showing the ``no dilation'' direction $\norm{\bfS\bfA\bfx}_1\leq (1+\eps)\norm{\bfA\bfx}_1$. The ``no contraction'' direction of $\norm{\bfS\bfA\bfx}_1\geq (1-\eps)\norm{\bfA\bfx}_1$ direction was already known in \cite{DBLP:conf/soda/ClarksonW15}, and we defer the details of handling our minor changes to Appendix \ref{section:appendix:no-contraction}. 

If we settle for dense embeddings, then we are able to get an improved sketching dimension that is independent of $n$ by first applying the dense $\ell_1$ subspace embedding of Theorem \ref{thm:1+eps-dense-cauchy}, which maps our subspace down to a subspace of dimension independent of $n$ and preserves $1$-norms up to a $(1+\eps)$ factor distortion:
\begin{cor}\label{cor:dense-embedding}
Let $\eps\in(0,1)$ and $\delta\in (0,1)$. Then there exists an oblivious $\ell_1$ subspace embedding $\bfS$ into $r$ dimensions with
\[
    r = \exp\parens*{\tilde O(d/\delta\eps))}
\]
such that for any $\bfA\in\mathbb R^{n\times d}$,
\[
    \Pr_\bfS\braces*{\forall \bfx\in\mathbb R^d, (1-\eps)\norm{\bfA\bfx}_1 \leq \norm{\bfS\bfA\bfx}_1\leq (1+\eps)\norm{\bfA\bfx}_1}\geq 1-\delta. 
\]
\end{cor}
\begin{proof}
By applying the sketch of Theorem \ref{thm:1+eps-dense-cauchy} first, we can take $\log\log n\leq d/\delta\eps^2$. By repeating again, we can take $\log n\leq \poly(d/\delta\eps)$. Then, the bounds for Theorem \ref{thm:1+eps-sparse-embedding} yield the desired result. 
\end{proof}

By a known lower bound in Theorem 1.1 of \cite{DBLP:conf/soda/WangW19}, which shows that embedding $d$ vectors requires $2^{\sqrt d}$ dimensions, the dependence on $d$ is optimal up to polynomial factors in the exponent. Note that an embedding which preserves the norm of a single vector with probability at least $1 - \delta$ for $\delta = (10d)^{-1}$ also preserves the norms of $d$ vectors with constant probability, so there is also a lower bound of $2^{\sqrt{1/\delta}}$, making the singly exponential dependence on $\delta$ tight up to polynomial factors in the exponent as well. 

\subsection{The embedding}

We first collect constants that will be used. The constants can all be written in terms of the dimensions $n$ and $d$ of the input matrix, the accuracy parameter $\eps$, and the failure rate $\delta$. 
\begin{dfn}[Useful constants]\label{def:useful-constants}
\begin{align*}
    h_{\max} &\coloneqq \log_2(n/\eps) &&= O(\log(n/\eps)) && \text{Sampling levels} \\
    q_{\max} &\coloneqq \log_2(ndh_{\max}/\delta\eps) &&= O(\log(nd/\eps)) && \text{Weight classes}  \\
    \alpha &\coloneqq 2\exp(d\log(3/\eps))q_{\max}/\delta &&= O\parens*{\frac{\exp(d/\eps)\log(nd/\eps)}{\delta}} && \text{Net union bounding} \\
    m_{\mathrm{crowd}} &\coloneqq 300\frac{d^{11}}{\eps^9\delta^4}\log^5(n) &&= O(\poly(d,\eps^{-1},\delta^{-1},\log n)) && \text{Overcrowding hash buckets} \\
    B &\coloneqq (m_{\mathrm{crowd}}h_{\max}q_{\max}/\delta)^{d/\delta\eps} &&= O\parens*{\poly(d,\eps^{-1},\delta^{-1},\log n)^{d/\delta\eps}} && \text{Branching factor}  \\
    N_0 &\coloneqq \frac{12B^u q_{\max}}{\eps^3}\log\alpha &&  && \text{Hash buckets in $0$th level} \\
    N &\coloneqq B\frac{8d^2\log d}{\eps^6}q_{\max}\parens*{\log\alpha}\parens*{\log\frac{B}{\eps}} &&= O(B\log n\poly(d, \eps^{-1})) && \text{Hash buckets per level}
\end{align*}
\end{dfn}

As described in the introduction, the construction of our embedding is essentially a variant of $M$-sketch \cite{DBLP:conf/soda/ClarksonW15}. However, instead of using fixed subsampling rates of $1/\poly(d)$, we use randomized subsampling rates which drop off geometrically by factors of $B = O\parens*{\poly(d,\eps^{-1},\delta^{-1},\log n)^{d/\delta\eps}}$.

\begin{dfn}
Let $u \sim[0,1]$ and define subsampling rates
\[
    p_h \coloneqq B^{-(u+h-1)}
\]
for each $h\in[h_{\max}]$.
\end{dfn}

\begin{dfn}\label{def:scaled-sampling}
For each $i\in[n]$ and $h\in[h_{\max}]$, let
\[
    b_{i,h} \coloneqq \begin{cases}
        1 & \text{w.p. $p_h$} \\
        0 & \text{w.p. $1-p_h$}
    \end{cases},
\]
and let $m_h\coloneqq \sum_{i\in[n]}b_{i,h}$.
\end{dfn}

\begin{dfn}
For each $i\in[n]$, let $\Lambda_i\sim\{\pm1\}$. Let $H_0: [n]\to [N_0]$ and $H_h: [m_h]\to [N]$ for each $h\in[h_{\max}]$ be a random hash functions. 
\end{dfn}

\begin{dfn}[Random-boundary $M$-sketch]\label{def:M-sketch}
Let $\bfC^{(0)}$ be an $N_0\times n$ \textsf{CountSketch} matrix (Definition \ref{def:countsketch}) with random signs $\Lambda_i$ and hash function $H_0$, that is, 
\[
    \bfC^{(0)}_{H_0(i), i} \coloneqq \Lambda_i
\]
for every $i\in[n]$ and $0$s everywhere else. For each $h\in[h_{\max}]$, let $\bfS^{(h)}$ be the $m_h\times n$ scaled sampling matrix given by
\[
    \bfe_j^\top\bfS^{(h)}\bfe_i = \begin{cases}
        \frac1{p_h} & j = \sum_{i'\in[i]} b_{i',h} \\
        0 & \text{otherwise}
    \end{cases}.
\]
For each $h\in[h_{\max}]$, let $\bfC^{(h)}$ be an $N\times m_h$ \textsf{CountSketch} matrix with random signs $\Lambda_i$ and hash function $H_h$, that is,
\[
    \bfC^{(h)}_{H_h(i), i} \coloneqq \Lambda_i
\]
for each $\bfy_i$ that was sampled, i.e., $b_{i,h} = 1$, and $0$s everywhere else. Then, our random-boundary $M$-sketch is given by
\[
    \bfS \coloneqq \begin{pmatrix}
        \bfC^{(0)} \\
        \bfC^{(1)}\bfS^{(1)} \\
        \bfC^{(2)}\bfS^{(2)} \\
        \vdots \\
        \bfC^{(h_{\max})}\bfS^{(h_{\max})}
    \end{pmatrix}.
\]
\end{dfn}

\subsection{Notation for analysis}

We first recall some notation from the analysis of $M$-sketch in \cite{DBLP:conf/soda/ClarksonW15}, as well as a few other definitions. 

\begin{dfn}
Let $\bfy\in\mathbb R^n$ be a unit $\ell_1$ vector and let $q\in\mathbb N$. We define weight classes
\[
    W_q(\bfy)\coloneqq \braces*{\bfy_i : 2^{-q}\leq \abs{\bfy_i}\leq 2^{1-q}}.
\]
When the $\bfy$ is clear from context, we simply write $W_q$ for brevity. For a set $Q\subseteq\mathbb N$, we write
\[
    W_Q\coloneqq \bigcup_{q\in Q} W_q.
\]
We also write $\abs{W_q}$ for the size of $W_q$ and
\[
    \norm{W_q}_1 \coloneqq \sum_{y\in W_q}\abs{y}.
\]
\end{dfn}

\begin{dfn}
For $h\in[h_{\max}]$ and $k\in[N]$, we write $L_{h,k}$ for the multiset of elements that get sampled and hashed to the $k$th bucket in the $h$th level.
\end{dfn}

We briefly digress to recall $\ell_1$ leverage score vectors. 

\begin{dfn}[$\ell_1$ well-conditioned basis (Definition 2, \cite{DBLP:conf/soda/ClarksonDMMMW13}, see also \cite{DBLP:journals/siamcomp/DasguptaDHKM09})]\label{def:well-conditioned-basis}
    A basis $\bfU$ for the range of an $n\times d$ matrix $\bfA$ is $(\alpha,\beta)$-conditioned if $\norm{U}_1\leq \alpha$ and for all $\bfx\in\mathbb R^d$, $\norm{\bfx}_\infty\leq \beta\norm{\bfU\bfx}_1$. We say that $\bfU$ is well-conditioned if $\alpha$ and $\beta$ are low-degree polynomials $d$, independent of $n$. It is known that an Auerbach basis for $\bfA$ is $(d,1)$-conditioned. 
\end{dfn}

\begin{dfn}[$\ell_1$ leverage scores (Definition 3, \cite{DBLP:conf/soda/ClarksonDMMMW13})]\label{def:l1-leverage-scores}
Given a $(d,1)$-conditioned basis $\bfU$ (see Definition \ref{def:well-conditioned-basis}) for the column space of $\bfA\in\mathbb R^{n\times d}$, define the vector $\bflambda\in\mathbb R^n$ of normalized $\ell_1$ leverage scores of $\bfA$ to be
\[
    \bflambda_i \coloneqq \frac{\norm{\bfe_i^\top\bfU}_1}{d}.
\]
\end{dfn}
\begin{rem}
As noted in \cite{DBLP:conf/soda/ClarksonDMMMW13}, the $\ell_1$ leverage scores are not defined uniquely. We also note that for convenience of notation, our normalization of the leverage scores is off by a factor of $d$ from standard definitions in the literature.
\end{rem}

In our analysis, we consider weight classes $W_q(\bflambda)$ of the $\ell_1$ leverage score vector $\bflambda$. For each weight class $W_q$, we set
\[
    h_q\coloneqq \floor*{\log_B\abs{W_q}}
\]
so that $B^{h_q}\leq \abs{W_q} < B^{h_q+1}$.

\begin{dfn}
For a pair $(h,q)\in[h_{\max}]\times\mathbb N$ and an interval $I$, define the event
\[
    \mathcal E_{h,q}(I) \coloneqq \braces*{p_h\abs{W_q(\bflambda)}\in I}
\]
in which sampling the weight class $W_q(\bflambda)$ at rate $p_h$ has an expected number of items in the window $I$. 
\end{dfn}

\begin{dfn}[Scaled leverage score samples]\label{def:scaled-leverage-score-samples}
For each $(h,q)\in[h_{\max}]\times[q_{\max}]$ and an interval $I$, define the random variables
\begin{align*}
    \mathcal S_{h,q} &\coloneqq \frac1{p_h}\sum_{\bflambda_i\in W_q} b_{i,h}\bflambda_i \\
    \mathcal T_{h,q}(I) &\coloneqq \frac1{p_h}\sum_{\bflambda_i\in W_q} b_{i,h}\bflambda_i\mathbbm{1}(\mathcal E_{h,q}(I))
\end{align*}
\end{dfn}

In the following sections, we give upper bounds on the mass of the sketch depending on the weight class of the leverage scores that we look at. We have the following intervals:
\begin{itemize}
    \item \textbf{Dead levels} $p_h\abs{W_q(\bflambda)}\in [0, \delta/h_{\max}q_{\max})$: In this interval, we sample none of these entries with high probability. 
    \item \textbf{Badly concentrated levels} $p_h\abs{W_q(\bflambda)}\in [\delta/h_{\max}q_{\max}, m_{\mathrm{crowd}})$: The expected mass of leverage scores coming from this level is at most $O(\eps/d)$, which means that with constant probability, the mass contribution for all subspace vectors is $O(\eps)$. 
    \item \textbf{Golidlocks levels} $p_h\abs{W_q(\bflambda)}\in [m_{\mathrm{crowd}}, Bm_{\mathrm{crowd}})$: In this interval, we can show that the mass contribution is at most a $(1+\eps)$ factor more than the expected mass coming from this interval with high probability. This level is counted only once, since the size of the interval is less than a $B$ factor. 
    \item \textbf{Oversampled levels} $p_h\abs{W_q(\bflambda)}\in [Bm_{\mathrm{crowd}}, \infty)$: In this interval, we sample so many of these entries that it overcrowds the \textsf{CountSketch} hash buckets, which makes the mass contribution at most an $\eps$ fraction due to the random sign cancellations. 
\end{itemize}

\subsection{Bounding badly concentrated levels}
For levels with expected mass in the interval $[1/\alpha, \log\alpha]$ at subsampling rate $p_h$,  we cannot hope to reason about the mass contribution of this level with high enough probability to union bound over a net, since we need expectation at most $1/\alpha$ for the level to get completely missed by the sampling, and we need at least $\log \alpha$ in order to get concentration. However, we show that because of our randomization of subsampling rates, the leverage score mass contribution from these rows is only an $O(\eps/d)$ fraction of the total mass of the leverage scores in expectation, which means it is only an $O(\eps)$ fraction of the total mass of any subspace vector with constant probability by a combination of properties of leverage scores and a Markov bound.

\begin{lem}[Randomized sampling rates]\label{lem:choose-good-sampling-prob}
Let $\delta'\in(0,1)$, let $0<a<1$ and $b>1$, and let $B'\coloneqq (b/a)^{1/\delta'}$. Let $u\sim[0,1]$, $p = B'^{-u}$, and let $t\in\mathbb R$. Then,
\[
    \Pr(pt \in [a,b]) \leq \begin{cases}
        0 & \text{if $t\geq b$ or $B't\leq a$} \\
        \delta & \text{otherwise}.
    \end{cases}
\]
\end{lem}
\begin{proof}
The first bound follows from the fact that $t = B'^0t\leq pt\leq B'^1t = B't$. For the second bound, we calculate
\[
    \Pr(pt \in [a,b]) = \Pr\parens*{u \in \log_{B'} t + [-\log_{B'} b, -\log_{B'} a]}\leq \log_{B'}(b/a) = \delta'\frac{\log(b/a)}{\log(b/a)} = \delta'.\qedhere
\]
\end{proof}

\begin{cor}\label{cor:bad-event-probability}
For every $h\in[h_{\max}]$ and $q\in[q_{\max}]$,
\[
    \Pr\parens*{\mathcal E_i\parens*{[\delta/h_{\max}q_{\max}, m_{\mathrm{crowd}})}} = \Pr_u\parens*{p_h\abs{W_q}\in[\delta/h_{\max}q_{\max}, m_{\mathrm{crowd}})}\leq \begin{cases}
        0 & \text{if $h\notin\{h_q,h_q+1\}$} \\
        \frac{\delta\eps}{d} & \text{otherwise}
    \end{cases}.
\]
\end{cor}
\begin{proof}
Note that for $h\geq h_q + 2$,
\[
    B^{-h}\abs{W_q}\leq B^{-h+h_q+1}\leq B^{-1}\leq \frac\delta{h_{\max}q_{\max}}
\]
and for $h\leq h_q-1$,
\[
    B^{-h}\abs{W_q}\geq B^{-h+h_q}\geq B^{1}\geq m_{\mathrm{crowd}}
\]
so for $h\notin\{h_q, h_q+1\}$, 
\[
    \Pr_u\parens*{p_h\abs{W_q}\in[\delta/h_{\max}q_{\max}, m_{\mathrm{crowd}})} = \Pr_u\parens*{B^{-u}\parens*{B^{-h}\abs{W_q}}\in[\delta/h_{\max}q_{\max}, m_{\mathrm{crowd}})} = 0. 
\]
On the other hand, for $h\in\{h_q, h_q+1\}$,
\[
    \Pr_u\parens*{p_h\abs{W_q}\in[\delta/h_{\max}q_{\max}, m_{\mathrm{crowd}})} \leq \frac{\delta\eps}{d}
\]
by Lemma \ref{lem:choose-good-sampling-prob}. 
\end{proof}

Note that by Corollary \ref{cor:bad-event-probability}, $\mathcal E_{h,q}([\delta/h_{\max}q_{\max}, m_{\mathrm{crowd}}))$ has nonzero probability for only $h\in\{h_q,h_q+1\}$.

\begin{lem}[Expected mass of bad leverage scores]\label{lem:expected-badly-concentrated-leverage-scores}
\[
    \E_{u,b}\parens*{\sum_{q\in[q_{\max}]}\sum_{h\in[h_{\max}]} \mathcal T_{h,q}([\delta/h_{\max}q_{\max}, m_{\mathrm{crowd}}))}\leq \frac{4\delta\eps}{d}.
\]
\end{lem}
\begin{proof}
Let $I\coloneqq [\delta/h_{\max}q_{\max}, m_{\mathrm{crowd}})$. Then,
\begin{align*}
    \E_{u,b}\parens*{\sum_{q\in[q_{\max}]}\sum_{h\in[h_{\max}]} \mathcal T_{h,q}(I)} &= \E_{u,b}\parens*{\sum_{q\in[q_{\max}]}\sum_{h\in\braces*{h_q,h_q+1}} \mathcal T_{h,q}(I)} \\ 
    &= \sum_{q\in[q_{\max}]}\sum_{h\in\braces*{h_q,h_q+1}}\E_{u,b}\mathcal T_{h,q}(I) \\
    &= \sum_{q\in[q_{\max}]}\sum_{h\in\braces*{h_q,h_q+1}} \sum_{\bflambda_i\in W_q}\E_{u,b} \bracks*{\frac1{p_h}b_{i,h}\bflambda_i\mathbbm{1}(\mathcal E_{h,q}(I))} \\
    &\leq \sum_{q\in[q_{\max}]} \sum_{h\in\braces*{h_q,h_q+1}} \sum_{\bflambda_i\in W_q} 2^{1-q}\E_u(\mathbbm{1}(\mathcal E_{h,q}(I))) \\
    &\leq \sum_{q\in[q_{\max}]} \sum_{h\in\braces*{h_q,h_q+1}} \sum_{\bflambda_i\in W_q} 2^{1-q}\frac{\delta\eps}{d} \\
    &= \sum_{q\in[q_{\max}]}2^{2-q}\abs{W_q} \frac{\delta\eps}{d} \\
    &\leq \frac{4\delta\eps}{d}.\qedhere
\end{align*}
\end{proof}

\begin{lem}\label{lem:leverage-score-bound}
For any $\bfx\in\mathbb R^d$ and $i\in[n]$, we have that
\[
    \frac{\abs{\bfe_i^\top\bfA\bfx}}{\norm{\bfA\bfx}_1}\leq d\bflambda_i
\]
\end{lem}
\begin{proof}
Let $\bfy\in\mathbb R^d$ be such that $\bfA\bfx = \bfU\bfy$. Then,
\[
    \frac{\abs{\bfe_i^\top\bfA\bfx}}{\norm{\bfA\bfx}_1} = \frac{\abs{\bfe_i^\top\bfU\bfy}}{\norm{\bfU\bfy}_1}\leq \frac{\norm{\bfe_i^\top\bfU}_1\norm{\bfy}_\infty}{\norm{\bfy}_\infty} = \norm{\bfe_i^\top\bfU}_1 = d\bflambda_i
\]
where the first inequality follows from properties of well-conditioned bases. 
\end{proof}

\subsection{Bounding Goldilocks levels}
In this level, the expected sampled mass is large enough to get concentration, but not large enough to overflow the hash buckets of the \textsf{CountSketch}. In this level, we show that the mass contribution is at most a $(1+\eps)$ factor more than the expected mass. The main idea for getting concentration here is using the bounds on the leverage scores to bound outliers, and using a Bernstein bound to get concentration on the rest of the entries with a good bound on the variance. 

\begin{dfn}
Define $\bfA^{(q)}$ to be the $n\times d$ matrix formed by taking the rows of $\bfA$ that correspond to leverage scores belonging to weight class $W_q(\bflambda)$, and $0$s everywhere else.
\end{dfn}

\begin{lem}\label{lem:goldilocks-bound}
Let $(h,q)\in[h_{\max}]\times[q_{\max}]$ with $p_h\abs{W_q(\bflambda)}\geq 3d^2\eps^{-4}\log\alpha$ and let $\bfx\in\mathbb R^d$ with $\norm{\bfx}_1 =1$. Then with probability at least $1 - 2/\alpha$, we have that
\[
    \norm{\bfS^{(h)}\bfA^{(q)}\bfx}_1 = \sum_{\bflambda_i\in W_q(\bflambda)} \frac{\abs{\bfe_i^\top\bfA\bfx}}{p_h}b_{i,h}\leq (1+\eps)\norm{\bfA^{(q)}\bfx}_1 + 4\eps \norm{W_q}_1\norm{\bfA\bfx}_1.
\]
\end{lem}
\begin{proof}
The average absolute value of an entry of $\bfA^{(q)}\bfx$ is $\mu_q\coloneqq \norm{\bfA^{(q)}\bfx}_1/\abs{W_q(\bflambda)}$. Then by averaging, there is at most an $\eps/d$ fraction of rows with absolute value greater than $d\mu_q/\eps$. Now for each $\bflambda_i\in W_q(\bflambda)$, define the event
\[
    \mathcal F_i\coloneqq \braces*{\abs{\bfe_i^\top\bfA\bfx}\geq \frac{d\mu_q}{\eps}}
\]
and the sample
\[
    X = \sum_{\bflambda_i\in W_q} b_{i,h}\mathbbm{1}(\mathcal F_i).
\]
Note that
\[
    \E X = \frac{\eps p_h\abs{W_q}}{d} \geq \frac{d\log \alpha}{\eps^4}\geq 3\log\alpha
\]
so by the Chernoff bound,
\[
    \Pr\parens*{X \geq 2\E X}\leq \exp\parens*{-\frac{\E X}{3}}\leq \frac1\alpha.
\]
Conditioned on the complement event, the mass contribution from rows $i$ for which $\mathcal F_i$ happens is at most
\[
    \sum_{\bflambda_i\in W_q} \frac{\abs{\bfe_i^\top\bfA\bfx}}{p_h}b_{i,h}\mathbbm{1}(\mathcal F_i)\leq 2\frac{\eps p_h\abs{W_q}}{d}\frac{\abs{\bfe_i^\top\bfA\bfx}}{p_h}\leq 2\eps\frac{\abs{W_q}}{d}d\bflambda_i \norm{\bfA\bfx}_1\leq 4\eps 2^{-q}\abs{W_q}\norm{\bfA\bfx}_1\leq 4\eps \norm{W_q}_1\norm{\bfA\bfx}_1
\]
where the second to last inequality follows from Lemma \ref{lem:leverage-score-bound}. 

We now consider the sample
\[
    Y = \sum_{\bflambda_i\in W_q} Y_i
\]
where
\[
    Y_i \coloneqq \frac{\abs{\bfe_i^\top\bfA\bfx}}{p_h}b_{i,h}\mathbbm{1}(\neg\mathcal F_i).
\]
Note that
\begin{align*}
    \E Y &\leq \norm{\bfA^{(q)}\bfx}_1 = \abs{W_q}\mu_q \\
    Y_i &\leq \frac1{p_h}\frac{d\mu_q}{\eps} \\
    \Var(Y_i) &\leq p_h\parens*{\frac1{p_h}\frac{d\mu_q}{\eps}}^2 = \frac1{p_h}\parens*{\frac{d\mu_q}{\eps}}^2
\end{align*}
Then by Bernstein's inequality,
\begin{align*}
    \Pr\parens*{Y - \E Y \geq \eps \abs{W_q}\mu_q} &\leq \exp\parens*{-\frac12\frac{(\eps\abs{W_q}\mu_q)^2}{\abs{W_q}(d\mu_q/\eps)^2/p_h + (\eps\abs{W_q}\mu_q)(d\mu_q/\eps p_h)/3}} \\
    &= \exp\parens*{-\frac12\frac{p_h\abs{W_q}\eps^2}{(d/\eps)^2 + d/3}} \leq \exp\parens*{-\frac{p_h\abs{W_q}}{3d^2\eps^{-4}}} \leq \frac1\alpha.
\end{align*}
We conclude by combining the two bounds.
\end{proof}

\subsection{Bounding oversampled levels}
When we expect to sample a large enough number of entries per hash bucket from a level, these entries cancel each other out due to the random signs. These levels fall under this criterion. 

\begin{lem}\label{lem:overcrowd-countsketch}
Let $(h,q)\in[h_{\max}]\times[q_{\max}]$ with $p_h\abs{W_q(\bflambda)}\geq bN$ for $b = 12(\frac{d h_{\max}}{\eps})^2 \log(Nh_{\max}q_{\max}/\delta)$. Then with probability at least $1 - 4\delta/h_{\max}q_{\max}$,
\[
    \norm{\bfC^{(h)}\bfS^{(h)}\bfA^{(q)}\bfx}_1\leq \frac{\eps}{h_{\max}} \norm{W_q(\bflambda)}_1\norm{\bfA\bfx}_1.
\]
Similarly, if $\abs{W_q(\bflambda)}\geq bN_0$ , then with probability at least $1-4\delta$,
\[
    \norm{\bfC^{(0)}\bfA^{(q)}\bfx}_1\leq \frac{\eps}{h_{\max}} \norm{W_q(\bflambda)}_1\norm{\bfA\bfx}_1.
\]
\end{lem}
\begin{proof}
We just show the first bound since the second is nearly identical. Note that by Lemma \ref{lem:leverage-score-bound}, $\abs{\bfe_i^\top\bfA^{(q)}\bfx}/\norm{\bfA\bfx}_1\leq d\bflambda_i\leq d 2^{1-q}$ for all $\bflambda_i\in W_q(\bflambda)$. 

By Chernoff's bound, the probability that a bucket $L$ in level $h$ gets $X = (1\pm 1/2)p_h\abs{W_q}/N$ elements from $W_q$ is at least
\[
    \Pr\parens*{\abs{X - \frac{p_h\abs{W_q}}{N}}\geq \frac12 \frac{p_h\abs{W_q}}{N}} \leq 2\exp\parens*{-\frac{(1/2)^2p_h\abs{W_q}}{3}} = 2\exp\parens*{-\frac{p_h\abs{W_q}}{12}}\leq 2\delta.
\]
We condition on this event. Then by Hoeffding's bound, the inner product of $m$ elements $\{a_i\}_{i=1}^m$ in the interval $[d2^{-q}\norm{\bfA\bfx}_1, d2^{1-q}\norm{\bfA\bfx}_1]$ with random signs $\eps_i$ concentrates around its mean as
\begin{align*}
    \Pr\parens*{\sum_{i=1}^m \eps_i a_i > d2^{1-q}\norm{\bfA\bfx}_1\sqrt{m}\sqrt{\log(Nh_{\max}q_{\max}/\delta)}} &\leq \exp\parens*{-\frac{(d2^{1-q}\norm{\bfA\bfx}_1\sqrt{m}\sqrt{\log(Nh_{\max}q_{\max}/\delta)})^2}{2d^2 2^{2-2q}\norm{\bfA\bfx}_1^2m}} \\
    &\leq \frac\delta{Nh_{\max}q_{\max}}.
\end{align*}
Then by a union bound over $N$ buckets, with probability at least $1-2\delta/h_{\max}q_{\max}$, we have for every bucket $L$ at this level that
\begin{align*}
    \Abs{\frac1{p_h}\sum_{\bfy_i\in L}\Lambda_i b_{i,h} \bfy_i} &\leq d2^{1-q}\norm{\bfA\bfx}_1\sqrt{X}\sqrt{\log(N/\delta)} \\
    &\leq \frac1{p_h}d2^{1-q}\norm{\bfA\bfx}_1\sqrt{\frac32 \frac{p_h\abs{W_q}}{N}}\sqrt{\log(N/\delta)} \\
    &\leq \frac1{p_h}\frac{\eps}{dh_{\max}} d2^{-q} \norm{\bfA\bfx}_1\frac{p_h \abs{W_q}}{N} \\
    &\leq \frac{\eps}{h_{\max}} \frac{\norm{W_q}_1}{N}\norm{\bfA\bfx}_1
\end{align*}
which gives the desired bound upon summing over the $N$ buckets. The overall success probability is at least $1 - 4\delta/h_{\max}q_{\max}$.
\end{proof}

\subsection{Net argument}
In this section, we collect the bounds obtained in previous sections and conclude with a net argument.

\begin{lem}\label{lem:bounding-all-but-goldilocks-and-0}
With probability at least $1 - 6\delta$, we have for all $\bfx\in\mathbb R^d$ that
\[
    \sum_{h\in[h_{\max}]}\sum_{q\in[q_{\max}]}\norm{\bfC^{(h)}\bfS^{(h)}\bfA^{(q)}\bfx}_1\mathbbm{1}\parens*{\mathcal F_{h,q}}\leq 5\eps\norm{\bfA\bfx}_1
\]
where
\[
    \mathcal F_{h,q} = \braces*{p_h\abs{W_q}\in [0, m_{\mathrm{crowd}})\cup [Bm_{\mathrm{crowd}}, \infty)}
\]
\end{lem}
\begin{proof}
We case on $p_h\abs{W_q}$ by intervals $[0, \delta/h_{\max}q_{\max})$, $[\delta/h_{\max}q_{\max}, m_{\mathrm{crowd}})$, and $[Bm_{\mathrm{crowd}}, \infty)$. 
\begin{itemize}
    \item \textbf{Dead levels}: First consider the $h$ for which $p_h\abs{W_q} < \delta/h_{\max}q_{\max}$. In this case, the probability that we sample any row corresponding to some $\bflambda_i\in W_q$ is at most $p_h\abs{W_q} < \delta/h_{\max}q_{\max}$ by a union bound. Then by a further union bound over all $(h, q)\in [h_{\max}]\times[q_{\max}]$, this category of levels contributes no mass with probability at least $1-\delta$. 
    
    \item \textbf{Badly concentrated levels}: Consider the subsampling levels with $p_h\abs{W_q}\in [\delta/h_{\max}q_{\max}, m_{\mathrm{crowd}})$. By Lemma \ref{lem:expected-badly-concentrated-leverage-scores}, the total expected leverage score mass contribution from all such pairs $(h, q)\in [h_{\max}]\times[q_{\max}]$ is at most $4\delta\eps/d$. Then by Markov's inequality, with probability at least $1-\delta$, the total expected leverage score mass is at most $4\eps/d$. Conditioned on this event, we have that
    \begin{align*}
        &\sum_{h\in[h_{\max}]}\sum_{q\in[q_{\max}]}\frac1{p_h}\sum_{\bflambda_i\in W_q} b_{i,h}\abs{\bfe_i^\top\bfA\bfx}\mathbbm{1}\parens*{\mathcal E_{h,q}([\delta/h_{\max}q_{\max}, m_{\mathrm{crowd}})} \\
        \leq\; &d\norm{\bfA\bfx}_1\sum_{h\in[h_{\max}]}\sum_{q\in[q_{\max}]}\frac1{p_h}\sum_{\bflambda_i\in W_q} b_{i,h}\bflambda_i\mathbbm{1}\parens*{\mathcal E_{h,q}([\delta/h_{\max}q_{\max}, m_{\mathrm{crowd}})} && \text{Lemma \ref{lem:leverage-score-bound}} \\
        \leq\; &d\frac{4\eps}{d}\norm{\bfA\bfx}_1 = 4\eps\norm{\bfA\bfx}_1 && \text{Lemma \ref{lem:expected-badly-concentrated-leverage-scores}}
    \end{align*}
    
    \item \textbf{Oversampled levels}: Consider the subsampling levels with $p_h\abs{W_q}\in [Bm_{\mathrm{crowd}},\infty)$. Note that $Bm_{\mathrm{crowd}}\geq bN$ is large enough to apply Lemma \ref{lem:overcrowd-countsketch}. By union bounding and summing over $h$ and $q$ for the result of the lemma, we have that
    \[
        \sum_{h\in[h_{\max}]}\sum_{q\in[q_{\max}]}\norm{\bfC^{(h)}\bfS^{(h)}\bfA^{(q)}\bfx}_1\mathbbm{1}\parens*{\mathcal E_{h,q}([Bm_{\mathrm{crowd}}, \infty)}\leq \sum_{h\in[h_{\max}]}\sum_{q\in[q_{\max}]}\frac{\eps}{h_{\max}}\norm{W_q}_1\norm{\bfA\bfx}_1\leq \eps\norm{\bfA\bfx}_1
    \]
    with probability at least $1 - 4\delta$. 
\end{itemize}
We thus conclude by a union bound over the above three events. 
\end{proof}

\begin{lem}[Tiny weight classes]\label{lem:tiny-weight-classes}
Let $q > q_{\max}$. Then with probability at least $1 - \delta$, it holds for all $\bfx\in\mathbb R^d$ that
\[
    \sum_{h\in[h_{\max}]}\sum_{q > q_{\max}}\norm{\bfS^{(h)}\bfA^{(q)}\bfx}_1\leq \eps\norm{\bfA\bfx}_1.
\]
\end{lem}
\begin{proof}
For the weight classes $q > q_{\max}$, the total leverage score mass contribution is bounded by
\[
    \sum_{q > q_{\max}}\norm{W_q(\bflambda)}_1\leq \sum_{q > q_{\max}}2^{1-q}\abs{W_q}\leq \frac{\delta\eps}{dnh_{\max}}\sum_{q > q_{\max}}\abs{W_q}\leq \frac{\delta\eps}{dh_{\max}}.
\]
Then in expectation, the sum of the scaled leverage score samples (Definition \ref{def:scaled-leverage-score-samples}) is bounded by
\begin{align*}
    \E\parens*{\sum_{h\in[h_{\max}]}\sum_{q > q_{\max}} \mathcal S_{h,q}} &= \sum_{h\in[h_{\max}]}\sum_{q > q_{\max}} \sum_{\bflambda_i\in W_q}\E\parens*{\frac1{p_h} b_{i,h}\bflambda_i} \\
    &= \sum_{h\in[h_{\max}]}\sum_{q > q_{\max}} \norm{W_q(\bflambda)}_1 \\
    &\leq \sum_{h\in[h_{\max}]}\frac{\delta\eps}{dh_{\max}} \\
    &= \frac{\delta\eps}{d}.
\end{align*}
Then with probability at least $1-\delta$, the above sum is at most $\eps/d$. We condition on this event. Then, for all $\bfx\in\mathbb R^d$,
\begin{align*}
    \sum_{h\in[h_{\max}]}\sum_{q > q_{\max}}\norm{\bfS^{(h)}\bfA^{(q)}\bfx}_1 &= \sum_{h\in[h_{\max}]}\sum_{q > q_{\max}} \sum_{\bflambda_i\in W_q}\frac1{p_h} b_{i,h}(\bfe_i^\top\bfA\bfx) \\
    &\leq d\norm{\bfA\bfx}_1 \sum_{h\in[h_{\max}]}\sum_{q > q_{\max}} \sum_{\bflambda_i\in W_q}\frac1{p_h} b_{i,h}\bflambda_i && \text{Lemma \ref{lem:leverage-score-bound}} \\
    &\leq d\norm{\bfA\bfx}_1\frac{\eps}{d} = \eps\norm{\bfA\bfx}_1
\end{align*}
as desired.
\end{proof}

\begin{lem}\label{lem:1+eps-approx-high-prob}
There is an event with probability $1 - 11\delta$ such that conditioned on this event, for every $\bfx\in\mathbb R^d$,
\[
    \Pr\parens*{\norm{\bfS\bfA\bfx}_1\leq (1+8\eps)\norm{\bfA\bfx}_1}\geq 1 - \frac{2q_{\max}}\alpha.
\]
\end{lem}
\begin{proof}
By Lemma \ref{lem:tiny-weight-classes}, the contribution from weight classes $q > q_{\max}$ is at most $\eps\norm{\bfA\bfx}_1$ with probability at least $1-\delta$. We let this event be $\mathcal E_1$ and restrict our attention to $q\leq q_{\max}$. 

For each $q\in [q_{\max}]$, we bound the mass contribution of rows corresponding to $W_q(\bflambda)$ at each subsampling level $\{0\}\cup [h_{\max}]$. Note that by Lemma \ref{lem:bounding-all-but-goldilocks-and-0}, there is an event $\mathcal E_2$ with probability at least $1 - 6\delta$ such that all levels $h,q$ except for those such that $h = 0$ or $p_h\abs{W_q}\in [m_{\mathrm{crowd}}, Bm_{\mathrm{crowd}})$ are bounded by at most $5\eps\norm{\bfA\bfx}_1$, so it remains to bound these levels. These are the $0$th level of subsampling (i.e., no subsampling) and the Goldilocks levels. 

Note that there exists at most one Goldilocks level $h\in[h_{\max}]$ such that $p_h\abs{W_q}\in [m_{\mathrm{crowd}}, Bm_{\mathrm{crowd}})$. In this case, Lemma \ref{lem:goldilocks-bound} applies since $m_{\mathrm{crowd}}\geq 3d^2\eps^{-1}\log\alpha$, and we have that
\[
    \norm{\bfS^{(h)}\bfA^{(q)}\bfx}_1 = \sum_{\bflambda_i\in W_q(\bflambda)} \frac{\abs{\bfe_i^\top\bfA\bfx}}{p_h}b_{i,h}\leq (1+\eps)\norm{\bfA^{(q)}\bfx}_1 + 4\eps \norm{W_q}_1\norm{\bfA\bfx}_1.
\]
with probability at least $1-2/\alpha$. If such a Goldilocks subsampling level $h$ exists, then note that
\[
    p_h\abs{W_q}\geq m_{\mathrm{crowd}}\implies \abs{W_q}\geq B^{u+h-1}m_{\mathrm{crowd}}\geq B^um_{\mathrm{crowd}}\geq bN_0. 
\]
Then by Lemma \ref{lem:overcrowd-countsketch}, the $0$th level of subsampling level contributes mass at most $(\eps/h_{\max})\norm{W_q(\bflambda)}_1$ with probability at least $1-4\delta/h_{\max}q_{\max}$. Thus by a union bound over all $q$s with a Goldilocks level and summing over these, the $0$th level contributes at most
\[
    \sum_{q\in q_{\max}}\norm{\bfC^{(0}\bfA^{(q)}\bfx}_1\mathbbm{1}(\exists h : p_h\abs{W_q}\in [m_{\mathrm{crowd}}, Bm_{\mathrm{crowd}}))\leq \sum_{q\in q_{\max}}\frac{\eps}{h_{\max}}\norm{W_q(\bflambda)}_1\norm{\bfA\bfx}_1\leq \eps\norm{\bfA\bfx}_1.
\]
Let this be event $\mathcal E_3$. On the other hand, for the Goldilocks level itself, there is a $1-2q_{\max}/\alpha$ probability that
\begin{align*}
    &\sum_{h\in[h_{\max}]}\sum_{q\in[q_{\max}]}\norm{\bfS^{(h)}\bfA^{(q)}\bfx}\mathbbm{1}\parens*{\mathcal E_{h, q}([m_{\mathrm{crowd}}, Bm_{\mathrm{crowd}}))} \\
    \leq\; &\sum_{q\in[q_{\max}]}(1+\eps)\norm{\bfA^{(q)}\bfx}_1\mathbbm{1}(\exists h : p_h\abs{W_q}\in [m_{\mathrm{crowd}}, Bm_{\mathrm{crowd}})) + 4\eps \norm{W_q}_1\norm{\bfA\bfx}_1 \\
    \leq\; &4\eps\norm{\bfA\bfx}_1 + \sum_{q\in[q_{\max}]}(1+\eps)\norm{\bfA^{(q)}\bfx}_1\mathbbm{1}(\exists h : p_h\abs{W_q}\in [m_{\mathrm{crowd}}, Bm_{\mathrm{crowd}}))
\end{align*}
by a union bound over the at most $q_{\max}$ weight classes. 

Otherwise, if a weight class $q$ has no Goldilocks level, then we have by the triangle inequality that
\[
    \norm{\bfC^{(0}\bfA^{(q)}\bfx}_1\leq \norm{\bfA^{(q)}\bfx}_1
\]
and thus we simply bound the contribution of the $0$th level by $\norm{\bfA^{(q)}\bfx}_1$. 

Note that $\mathcal E_1\cap \mathcal E_2\cap \mathcal E_3$ occurs with probability at least $1 - 11\delta$. Then conditioned on this event, every $\bfx\in\mathbb R^d$ has a $1 - 2q_{\max}/\alpha$ probability that
\begin{align*}
    \norm{\bfS\bfA\bfx}_1 &= \bracks*{\sum_{q > q_{\max}}\norm{\bfS\bfA^{(q)}\bfx}_1} +  \sum_{q\in[q_{\max}]}\bracks*{\norm{\bfC^{(0}\bfA^{(q)}\bfx}_1 +\sum_{h\in[h_{\max}]}\norm{\bfC^{(h)}\bfS^{(h)}(\bfA^{(q)}\bfx}_1} \\
    &\leq \eps\norm{\bfA\bfx}_1 + \underbrace{(1+\eps)\norm{\bfA\bfx}_1}_{\text{Goldilocks or $0$th level}} + \underbrace{\eps\norm{\bfA\bfx}_1}_{\text{$0$th level if Goldilocks level exists}} + \underbrace{5\eps\norm{\bfA\bfx}_1}_{\text{badly concentrated and oversampled levels}} \\
    &\leq (1+8\eps)\norm{\bfA\bfx}_1
\end{align*}
which is the desired bound. 
\end{proof}

We conclude by a standard net argument. 
\begin{thm}[No expansion]
With probability at least $1-11\delta$, we have that for all $\bfx\in\mathbb R^d$,
\[
    \norm{\bfS\bfA\bfx}_1\leq (1+11\eps)\norm{\bfA\bfx}_1.
\]
\end{thm}
\begin{proof}
By Lemma \ref{lem:1+eps-approx-high-prob}, there is an event with probability at least $1-10\delta$ such that conditioned on this event, for each $\bfx$, there is a $1-2/\alpha$ probability that
\begin{equation}\label{eqn:1+eps-approx}
    \norm{\bfS\bfA\bfx}_1\leq (1+8\eps)\norm{\bfA\bfx}_1.
\end{equation}

It is well-known (see e.g., \cite{bourgain1989approximation}), that there exists an $\eps$-net $\mathcal N$ of size at most $(3/\eps)^d = \exp(d\log(3/\eps))$ over the set $\braces*{\bfA\bfx : \bfx\in\mathbb R^d, \norm{\bfA\bfx} = 1}$. Then by a union bound over the net, Equation \ref{eqn:1+eps-approx} holds for every $\bfA\bfx\in \mathcal N$ with probability at least $1-\delta$. 

Finally, let $\bfx\in\mathbb R^d$ be arbitrary with $\norm{\bfA\bfx}_1 = 1$. It is shown in \cite[Theorem 3.5]{DBLP:conf/soda/WangW19} that $\bfA\bfx = \sum_{i=0}^\infty \bfy^{(i)}$ where each nonzero $\bfy^{(i)}$ has $\bfy^{(i)}/\norm{\bfy^{(i)}}_1\in \mathcal N$ and $\norm{\bfy^{(i)}}_1\leq \eps^i$. We then have that
\[
    \norm{\bfS\bfA\bfx}_1 = \Norm{\bfS\sum_{i=0}^\infty\bfy^{(i)}}_1\leq \sum_{i=0}^\infty\Norm{\bfS\bfy^{(i)}}_1\leq (1+8\eps)\sum_{i=0}^\infty\Norm{\bfy^{(i)}}_1\leq (1+8\eps)\sum_{i=0}^\infty\eps^i\leq 1+11\eps. 
\]
We conclude by homogeneity. 
\end{proof}

%% file: entrywise_embedding.tex
\section{Near Optimal Trade-offs for \texorpdfstring{$\ell_1$}{l1} Entrywise Embeddings}\label{section:entrywise-embeddings}

In this section, we obtain algorithmic trade-offs between sketching dimension and distortion for $\ell_1$ entrywise embeddings, and show that this is nearly tight for $d\times d$ matrices. 

\subsection{Algorithm}

Our algorithm is an $M$-sketch with subsampling rates $p_h = B^{-h}$, where $B = (\frac{d}{\delta}\log n)^\alpha$ for $\alpha\in(0,1)$, and \textsf{CountSketch} hashes into $\tilde \Theta(\frac{B}{\delta}\log n)$ buckets. By homogeneity, we assume that $\norm{\bfA}_1 = 1$ throughout this section. 

\begin{dfn}[Useful constants]
\begin{align*}
    B &\coloneqq \parens*{\frac{d}{\delta}\log n}^\alpha \\
    h_{\max} &\coloneqq \log_B n \\
    q_{\max} &\coloneqq \log_2\frac{nd h_{\max}}{\delta} \\
    p_h &\coloneqq B^{-h}, && h\in[h_{\max}]
\end{align*}
\end{dfn}

\begin{thm}\label{thm:l1-entrywise-embedding}
Let $\delta\in (0,1)$ and $\alpha\in(0,1)$. Then there exists a sparse oblivious $\ell_1$ entrywise embedding $\bfS$ into $k$ dimensions with
\[
    k = \parens*{\frac{d}{\delta}\log n}^\alpha \poly(\delta^{-1},\log n)
\]
such that for any $\bfA\in\mathbb R^{n\times d}$,
\[
    \Pr\braces*{\Omega(1)\norm{\bfA}_1 \leq \norm{\bfS\bfA}_1 \leq O\parens*{\frac1{\delta\alpha}}\norm{\bfA}_1} \geq 1 - \delta.
\]
\end{thm}

Our analysis revolves around the vector of row norms.

\begin{dfn}[Row norms vector]
For an $n\times d$ matrix $\bfA$ with $\norm{\bfA}_1 = 1$, we define the row norms vector $\bfa\in\mathbb R^n$ by $\bfa_i = \norm{\bfe_i^\top\bfA}_1$. Using this vector, we define weight classes $W_q(\bfa)$ and restrictions $\bfA^{(q)}$ of $\bfA$ to our weight classes, analogously to the analysis in Section \ref{section:1+eps-l1-subspace-embedding}. 
\end{dfn}

In order to avoid shrinking the vector $\bfa$ by more than a constant factor with probability at least $\delta$, we apply Lemma \ref{thm:no-contraction-high-prob} with failure rate $\delta$ and constant $\eps$, which gives an $M$-sketch with $0$th level hash bucket size
\[
    N_0 = \tilde O\parens*{\frac{B}{\delta}\log\log n}
\]
and $h$th level hash bucket size
\[
    N = \tilde O\parens*{B\log n}.
\]

We now show that this does not dilate the entrywise $1$-norm of $\bfA$ by more than $O(1/\alpha)$. As in the analysis in \cite{verbin2012rademacher}, we use the Rademacher dimension. 

\begin{lem}[Rademacher dimension of $\ell_1^d$]\label{lem:rademacher-dimension-l1}
Let $\{\bfx_i\}_{i=1}^s\subseteq \mathbb R^d$ with $\norm{\bfx_i}_1\leq 1$ for each $i\in[s]$, and let $\delta\in (0,1)$. Let $\{\eps_i\}_{i=1}^s$ be independent Rademacher variables. Then with probability at least $1-\delta$,
\[
    \Norm{\sum_{i=1}^s \eps_i \bfx_i}_1\leq d\sqrt{\frac12\log\frac{2d}{\delta}}\sqrt{s}.
\]
\end{lem}
\begin{proof}
The proof uses standard concentration inequalities and is similar to \cite[Lemma 1]{verbin2012rademacher}. The details are deferred to Appendix \ref{section:appendix:entrywise_embeddings}. 
\end{proof}

We follow the approach of \cite{verbin2012rademacher}. Using the Rademacher dimension, we first show that if we sample too many elements, then the contribution from this level is at most a negligible fraction of the total mass. 
\begin{lem}\label{lem:entrywise-sign-cancellation}
Let $q\in[q_{\max}]$. Let $p_h\abs{W_q(\bfa)}\geq bN$ for
\[
    b = 2d^2h_{\max}^2q_{\max}^2\log\parens*{\frac{2dNh_{\max}q_{\max}}{\delta}}
\]
Then with probability at least $1 - \delta/q_{\max}$,
\[
    \sum_{h : p_h\abs{W_q(\bfa)}\geq bN} \norm{\bfC^{(h)}\bfS^{(h)}\bfA^{(q)}}_1 \leq \frac{1}{q_{\max}}\norm{\bfA^{(q)}}_1
\]
\end{lem}
\begin{proof}
By Chernoff bounds, the probability that we sample $(1\pm1/2)p_h\abs{W_q(\bfa)}/2N$ elements in a given bucket in the $h$th level is at most
\[
    \exp\parens*{-\frac{(1/2)^2 p_h\abs{W_q(\bfa)}/N}{3}}\leq \exp\parens*{-\frac{b}{12}}\leq  \exp\parens*{-\log\parens*{\frac{Nh_{\max}q_{\max}}{\delta}}} = \frac{\delta}{Nh_{\max}q_{\max}}
\]
so by a union bound over the $N$ buckets, this holds simultaneously for all buckets at the $h$th level with probability at least $\delta/h_{\max}q_{\max}$. 

We condition on the above event. Then, each bucket $L$ is a randomly signed sum of $s\geq b$ elements $\bfe_i^\top\bfA$ with $\norm{\bfe_i^\top\bfA}_1\leq 2^{1-q}$. Thus by Lemma \ref{lem:rademacher-dimension-l1}, with probability at least $\delta/Nh_{\max}q_{\max}$,
\begin{align*}
    \Norm{\sum_{\bfa_i\in L_{h,k}}\Lambda_i \bfe_i^\top\bfA}_1 &\leq 2^{1-q} d\sqrt{\frac12\log\frac{2dN h_{\max}q_{\max}}{\delta}}\sqrt{s} \\
    &\leq \frac{\norm{W_q(\bfa)}_1}{N}\frac{2d}{\sqrt{s}}\sqrt{\frac12\log\frac{2dN h_{\max}q_{\max}}{\delta}} \\
    &\leq \frac{\norm{W_q(\bfa)}_1}{Nh_{\max}q_{\max}}
\end{align*}
as we have set
\[
    \frac{2d}{\sqrt{s}}\sqrt{\frac12\log\frac{2dN h_{\max}q_{\max}}{\delta}}\leq \frac{2d}{\sqrt{b}}\sqrt{\frac12\log\frac{2dN h_{\max}q_{\max}}{\delta}}\leq \frac{\eps}{h_{\max}q_{\max}}.
\]
Summing over the buckets $k\in [N]$ and union bounding and summing over $h\in[h_{\max}]$ yields the desired result. 
\end{proof}

Next, we handle the remaining levels. We pay the price of having smaller hash buckets in the distortion at this point. 

\begin{lem}\label{lem:entrywise-expectation-argument}
Let $q\in [q_{\max}]$. Then with probability at least $1 - 2\delta/q_{\max}$,
\[
    \sum_{h : p_h\abs{W_q(\bfa)} < bN}\norm{\bfS^{(h)}\bfA^{(q)}}_1\leq O\parens*{\frac{1}{\delta\alpha}}\norm{\bfA^{(q)}}_1
\]
\end{lem}
\begin{proof}
Note that if $p_h\abs{W_q(\bfa)}\leq \delta/h_{\max}q_{\max}$, then by a union bound over the at most $h_{\max}$ levels of $h$, none of these levels $h$ sample any elements from weight class $q$ with probability at least $\delta/q_{\max}$. Then for each weight class $q$, only the subsampling levels $h$ for
\[
    \frac{\delta}{h_{\max}q_{\max}} \leq p_h\abs{W_q(\bfa)}\leq bN
\]
can contribute to the mass of the sketch $\norm{\bfS\bfA}_1$. Note that this is only
\[
    \log_B\parens*{\frac{bNh_{\max}q_{\max}}{\delta}} = \log_B\bracks*{\poly(d, \log n, \delta^{-1})} = O\parens*{\frac1\alpha}
\]
levels of subsampling, where each level contributes at most 
\[
    \E\norm{\bfC^{(h)}\bfS^{(h)}\bfA^{(q)}}_1\leq \E\norm{\bfS^{(h)}\bfA^{(q)}}_1 = \norm{\bfA^{(q)}}_1
\]
in expectation. We thus conclude by summing over $h$ with $p_h\abs{W_q(\bfa)} < bN$ and then applying a Markov bound. 
\end{proof}

Putting the above pieces together yield the following:
\begin{proof}[Proof of Theorem \ref{thm:l1-entrywise-embedding}]
As previously discussed in this section, the ``no contraction'' direction of $\norm{\bfS\bfA}_1\geq \Omega(1)\norm{\bfA}_1$ is handled in  Lemma \ref{thm:no-contraction-high-prob}, so we focus on proving the ``no dilation'' direction of $\norm{\bfS\bfA}_1\geq O(1/\delta\alpha)\norm{\bfA}_1$. 

We union bound and sum over the results from Lemmas \ref{lem:entrywise-sign-cancellation} and \ref{lem:entrywise-expectation-argument} for $q\in [q_{\max}]$ to see that with probability at least $1 - 3\delta$,
\[
    \sum_{h\in[h_{\max}]}\sum_{q\in[q_{\max}]}\norm{\bfC^{(h)}\bfS^{(h)}\bfA^{(q)}}_1\leq O\parens*{\frac1{\delta\alpha}}\sum_{q\in[q_{\max}]}\norm{\bfA^{(q)}}. 
\]
We also note that $\norm{\bfC^{(0}\bfA}_1\leq \norm{\bfA}_1$ by the triangle inequality. Finally, we have that in expectation, the weight classes $q > q_{\max}$ contribute at most
\begin{align*}
    \sum_{q > q_{\max}}\sum_{h\in[h_{\max}]}\E\norm{\bfC^{(h)}\bfS^{(h)}\bfA^{(q)}}_1 &\leq \sum_{q > q_{\max}}\sum_{h\in[h_{\max}]}\norm{\bfA^{(q)}}_1 \\
    &\leq \sum_{h\in [h_{\max}]}\frac{2\delta}{nd h_{\max}}\norm{\bfA}_1\sum_{q > q_{\max}}\abs{W_q(\bfa)}\\
    &\leq 2\delta\norm{\bfA}_1.
\end{align*}
Then by Markov's inequality, with probability at least $1-\delta$, these levels contribute at most $2\norm{\bfA}_1$. Summing these three results, we find that
\begin{align*}
    \norm{\bfS\bfA}_1 &\leq \norm{\bfC^{(0}\bfA}_1 + \sum_{h\in[h_{\max}]}\sum_{q\in[q_{\max}]}\norm{\bfC^{(h)}\bfS^{(h)}\bfA^{(q)}}_1 + \sum_{h\in[h_{\max}]}\sum_{q > q_{\max}}\norm{\bfC^{(h)}\bfS^{(h)}\bfA^{(q)}}_1 \\
    &\leq 3\norm{\bfA}_1 +  O\parens*{\frac1{\delta\alpha}}\sum_{q\in[q_{\max}]}\norm{\bfA}_1 \\
    &\leq O\parens*{\frac1{\delta\alpha}}\norm{\bfA}_1
\end{align*}
as desired.
\end{proof}

\subsection{Lower bound}

We show that for $d\times d$ matrices, the above trade-off between the sketching dimension and distortion is nearly optimal, up to log factors. Note that for constant $\delta$, the above result gives a $d^\alpha\poly\log d$ sized sketch with distortion $1/\alpha$. We show that with a sketch of size $r$, a distortion of $\Omega((\log d)/(\log r))$ is necessary. 

By Yao's minimax principle, we assume that the $r\times d$ sketch matrix $\bfS$ is fixed, and show that the distortion is $\Omega((\log d)/(\log r))$ with constant probability over a distribution over input matrices $\bfA$. 

The following simple lemma is central to our analysis:
\begin{lem}\label{lem:sketch-cauchy-matrix}
Let $\bfS$ be an $r\times n$ matrix, and let $\bfA$ be drawn as an $n\times d$ matrix with all of its columns drawn as i.i.d.\ Cauchy variables. Then,
\[
    \Pr\braces*{\Omega(d\log d)\norm{\bfS}_1\leq \norm{\bfS\bfA}_1 \leq O(d\log(rd))\norm{\bfS}_1} \geq \frac{99}{100}.
\]
\end{lem}
\begin{proof}
The proof relies on standard tricks and is deferred to Appendix \ref{section:appendix:entrywise_embeddings}.
\end{proof}

\begin{thm}\label{thm:l1-entrywise-embedding-lower-bound}
Let $\bfS$ be a fixed $r\times d$ matrix. Then there is a distribution $\mu$ over $d\times d$ matrices such that if
\[
    \Pr_{\bfA\sim\mu}\parens*{\norm{\bfA}_1\leq \norm{\bfS\bfA}_1\leq \kappa \norm{\bfA}_1}\geq \frac23
\]
then $\kappa = \Omega((\log d)/(\log r))$. 
\end{thm}
\begin{proof}
We draw our matrix $\bfA$ from $\mu$ as follows. Let $\mu_1$ be the distribution that draws $\bfA$ as a $d\times d$ i.i.d.\ matrix with Cauchy entries, and let $\mu_2$ be the distribution that draws $\bfA$ with its first $r$ columns as a $d\times r$ i.i.d.\ matrix with Cauchy entries scaled by $d/r$, and the rest of the $d-r$ columns all $0$s. Then, $\mu$ draws from $\mu_1$ with probability $1/2$ and $\mu_2$ with probability $1/2$. 

Note that by Lemmas 2.10 and 2.12 of \cite{DBLP:conf/soda/WangW19},
\begin{align*}
    \Pr_{\bfA\sim\mu_1}\parens*{\Omega(d^2\log d)\leq \norm{\bfA}_1 \leq O(d^2\log d)} &\geq \frac{99}{100} \\
    \Pr_{\bfA\sim\mu_2}\parens*{\frac{d}{r}\Omega(rd\log(rd))\leq\norm{\bfA}_1 \leq \frac{d}{r}O(rd\log(rd))} &\geq \frac{99}{100}
\end{align*}

By Lemma \ref{lem:sketch-cauchy-matrix}, if $\bfA\sim \mu_1$, then $\norm{\bfS\bfA}_1 = \Omega(d\log d)\norm{\bfS}_1$ with probability at least $99/100$. Now suppose for contradiction that $\norm{\bfS}_1 = \omega(\kappa d)$. Then with probability at least $1 - (1/3 + 1/2 + 1/100 + 1/100) > 0$, we have that
\[
    \omega(\kappa d^2\log d) = \Omega(d\log d)\norm{\bfS}_1\leq \norm{\bfS\bfA}_1\leq \kappa \norm{\bfA}_1 = O(\kappa d^2\log d)
\]
which is a contradiction. Thus, $\norm{\bfS}_1 = O(\kappa d)$. 

Now consider $\bfA\sim\mu_2$. By Lemma \ref{lem:sketch-cauchy-matrix}, $\norm{\bfS\bfA}_1\leq O(r\log r)\norm{\bfS}_1$ with probability at least $99/100$. Then, with probability at least $1 - (1/3 + 1/2 + 1/100 + 1/100) > 0$,
\[
    \Omega(d^2\log d)\leq \norm{\bfA}_1\leq \norm{\bfS\bfA}_1\leq O(d\log r)\norm{\bfS}_1 = O(\kappa d^2\log r)
\]
so
\[
    \kappa = \Omega\parens*{\frac{\log d}{\log r}}
\]
as desired.
\end{proof}

%% file: prod_test_ell1.tex
\section{Independence Testing in the \texorpdfstring{$\ell_1$}{l1} norm}

In this section, we present our result for estimating $\norm{P-Q}_1$, where $P$ is the joint distribution and $Q$ the product distribution defined by the marginals, which are determined by the stream items as introduced in Section~\ref{sec:intro}. We first prepare a heavy hitter data structure in Section~\ref{sec:HH} and present our $(1+\eps)$-approximation algorithm to the $\ell_1$ norm of order-$2$ tensors in Section~\ref{sec:1+eps}. To move to higher dimensions, we  need a rough estimator for the product distribution in Section~\ref{sec:rough}. Finally, we apply the result for order-$2$ tensors iteratively in Section~\ref{sec:tvd} to obtain a $(1+\eps)$-approximation to $\norm{P-Q}_1$.

\subsection{Heavy Hitters}\label{sec:HH}
This subsection is devoted to a data structure, called the \textsc{HeavyHitter} structure, which is analogous to the classical \textsf{CountSketch} data structure for a general functional $f$ on a general linear space.

Suppose that $f:\R\to \R$ is function satisfying the following properties:
\begin{enumerate}[topsep=0pt,itemsep=-1ex,partopsep=1ex,parsep=1ex]
\item $f(0) = 0$;
\item $f(x) = f(-x)$; 
\item $f(x)$ is increasing on $[0,\infty)$;
\item There exists a constant $C_f$ such that it holds for any integer $s\geq 1$ and any $x_1,\dots,x_s,y_1,\dots,y_s\in \R$ that $\sum_{i=1}^s f(x_i+y_i) \leq C_f(\sum_{i=1}^s (f(x_i) + f(y_i)))$.
\item There exists a function $h:[0,\infty)\to[0,\infty)$ such that
\begin{enumerate}[topsep=0pt,itemsep=-1ex,partopsep=1ex,parsep=1ex]
	\item $\lim_{\eps\to 0^+} h(\eps) = 0$;
	\item it holds for any integer $s\geq 1$ and any $x_1,\dots,x_s,y_1,\dots,y_s\in \R$ that $|\sum_{i=1}^s f(x_i+y_i) - \sum_{i=1}^s f(x_i)| \leq h(\eps)\sum_i f(x_i)$ whenever $\sum_i f(y_i)\leq \eps \sum_i f(x_i)$.
\end{enumerate}
\end{enumerate}

\smallskip

We abuse notation and define for $x\in \R^m$ that $f(x) = \sum_i f(x_i)$. 

We define a different Rademacher dimension as follows. The Rademacher dimension $B=B(f;\eta)$ is  the smallest integer such that the following holds for any integer $s\geq 1$. Let $\sigma_1,\dots,\sigma_s$ be i.i.d.\@ Rademacher variables and $\xi_1,\dots,\xi_s$ be i.i.d.\@ Bernoulli variables such that $\E\xi_i = 1/B$. It holds for any $x_1,\dots,x_s\in \R^m$ that
\[
\Pr\left\{ f\left(\sum_i \sigma_i \xi_i x_i\right) \leq \eta\sum_i f(x_i) \right\}\geq 0.9.
\]

%
\begin{lemma}\label{lem:heavy-hitter}
Let $\gamma,\zeta \in (0,1/3)$, there exists $r=r(\gamma,\zeta)$ and a randomized linear map $T:\R^m\to \R^r$, and a subrecovery  algorithm $\cB$ such that for each $x\in \R^m$, with probability at least $1-\zeta$, it holds that $(1-\gamma)f(x)\leq \cB(Tx)\leq  (1+\gamma)f(x)$.

Then, for $\theta,\delta\in(0,1/3)$, there exists a randomized linear function $M: (\R^m)^d\to \R^{S}$, where $S = O(B\log(d/\delta)\cdot r(\gamma,\zeta'))$ for $B=B(f;h^{-1}(\theta)\theta)$ and  $\zeta'=O(\zeta/(B\log(d/\delta))$, and a recovery algorithm $\cA$ satisfying the following. For any $x=(x_1,\dots,x_n)\in (\R^m)^d$ with probability $\geq 1-\delta-\zeta$, $\cA$ reads $Mx$ and outputs an estimate $\tilde f_i$ for each $i\in [d]$ such that 
\begin{enumerate}
\item $|\tilde f_i - f(x_i)| \leq (\gamma+\theta+\gamma\theta)f(x_i)$ whenever $f(x_i)\geq \theta f(x)$;
\item $|\tilde f_i|\leq C_f\theta(1+\gamma)(1+h(\theta))f(x)$ whenever $f(x_i) < \theta f(x)$.
\end{enumerate}
\end{lemma}

\begin{proof}
The linear sketch $M$ is essentially a \textsf{CountSketch} data structure, which hashes $\{Tx_i\}_i$ into $B=B(f;\min\{\eps^2, h(\theta)\}\eta)$ buckets under a hash function $h$. The $b$-th bucket contains
\[
T_b = \sum_{i:h(i)=b} \sigma_i Tx_i.
\]
For $i^\ast$ such that $f(x_{i^\ast})\geq \theta f(x)$, the algorithm will just return $\tilde f_{i^\ast} = \cB(T_{h(i^\ast)})$. Next we analyse the estimation error. Let $b=h(i^\ast)$. Note that $\sum_{i:h(i)=b} \sigma_i Tx_i$ is identically distributed as $Tx_{i^\ast} + T\nu$, where $\nu = \sum_{i\neq i^\ast: h(i)=b} \sigma_i x_i$. Since $B=B(f;h^{-1}(\theta)\theta)$, it holds with probability at least $0.9$ that
\[
f(\nu) \leq h^{-1}(\theta) \theta f(x)\leq h^{-1}(\theta) f(x_{i^\ast}),
\]
which implies that
\[
\left(1-\theta\right)f(x_{i^\ast}) \leq f\left(x_{i^\ast} + \nu \right) \leq \left(1+\theta\right)f(x_{i^\ast}).
\]
and, with probability at least $0.9-\zeta$ that
\[
(1-\gamma)(1-\theta)\cB(Tx_{i^\ast}) \leq \cB(T_b) \leq (1+\gamma)(1+\theta)\cB(Tx_{i^\ast}).
\]
On the other hand, when $f(x_{i^\ast})\leq \theta f(x)$,
\[
f(x_{i^\ast}+\nu)\leq C_f(f(x_{i^\ast})+f(\nu))\leq C_f(\theta + h(\theta)\theta)f(x)
\]
and, with probability at least $1-\zeta$,
\[
\cB(T_b) \leq C_f\theta(1+\gamma)(1 + h(\theta))f(x).
\]
Repeat $\Theta(\log(d/\delta))$ times to drive the failure probability down to $\delta/d$ to take a union bound over all $i^\ast$.
%
\end{proof}

The data structure described in Lemma~\ref{lem:heavy-hitter} is our \textsf{HeavyHitter} structure, parameterized with $(\theta,\delta)$.

\begin{algorithm2e}
\caption{Data Structure for constant failure probability algorithm (\textsf{SubsamplingHeavyHitter})}
\label{alg:data structure}

\DontPrintSemicolon
\LinesNumbered

\SetKwInput{KwData}{Require}
\KwData{$\eps,\delta,K,N,t,\zeta$}

$L\gets \log(K N/\eps)$\;
$\hat L\gets \log N$\;
$\theta\gets \min\{\Theta(\eps^3/(C_f L^3)),h^{-1}(\alpha\eps/3),\alpha\eps/4\}$\;
$B \gets B(f; h^{-1}(\theta)\theta\})$\;
$Q\gets O(B(\hat L+1)\log(N \hat L))$\;
Instantiate a subsampling function $H$, which hashes $[N]$ into $\hat L$ levels such that the sampling probability for the $\ell$-th level is $2^{-\ell}$ and is pairwise independent\;
\For{each $\ell = 0,1,\dots,\hat L$}{
	Instantiate a \textsf{HeavyHitter} structure $\cD_{\ell}$ with parameters $(\theta,0.05/(\hat L+1))$, in which each bucket stores a vector of length $t = t(\alpha\eps/2,\zeta/Q)$\;
}
\end{algorithm2e}

\begin{algorithm2e}
\caption{Algorithm for an update to $x_i$ for our constant failure probability algorithm} \label{alg:update}

\DontPrintSemicolon
\LinesNumbered

\KwIn{an update of the form $x_i \gets x_i + \Delta x_i$}

\For{each $\ell=0,\dots,\hat L$}{
		\If(\cmt*[f]{Assume $H$ hashes every $i$ to level $0$}){$H$ hashes $i$ into level $\ell$}{
			$b_\ell \gets $ index of the bucket containing $i$ in $\cD_\ell$\;
			Add $T(\Delta x_i)$ to the $b_{\ell}$-th bucket\;
		}
}
\end{algorithm2e}

\begin{algorithm2e}
\caption{$(1+\eps)$-approximator to $f(x)$ with constant failure probability}\label{alg:ell_1}

\DontPrintSemicolon
\LinesNumbered

\SetKwInput{KwData}{Require}

\KwData{(i) A subsampling scheme $H$ such that the $i$-th level has subsampling probability $p_i = 2^{-i}$; (ii) $\hat L+1$ \textsf{HeavyHitter} structures $\cD_0,\dots,\cD_{\hat L}$ with the same parameters $(\theta, \delta)$, where $\hat L = \log N$, $\theta = \min\{\Theta(\frac{\eps^3}{C_f L^3}),\frac{\alpha\eps}{4},h^{-1}(\frac{\alpha\eps}{3})\}$ and $\delta = \frac{1}{20(\hat L+1)}$; (iii) an approximation $\widehat{M}$ such that $M\leq \widehat{M}$; (iv) an integer $K\geq 2$ which is a power of $2$. }

$L \gets \log(2N/\eps)$\;
$\hat L \gets \log N$\;
\For{$j=0,\dots,\hat L$}{
	$\Lambda_j \gets \text{top }\Theta(L^3/\eps^3)\text{ heavy hitters from }\cD_j$\;
}
$j_0\gets \log(4K\eps^{-3}L^3)$\;
$\zeta\gets$ uniform variable in $[1/2,1]$\;
\For{$j = 0,\dots, j_0$}{
	Let $\lambda_1^{(j)},\dots,\lambda_s^{(j)}$ be the elements in $\Lambda_0$ contained in $[(1+\eps)\zeta\frac{\widehat{M}}{2^j}, (2-\eps)\zeta\frac{\widehat{M}}{2^j}]$\;
	$\wtilde{M}_j \gets f(\lambda_1^{(j)})+\cdots+f(\lambda_s^{(j)})$\;
}
\For{$j = j_0 + 1,\dots, L$}{
	Find the biggest $\ell$ such that $\Lambda_\ell$ contains $s$ elements $\lambda^{(j)}_1,\dots,\lambda^{(j)}_s$ in $[(1+\eps)\zeta\frac{\widehat{M}}{2^j}, (2-\eps)\zeta\frac{\widehat{M}}{2^j}]$ for $(1-\sqrt{20}\eps)\frac{L^2}{\eps^2}\leq s\leq 2(1+\sqrt{20}\eps)\frac{L^2}{\eps^2}$\label{alg:deepest level}\;
	\eIf{such $\ell$ exists}{
		$\wtilde{M}_j \gets (f(\lambda^{(j)}_1)+\cdots+f(\lambda^{(j)}_s))2^\ell$\;
	}{
		$\wtilde{M}_j \gets 0$\;
	}
}
\Return $\wtilde{M}\gets \sum_j \wtilde{M}_j$\;
\end{algorithm2e}

\subsection{\texorpdfstring{$(1+\eps)$}{(1+eps)}-Approximator}\label{sec:1+eps}
Suppose that for any $\gamma,\zeta\in(0,1)$ that are small enough, there exist $t=t(\gamma,\zeta)$, a randomized linear map $T:\R^m\to \R^t$ and a subrecovery algorithm $\cB$ such that for each $x\in X$, with probability at least $1-\zeta$, it holds that $(1-\gamma)f(x)\leq \cB(Tx)\leq  (1+\gamma)f(x)$.

Let $x = (x_1,\dots,x_N)\in (\R^m)^N$. In this subsection, we consider the problem of approximating $M = \sum_i f(x_i)$ up to a $(1+\eps)$-factor. We also assume that we have an approximation $\widehat{M}$ to $M$ such that $M \leq \widehat{M}\leq KM$.

Our algorithm is inspired from arguments in \cite{earth_mover}. We prepare the following data structure (Algorithm~\ref{alg:data structure}) with the entry update algorithm (Algorithm~\ref{alg:update}). The recovery algorithm is presented in Algorithm~\ref{alg:ell_1}.

\begin{theorem}\label{thm:ell_1_main-simple}
Let $\eps \in (0,1)$ be small enough and $K\geq 2$ be a power of $2$. Let $\theta,B,Q$ be as defined in Algorithm~\ref{alg:data structure}. There exists an absolute constant $\alpha<1$ and a randomized linear sketch $\Pi: (\R^m)^{N} \to \R^S$, where $S = O(Q \cdot t(\alpha\eps/2, 0.05/Q))$ and a recovery algorithm $\cA$ satisfying the following. 

For any $x=(x_1,\dots,x_N)\in X^N$ and an approximation $\widehat{M}\geq M =\sum_i f(x_i)$, with probability at least $0.6$, $\cA$ reads $\Pi x$ and outputs $\wtilde{M}(x)$ such that
\begin{enumerate}
\item $(1-\eps)M \leq \wtilde{M}(x)\leq (1+\eps)M$ if $\widehat{M} \in [(K/2)M,KM]$;
\item $\wtilde{M}(x)\leq (1+\eps)M$ otherwise.
\end{enumerate}
\end{theorem}
\begin{proof}
There are $\Theta(\log(1/\delta))$ repetitions. In each repetition, there are $(\hat L+1)$ \textsf{HeavyHitter} structures of $O(B\log(N\hat L))$ buckets. 
There are $O(B(\hat L+1)\log(N\hat L))$ buckets in each repetition. Each bucket stores a sketch of length $t(\alpha\eps/2,0.05/Q)$. The total space complexity follows.

%
Since for each bucket the failure probability is $0.05/Q$, we can take a union bound over all buckets and assume that $\cB$ gives accurate answers on all buckets in a repetition with probability at least $0.95$. Then the claimed result follows from Theorem~\ref{thm:ell_1_algorithm} for $\hat M\in [(K/2)M,KM]$ and from Lemma~\ref{lem:ell_1_bigger_overestimate} for $\widehat{M}> KM$ and Lemma~\ref{lem:ell_1_smaller_overestimate} for $\widehat{M} < (K/2)M$.
%
%
\end{proof}

Next we extend the algorithm to handle the case where $\widehat{M} < (K/2)M$.

\begin{theorem}\label{thm:ell_1_main}
Let $\eps,\theta,B,Q,S$ be as in Theorem~\ref{thm:ell_1_main-simple} and $\delta \in (0,1)$. There exists an absolute constant $\alpha<1$ and a randomized linear sketch $\Pi: (\R^m)^{N} \to \R^{S'}$, where $S' = O(S \log K \cdot\log(\delta^{-1}\log K))$ and a recovery algorithm $\cA$ satisfying the following. 

For any $x=(x_1,\dots,x_N)\in X^N$ and an approximation $\widehat{M}$ such that $M\leq \widehat{M}\leq KM$, where $M=\sum_i f(x_i)$, with probability at least $1-\delta$, $\cA$ reads $\Pi x$ and outputs $\wtilde{M}(x)$ such that $(1-\eps)M \leq \wtilde{M}(x)\leq (1+\eps)M$.
\end{theorem}

\begin{proof}
First, in view of Theorem~\ref{thm:ell_1_main-simple}, repeating the Algorithm~\ref{alg:ell_1} $\Theta(\log(1/\zeta))$ times and taking the median reduces the failure probability of a single run to $\zeta$. Hence, with sketch length $O(S\log(1/\zeta))$, we have an algorithm outputting $\wtilde{M}$ such that $(1-\eps)M \leq \wtilde{M}(x)\leq (1+\eps)M$, provided that $(K/2)M\leq \widehat{M}\leq KM$. 

For a general $\widehat{M}$, we run $\log K$ instances of the aforesaid algorithm in parallel, where the parameter $K$ in Algorithm~\ref{alg:ell_1} takes values $2,4,8,\dots,K$, respectively. Note that $\widehat{M}\in [(K/2)M,KM]$ in one of these instances and, with probability at least $1-\zeta$, the output $\widehat{M}$ of this instance satisfies that $\widehat{M}\in [(1-\eps)M, (1+\eps)M]$. For each other instance, with probability at $1-\zeta$, the outputted $\widehat{M}\leq (1+\eps)M$. Setting $\zeta = \delta/(\log K)$ and taking the maximum output $\widehat{M}$ among the $\log K$ instances with a union bound over $\log K$ instances, we obtain an estimate in $[(1-\eps)M, (1+\eps)M]$ with probability at least $1-\delta$, as desired.
\end{proof}

%% file: prod_test_ell1_thm.tex
\subsection{Rough Approximator for \texorpdfstring{$\ell_1$}{l1}-Norm}\label{sec:rough}
Consider the problem of estimating $\|x\|_1$ up to a constant factor for $x\in \Z^{d^q}$ in the turnstile streaming model, where each update changes a coordinate by a $+1$ or a $-1$. Let $N = d^q$.
The following result is due to Braverman and Ostrovsky~\cite{BO09}.

\begin{theorem}[Rough approximation; Corollary 6.6 and Lemma 6.7 in \cite{BO09}]\label{thm:rough-approx}
There exists a randomized linear sketch $\Pi:\Z^N\to \Z^S$ for $S=\tilde{O}(q\log(md))$ and a recovery algorithm $\cA$ satisfying the following. For any $x\in \Z^N$ given in the aforementioned turnstile streaming model of length $m$, with probability at least $0.95$, $\cA$ reads $\Pi x$ and outputs $\widehat{M}$ such that $\|x\|_1\leq \widehat{M}\leq 4^{q^2}(\log d)^q \|x\|_1$.
\end{theorem}

\subsection{Estimation of Total Variation Distance}\label{sec:tvd}
Now we wish to estimate $\|P-Q\|_1$. Recall that $P$ is a general joint distribution and $Q$ the product distribution induced by the marginals of $P$. 

We apply the data structure iteratively in Section~\ref{sec:1+eps}. For $\ell_1$-norm, $f(x) = h(x) = x$, $C_f = 1$, $B(f;\eps) = \Theta(1/\eps^2)$. Therefore, in Theorem~\ref{thm:ell_1_main}, one can take $B_i=(L/\eps)^c$ for some absolute constant $c\geq 4$. The basic setup is presented in Algorithm~\ref{alg:data structure P}. For each $i$, we apply Theorem~\ref{thm:ell_1_main} and obtain a linear sketch $\Pi^{(i)}$ and a recovery algorithm $\cA_i$. The sub-recovery algorithm for $\cD^{(i)}_\ell$ is $\cA_{i-1}$. The entry update calls \texttt{EntryUpdate}($i_1,\dots,i_q,\Delta,q$) on the final sketch (see Algorithm~\ref{alg:update P}) if there is an entry update of $\Delta$ at position $(i_1,\dots,i_q)$. For notational convenience, we assume it is always true that a subsampling hash function hashes all coordinates into level $0$. The overall decoding algorithm calls \texttt{Decode}($q$), see Algorithm~\ref{alg:decode P}.

\begin{algorithm2e}\caption{Data Structure for $P$} \label{alg:data structure P}
\DontPrintSemicolon
\LinesNumbered

$\eps_q\gets \eps$, $\delta_q\gets \delta$, $K\gets 4^{q^2}\log^q d$, $N_q \gets d$, $L_q\gets \log(KN_q/\eps_q)$\;
\For{each $i=q-1,\dots,1$}{
	$\eps_i\gets \alpha\eps_{i+1}$\;
	$N_i\gets d$\;
	$L_i\gets \log(K N_i/\eps_i)$\;
	$\delta_i\gets O(1/(L_i/\eps_i)^c)$\;
}
$t_0\gets 1$\;
\For{each $i=1,\dots,q$}{
	$t_i\gets O((L_i/\eps_i)^c t_{i-1}\log K\log(K/\delta_i))$ \cmt*{sketch length in each bucket}
	$R_i\gets \Theta(\log(1/\delta_i))$ \cmt*{number of repetitions}
	\For{each $r = 1,\dots,R_i$}{
		Initialize $H^{(i,r)}, \cD^{(i,r)}_0,\dots,\cD^{(i,r)}_{\log N+1}$ as in Algorithm~\ref{alg:data structure} for parameters $\eps_i,\delta_i,K,N_i,t_{i-1}$\;
	}
}
\end{algorithm2e}

\begin{algorithm2e}\caption{Update algorithm for $P$: an entry update of $\Delta$ at position $(i_1,\dots,i_q)$} \label{alg:update P}

\DontPrintSemicolon
\LinesNumbered
\SetKwFunction{FEntryUpdate}{EntryUpdate}
\SetKwProg{Fn}{Function}{:}{}

\Fn(\cmt*[f]{invoked on some sketch structure}){\FEntryUpdate{$i_1,\dots,i_q,\Delta,d$}}{
	\For{each pair $(r,\ell)\in [R_d]\times \{0,\dots,\log N_d\}$}{
			\If{$H^{(d,r)}$ hashes $i_d$ into the $\ell$-th level}{
				$B\gets$ set of indices of buckets containing $i_d$ in $\cD^{(d,r)}_{\ell}$\;
				\For{each bucket $b\in B$}{
					\eIf{$d>1$}{
						$\Delta' \gets $ \FEntryUpdate{$i_1,\dots,i_q,\Delta,d-1$} on bucket $b$\;
						Add $\Delta'$ to $b$\;
					}{
						Add $T(\Delta)$ to $b$\;
					}
				}
			}
	}
	\Return the incremental vector to the sketch under $\Pi^{(d)}$\;
}
\end{algorithm2e}

\begin{algorithm2e}\caption{Decoding algorithm $\cA_d$ (for $P$ and $P-Q$)} \label{alg:decode P}
\DontPrintSemicolon
\LinesNumbered
\SetKwFunction{FDecode}{Decode}
\SetKwProg{Fn}{Function}{:}{}

\Fn(\cmt*[f]{This is $\cA_d$}){\FDecode{$d$}}{ 
	\For{each $r=1,\dots,R_d$}{
		$Z_r \gets $ Output of Algorithm~\ref{alg:ell_1} with subdecoding algorithm $\cA_{d-1}$\;
	}
	\Return $\median_r Z_r$\;
}
\end{algorithm2e}

\begin{algorithm2e}
\caption{Data Structure for $Q$} \label{alg:data structure Q}

\DontPrintSemicolon
\LinesNumbered

Let $\eps_i,\delta_i,K,R_i$ be the same as in Algorithm~\ref{alg:data structure P}\;
Let the \textsf{HeavyHitter} sketches $\hat\cD^{(i,r)}_\ell$ be the same as $\cD^{(i,r)}_\ell$ in Algorithm~\ref{alg:data structure} (same hash functions) except for $t_i=1$\;
\end{algorithm2e}

\begin{algorithm2e}
\DontPrintSemicolon
\LinesNumbered
\caption{Entry update for $Q$} \label{alg:update Q}

\KwIn{an update of $\Delta$ at position $(i_1,\dots,i_q)$}
	\For{each $d=1,\dots,q$}{
		\For{each $(r,\ell) \in [R_d]\times\{0,\dots,\log N_d\}$}{
			\If{$H^{(d,r)}$ hashes $i_d$ in level $\ell$}{
				$B\gets $ set of indices of buckets containing $i_d$ in $\hat\cD^{(j,r)}_\ell$\;
				Add $\Delta$ to bucket $b$ for every $b\in B$\;
			}
		}
	}
\end{algorithm2e}

\begin{algorithm2e}\caption{Tensorization of $Q$: construct the sketch for $P_1^f\otimes\cdots\otimes P_n^f$}\label{alg:tensorize Q}

\DontPrintSemicolon
\LinesNumbered

$v^{(0)} = 1$\;
\For{each $i=1,\dots,d$}{
	\For{each $(r,\ell) \in [R_d]\times\{0,\dots,\log N_d\}$}{
		\If{$H^{(d,r)}$ hashes $i_d$ in level $\ell$}{
			$B\gets $ set of indices of buckets containing $i_d$ in $\cD^{(j,r)}_\ell$\;
			\For{each $b\in B$}{
				$a\gets $ bucket value of bucket $b$ in $\hat\cD^{(j,r)}_\ell$\;
				Add $a\cdot v^{(i-1)}$ to bucket $b$\;
			}
		}
	}
	Form $v^{(i)}$ which conforms to the structure of $\Pi^{(i)}$\;
}
\end{algorithm2e}

\begin{theorem}\label{thm:independence-testing}
Suppose that the stream length $m = \poly(d^q)$. There is a randomized sketching algorithm which outputs a $(1\pm\eps)$-approximation to $\|P-Q\|_1$ with probability at least $0.9$, using $\exp(O(q^2 + q\log(q/\eps)+q\log\log d))$ bits of space. The update time is $\exp(O(q^2 + q\log(q/\eps)+q\log\log d))$.
\end{theorem}
\begin{proof}
Let $P^f$ be the frequency vector of the empirical distribution of the input stream and $P_i^f$ be the corresponding frequency vector for the marginal on $X_i$. We have $P = P^f/m$ and $P_i = P_i^f/m$.

Let $\Pi^{(q)}$ be the final linear sketch described above. In parallel we run the rough approximator (Theorem~\ref{thm:rough-approx}), which applies in our setting because the stream items are samples from the distribution and we are counting the empirical frequency. We maintain $\Pi(P^f)$ as described in Algorithm~\ref{alg:update}. For the marginals $P_i^f$, we maintain sketches $S^{(i)}_0 P_i^f,\dots, S^{(i)}_{L'} P_i^f$ as in Algorithm~\ref{alg:data structure Q} and Algorithm~\ref{alg:update Q}. At the end of the stream, we construct $\Pi^{(q)} Q^f$ for $Q^f = P_1^f\otimes\cdots \otimes P_n^f$ as in Algorithm~\ref{alg:tensorize Q}. Then we compute $m^{q-1}\Pi^{(q)} P^f - \Pi^{(q)} Q^f = m^q\Pi(P-Q)$, from which we can recover an approximation to $\|P-Q\|_1$ by invoking $\cA_q$.

Next we analyze the space complexity. Let $N_i = d$. Since we are sketching $m^q (P-Q)$, whose $\ell_1$ norm is an integer and is at most $2m^q$, we see that $K\leq 2m^q = d^{\Theta(q^2)}$ by our assumption that $m = \poly(d^q)$. Set $\eps_q = \eps$ and $\delta_q = O(1)$, then
\[
\eps_{i-1} = \alpha\eps_i,\quad \delta_{i-1} = \poly\left(\frac{\eps_i}{\log(K N_i/\eps_i)}\right)
\]
for all $i$. This implies that
\[
\eps_{q-i} = \alpha^{i-1}\eps,\quad \delta_{q-i} = \frac{\eps^{\Theta(i)}}{q^{\Theta(i)}\log^{\Theta(i)}(\frac{qd}{\eps})}
\]
Therefore the target dimension of $\Pi^{(i)}$ is
\begin{align*}
t_{i+1} &\leq C\left(\frac{\log(K N_i/\eps_{i+1})}{\eps_{i+1}}\right)^c \cdot t_i \cdot\log K \log\left(\frac{K}{\delta_{i+1}}\right)\\
&\leq C' \left(\frac{q^2}{\alpha^{q-i}\eps}\right)^{\Theta(1)} \cdot t_i\cdot (q-i)\polylog\left(\frac{qd}{\eps}\right)
\end{align*}
with $t_1=1$. This implies that
\[
t_n \leq (C')^q\frac{q^{\Theta(q)}}{\alpha^{\Theta(q^2)}\eps^{\Theta(q)}}\cdot q!\cdot \log^{O(q)}\left(\frac{qd}{\eps}\right) =\exp\left(O\left(q^2 + q\log\frac{q}{\eps} + q\log\log d \right)\right).
\]
This space dominates the space needed by the rough estimator. Each coordinate requires $O(\log(m^{q})) = O(q^2\log d)$ bits and the overall space complexity (in bits) follows.

The update time is clearly dominated by the update time for $P$, which is dominated by the sketch length.
\end{proof}

%% file: prod_test_ell1_proof.tex
\subsection{Correctness of Algorithm~\ref{alg:ell_1}}\label{sec:correctness}

We adopt the notation from Section~\ref{sec:1+eps}. Recall that our goal is to estimate $M = \sum_i f(x_i)$  up to a $(1+\eps)$-factor and we also assume that we have an approximation $\widehat{M}$ to $M$ which satisfies that $(K/2)M\leq \widehat{M}\leq KM$. 

Let $\zeta$ be a uniform random variable on $[1/2,1]$. For a magnitude level $j$, define
\[
T_j = \zeta\frac{\widehat{M}}{2^j},
\]
and
\[
S_j = \left\{i\in U: f(x_i) \in \left(T_j, 2T_j\right]\right\},\quad s_j = |S_j|.
\]
Observe that if we scale $K$ by a factor of $2^t$, the magnitude levels are shifted by $t$ levels (new top levels are empty). It is easy to see that the behaviour of Algorithm~\ref{alg:ell_1} is invariant under the concurrent scaling of $K$ and shifting of the magnitude levels (since the bucket contents in the \textsf{HeavyHitter} structures remain the same), we may, with loss of generality, assume that $K=2$ and $M \leq \widehat M\leq 2M$.

Observe that $\sum_{j\geq 1}\sum_{i\in S_j} f(x_i) = M$. Note that each element in level $j > \log(2N/\eps)$ is at most $\widehat{M}/2^j < (\eps/(2N))\widehat{M} < (\eps/N)M$, so it contribute at most $\eps M$ and thus can be omitted. That is, we only need to consider the levels up to $L = \log(2N/\eps)$.

We call a level $j$ important if
\[
\frac{s_j}{2^j} \geq \frac{\eps}{2 L}
\]
and we let $\cJ$ denote the set of important levels $j$. The non-important levels contributes at most
\[
\sum_{j\not\in \cJ} \sum_{i\in S_j} f(x_i) \leq \sum_{j\not\in \cJ} \frac{2\eps}{2 L} M \leq \eps \widehat{M} \leq 2\eps M.
\]

The goal of this section is to prove the following theorem.
\begin{theorem}\label{thm:ell_1_algorithm}
Algorithm~\ref{alg:ell_1} returns an estimate $\wtilde{M}$, which with probability at least $0.7$ (over $\zeta$ and subsampling) satisfies that
\[
(1-O(\eps))M \leq \wtilde{M} \leq (1+O(\eps))M.
\]
\end{theorem}
The rest of the section is devoted to the proof of the theorem. 
We assume that all \textsc{Count-Min} structures return correct values, at the loss of $0.05$ probability. The main argument is decomposed into the following lemmas.

\begin{lemma}\label{lem:subsampling}
With probability at least $0.95$ (over subsampling), the following holds for all $j > j_0$ and $j\in \cJ$. There exists an $\ell$ such that the substream induced by $H_\ell$ contains at least $(1-O(\eps))L^2/\eps^2$ and at most $2(1+O(\eps))L^2/\eps^2$ elements of $S_j$. Furthermore, it holds $\frac{3}{4}L^2/\eps^2\leq s_j2^{-\ell}\leq \frac{9}{4}L^2/\eps^2$ for any such $\ell$.
\end{lemma}
\begin{proof}
Since $j\in \cJ$ and $j > j_0$,
\[
s_j \geq \frac{\eps}{2L}\cdot 2^{j_0} = \frac{4}{\eps^2} L^2.
\]
Let $N_{j,\ell}$ denote the number of survivors in the $\ell$-th subsampling level. For $\ell=1$, we have
\[
\E N_{j,\ell} = s_j 2^{-\ell} \geq \frac{4L^2}{\eps^2} 2^{-\ell} = \frac{2L^2}{\eps^2}.
\]
Note that $s_j\leq 2^j$, and thus for $\ell = j - \log(\eps^{-2}L^2) > 0$, we have 
\[
\E N_{j,\ell} = s_j 2^{-\ell} \leq  2^{j-\ell} = \eps^{-2} L^2
\]
survivors after sampling. Hence, there exists $\ell$ such that $\eps^{-2}L^2\leq \E N_{j,\ell} \leq 2\eps^{-2} L^2$. For any such $\ell$, since $H_{\ell}$ is pairwise independent, we have $\Var(N_{j,\ell})\leq \E N_{j,\ell}$ and it follows from Chebyshev's inequality that with probability at least $1-1/(20L^2)$, 
\[
N_{j,\ell} = \E N_{j,\ell}  \pm \sqrt{20L^2\cdot \E N_{j,\ell}},
\]
that is,
\begin{equation}\label{eqn:N_{j,l}}
(1-\sqrt{20}\eps)\frac{L^2}{\eps^2} \leq N_{j,\ell} \leq 2(1+\sqrt{20}\eps)\frac{L^2}{\eps^2}.
\end{equation}
A similar argument shows that for each $\ell$, with probability at least $1-1/(20L^2)$, we have $N_{j,\ell} \geq \frac{9}{4}(1-\sqrt{20}\eps)L^2/\eps^2$ if $\E N_{j,\ell} \geq \frac{9}{4}L^2/\eps^2$ and $N_{j,\ell} \leq \frac{3}{4}(1+\sqrt{20}\eps)L^2/\eps^2$ if $\E N_{j,\ell} \leq \frac{3}{4}L^2/\eps^2$. Taking a union bound over all $L$, we have that with probability at least $1-1/(20L)$ there exists a unique $\ell$ such that \eqref{eqn:N_{j,l}} holds; furthermore, $s_j2^{-\ell} = \E N_{j,\ell} = \Theta(\eps^{-2}L^2)$ for this $\ell$.

The claimed result follows from a union bound over $j$.
\end{proof}


%

Let $\alpha \in (0,1)$ be a small constant. Define
\begin{gather*}
S_j^\ast = \left\{i\in S_j: f(x_i) \in [(1+(1-\alpha)\eps)T_j, (2-(1-\alpha)\eps)T_j] \right\}\\
S_j^{\ast\ast} = \left\{i\in S_j: f(x_i) \in [(1+(1+\alpha)\eps)T_j, (2-(1+\alpha)\eps)T_j] \right\}
\end{gather*}
and
\[
M_j^\ast = \sum_{i\in S_j^\ast} f(x_i),\quad M_j^{\ast\ast} = \sum_{i\in S_j^{\ast\ast}} f(x_i).
\]
Suppose that the event in Lemma~\ref{lem:subsampling} occurs. 

\begin{lemma}\label{lem:deep_layers}
With probability at least $0.9$ (over the subsamplings), it holds for each $j\in \cJ$ and $j > j_0$ that $(1-O(\eta))M_j \leq \E\wtilde{M}_j\leq (1+O(\eta))$, where the expectation is taken over the subsampling.
\end{lemma}
\begin{proof}
Let $I_\ell\subseteq [N]$ be the set of indices in subsampling level $\ell$, where $\ell$ is found in Step~\ref{alg:deepest level} of Algorithm~\ref{alg:ell_1}. Then
\[
\E f(x_{I_\ell}) = \frac{f(x)}{2^\ell}.
\]
By Lemma~\ref{lem:subsampling} and our choice of $\beta$, we have
\begin{equation}\label{eqn:s_j}
\frac{s_j}{2^\ell} \leq \frac{9}{4}\cdot \frac{L^2}{\eps^2}.
\end{equation}
Together with the assumption that $j\in\cJ$,
\[
2^j\leq \frac{L}{\eps}s_j \leq \frac{9}{4}\cdot \frac{L^3}{\eps^3} 2^\ell,
\]
which implies that (by adjusting constants)
\[
\frac{\eps^3}{L^3}\E f(x_{I_\ell}) \leq \frac{9}{4}\frac{f(x)}{2^j} = \frac{9}{4}\frac{1}{\zeta}T_j\leq \frac{9}{2}T_j.
\]
Except with probability $0.05/L$, we have
\[
\frac{\eps^3}{180L^3} f(x_{I_\ell}) \leq T_j.
\]
Now, let $\theta = \min\{\eps^3/(180C_f L^3),\alpha\eps/4,h^{-1}(\alpha\eps/3)\}$ in Lemma~\ref{lem:heavy-hitter}, we have the guarantees that (1) if $f(x_i)\geq \theta T_j$ then it is estimated up to an additive error of at most 
\[
(\gamma+\theta+\gamma\theta) f(x_i) \leq (\gamma+2\theta)f(x_i) \leq \left(\frac{\alpha\eps}{2} + 2\cdot\frac{\alpha\eps}{4}\right)f(x_i) = \alpha\eps f(x_i),
\]
and (2) if $f(x_i)\leq \theta T_j$ we obtain an estimate at most 
\begin{align*}
C_f\theta (1+\gamma)(1+h(\theta))f(x_{I_\ell}) \leq \frac{\eps^3}{40L^3} (1+\gamma) (1+h(\theta))f(x_{I_\ell}) &\leq (1+\gamma)(1+h(\theta))T_j \\
&\leq \left(1+\frac{\alpha\eps}{2}\right)\left(1+\frac{\alpha\eps}{3}\right) T_j \\
&\leq (1+\alpha\eps) T_j.
\end{align*}
Hence, all survivors in level $S_j^{\ast\ast}$ will be recovered and all survivors in the higher levels will not be mistakenly recovered in level $j$; 
survivors from lower levels will not collude to form a heavy hitter. 

Let $R^{(j)} = \{i_1, \dots, i_s\}$, we have $\lambda^{(j)}_r = (1+O(\eps))f(x_{i_r})$ for all $r\in [s]$. Then
\[
\wtilde{M}_j = (\lambda^{(j)}_1 + \cdots + \lambda^{(j)}_s)2^\ell = (1\pm (O(\eps))\wtilde{M}_j',
\]
where
\[
\wtilde{M}_j' = 2^\ell \sum_{r=1}^s f(x_{i_r})
\]
will be our focus. Combining Lemma~\ref{lem:subsampling} with the recovery guarantee of $\cD_\ell$, we see that all elements in $S_j$ that survives the subsampling at level $\ell$ will be recovered. 
Hence, $\Pr\{i\in R^{(j)}\}\leq 2^{-\ell}$ for $i\in S_j^{\ast}$ (because it may not be recovered in our range) and $\Pr\{i\in R^{(j)}\}= 2^{-\ell}$ for $i\in S_j^{\ast\ast}$ (because if it survives the subsampling it would be recovered). Hence
\[
\E \wtilde{M}_j' = 2^\ell \sum_{i\in S_j^{\ast}} f(x_i) \Pr\{i\in R^{(j)}\} \leq 2^\ell \sum_{i\in S_j^\ast} f(x_i) 2^{-\ell} = M_j^\ast
\]
and
\[
\E \wtilde{M}_j' \geq 2^\ell \sum_{i\in S_j^{\ast\ast}} f(x_i) \Pr\{i\in R^{(j)}\} = 2^\ell \sum_{i\in S_j^{\ast\ast}} f(x_i) 2^{-\ell} = M_j^{\ast\ast}
\]
\end{proof}

\begin{lemma}\label{lem:top_layers}
With probability at least $0.95$ (over the subsamplings), it holds for all $j\leq j_0$ that $(1-O(\eps))M_j^{\ast\ast} \leq \wtilde{M}_j\leq (1+O(\eps))M_j^\ast$. 
\end{lemma}
\begin{proof}
The argument is similar to the preceding lemma. Note that there are at most $2^{j_0+1} = 4L^3/\eps^3$ elements of interest in this case, and $\cD_0$ is guaranteed to recover all of them, since 
\[
f(x_i)\geq \xi 2^{-j_0}f(x) \geq \frac{\eps^3}{4L^3}f(x)
\]
and we choose $\theta=\min\{\eps^3/(4 C_f L^3),\alpha\eps/4,h^{-1}(\alpha\eps/3)\}$ for $\cD_0$, where $C$ is an absolute constant. Each $f(x_i)$ is estimated up to an $(1+O(\eps))$-factor.
\end{proof}

\begin{lemma}
With probability at least $0.8$ (over subsamplings) it holds that
\[
(1-O(\eps))\sum_{j\in\cJ} M_j^{\ast\ast} - O(\eps M)\leq \sum_{j\in\cJ} \wtilde{M}_j\leq (1+O(\eps) )M.
\]
\end{lemma}
\begin{proof}
Note that $f(x_i)$ are within a factor of $2$ from each other for $i\in S_j^\ast$, thus
\[
\sum_{i\in S_j^\ast} f(x_i)^2 \leq \frac{4}{|S_j^\ast|}  \left(\sum_{i\in S_j^\ast} f(x_i)\right)^2 = \frac{4}{|S_j^\ast|} (M_j^\ast)^2.
\]
When $j>j_0$ and $j\in \cJ$, we showed that $|S_j^\ast| \geq \eps^{-2}L^2$ (Lemma~\ref{lem:subsampling}), thus
\[
\sum_{i\in S_j^\ast} f(x_i)^2 = O\left(\frac{\eps^2}{L^2} (M_j^\ast)^2\right).
\]
It follows from Chebyshev's inequality that with probability at least $0.95$,
\[
\left|\sum_{j>j_0,j\in\cJ}\wtilde{M}_j - \E\sum_{j>j_0,j\in\cJ}\wtilde{M}_j\right| = O\left(\eps \sum_{j>j_0,j\in\cJ}M_j^\ast\right)
\]
Combining with Lemma~\ref{lem:deep_layers}, we have with probability at least $0.85$,
\[
(1-O(\eps))\sum_{j>j_0,j\in\cJ}M_j^{\ast\ast} - O(\eps)\sum_{j>j_0,j\in\cJ}M_j^\ast \leq 
\sum_{j>j_0,j\in\cJ}\wtilde{M}_j \leq (1+O(\eps)) \sum_{j>j_0,j\in\cJ}M_j^\ast
\]
Further combining with Lemma~\ref{lem:top_layers}, we have with probability at least $0.8$,
\[
(1-O(\eps))\sum_{j\in\cJ}M_j^{\ast\ast} - O(\eps)\sum_{j>j_0,j\in\cJ}M_j^\ast \leq 
\sum_{j\in\cJ}\wtilde{M}_j \leq (1+O(\eps)) \sum_{j\in\cJ}M_j^\ast
\]
The result follows from the observation that $\sum_{j} M_j^\ast\leq M$.
\end{proof}

Note that the levels $j\not\in \cJ$ contribute at most $O(\eps M)$ in expectation to the total norm. By Markov's inequality, except with probability 0.05 (over subsampling), they contribute at most $O(\eps M)$. Combining with the preceding lemma, we have concluded that with probability at least $0.75$,
\[
(1-O(\eps))\sum_{j\geq 1}M_j^{\ast\ast} - O(\eps M) \leq 
\sum_{j\geq 1}\wtilde{M}_j \leq (1+O(\eps)) M.
\]
Over the randomness of $\zeta$, for each $i$, with probability at least $1-O(\eps)$, we have $i\in S_{j'}^{\ast\ast}$ for some $j'$. This implies that 
\[
\E\left(M - \sum_{j\geq 1} M_j^{\ast\ast}\right) = O(\eps M).
\]
By Markov's inequality, we have with probability (over $\zeta$) at least $0.95$ that
\[
(1-O(\eps))M\leq \sum_{j\geq 1} M_j^{\ast\ast} \leq M.
\]
Finally, combining with the failure probability of the \textsf{HeavyHitter} structures,  we conclude that with probability at least $0.7$,
\[
(1-O(\eps))M \leq \sum_{j\geq 1} \wtilde{M}_j  \leq (1+O(\eps))M.
\]

%
%

\subsection{Analysis of Algorithm~\ref{alg:ell_1} with Bad \texorpdfstring{$\widehat{M}$}{hat M}}

We have proved that Algorithm~\ref{alg:ell_1}, when provided a good overestimate $\widehat{M}$, gives a good estimate $\wtilde{M}$ to $M$ in the preceding Section~\ref{sec:correctness}. In this section, we show that the algorithm does not overestimate when $\widehat{M}$ is bad. We follow the notations in the preceding section and assume likewise that $K=2$.

\begin{lemma}\label{lem:ell_1_bigger_overestimate}
Suppose that $\widehat{M} > 2M$. Algorithm~\ref{alg:ell_1} returns an estimate $\wtilde{M}$, which with probability at least $0.7$ (over $\zeta$ and subsampling) satisfies that $\wtilde{M} \leq (1+O(\eps))M$.
\end{lemma}
\begin{proof}
There exists $j^\ast < K$ such that $\widehat{M}^\ast = \widehat{M}/2^{j^\ast} \in [M,2M]$. We compare the behavior of Algorithm~\ref{alg:ell_1} on estimate $\widehat{M}$  and $\widehat{M}^\ast$, under the same randomness in the subsampling functions, heavy hitter structures and $\zeta$. 
Denote the magnitude levels associated with $\widehat{M}^\ast$ by $S^\ast_1, S^\ast_2, \dots$ and the levels associated with $\widehat{M}$ by $S_1,S_2\dots$. It is clear that $S_1 = \cdots = S_{j^\ast} = \emptyset$ and $S_j = S^\ast_{j - j^\ast}$ for $j > j^\ast$. Hence for $j\leq j_0$, we can still recover all items in $S^\ast_1,\dots,S^\ast_{j_0^\ast}$ for $j_0^\ast = j_0 - j^\ast$, that is, all items in $S_1,\dots,S_{j_0^\ast+j^\ast}$. Observe that $j^\ast < \log K$ and so $j_0^\ast + j^\ast < j_0$, and so it is possible that we miss the levels $S_j$ for $j = j_0^\ast+j^\ast+1,\dots,j_0$ since the subsequent for-loop starts with $S_{j_0+1}$. All the recovered levels are within $(1\pm O(\eps))$-factor of their true values, according to the proof of Theorem~\ref{thm:ell_1_algorithm}, with probability at least $0.7$. Therefore, we shall never overestimate, that is, $\wtilde{M} \leq (1+O(\eps))M$.
\end{proof}

\begin{lemma}\label{lem:ell_1_smaller_overestimate}
Suppose that $\widehat{M} < M$. Algorithm~\ref{alg:ell_1} returns an estimate $\wtilde{M}$, which with probability at least $0.7$ (over $\zeta$ and subsampling) satisfies that $\wtilde{M} \leq (1+O(\eps))M$.
\end{lemma}
\begin{proof}
There exists $j^\ast$ such that $2^{j^\ast}\widehat{M} = \widehat{M}^\ast \in [M, 2M]$. Similar to the proof of Lemma~\ref{lem:ell_1_bigger_overestimate}, we compare the behavior of Algorithm~\ref{alg:ell_1} on estimate $\widehat{M}$ and $\widehat{M}^\ast$, under the same randomness in the subsampling functions, heavy hitter structures and $\zeta$. Let $\{S_j^\ast\}$ and $\{S_j\}$ be as defined in the proof of Lemma~\ref{lem:ell_1_bigger_overestimate}. Now we have $S_j = S^\ast_{j+j^\ast}$ and may miss the bands $S^\ast_1,\dots,S^\ast_{j^\ast}$. The rest follows as in Lemma~\ref{lem:ell_1_bigger_overestimate}.
\end{proof}

%% file: subspace_embedding_random.tex
\section{\texorpdfstring{$\ell_1$}{l1} Subspace Embeddings for i.i.d.\ Random Design Matrices}\label{section:subspace_embeddings_random}

In this section we present oblivious $\ell_1$ subspace embeddings for i.i.d.\ random design matrices. This allows us to achieve a polynomial-sized sketch without paying the general case distortion lower bound of $\Omega(d/\log^2 r)$ of \cite{DBLP:conf/soda/WangW19}. 

In consideration of practical applications, we specifically focus on \emph{heavy-tailed distributions}. In fact, as we will see, these are the most interesting from a theoretical perspective as well. Our model for our heavy-tailed distributions will be symmetric power law distributions of index $p$, which are distributions that satisfy
\[
    1 - F(x) \sim cx^{-p}
\]
for a constant $c$. In the literature, works such as \cite{DBLP:conf/nips/ZhangZ18, balkema2018linear} have considered linear regression in the $\ell_1$ norm with heavy tailed i.i.d.\ design matrices. 

\begin{table}[ht]
    \centering
    \begin{tabular}{|c|c|c|}
        \hline
        $p$ & Distortion upper bound & Distortion lower bound \\
        \hline
        $p\in (0,1)$ & $O(1)$ (Theorem \ref{thm:subspace-embedding-p<1}) & $1$ \\
        \hline
        $p = 1$ & $O\parens*{\frac{\log n}{\log(r/d^2\log d)}}$ (Theorem \ref{thm:sketch-cauchy}) & $\Omega\parens*{\frac{\log n}{\log r}}$ (Theorem \ref{thm:lower-bound-cauchy}) \\
        \hline
        $p\in (1, 2)$ & \begin{tabular}{c}
            $O\parens*{\frac{d^{1/p}}{(r/d^2)^{1-1/p}}}$, $n^{1-1/p}>d^{1/p}\log d$ (Theorem \ref{thm:1<p<2-large-n}) \\
            $O\parens*{\frac{d^{1/p}\log d}{(r/d^2\log d)^{1-1/p}}}$, $n^{1-1/p}\leq d^{1/p}\log d$ (Theorem \ref{thm:p-1-2-countsketch})
        \end{tabular} & $\Omega\parens*{\frac{d^{1/p}}{r^{1-1/p}}}$ (Theorem \ref{thm:lower-bound-p-stable}) \\
        \hline
        $p\geq 2$ & $1+\eps$ (Theorem \ref{thm:p>=2-upper-bound}) & $1$ \\
        \hline
    \end{tabular}
    \caption{Results for i.i.d.\ symmetric power law design matrices}
    \label{tab:my_label}
\end{table}

Throughout this section, let $\mathcal D$ be a symmetric power law distribution with index $p\geq 0$, and let $\bfA\sim \mathcal D^{n\times d}$ be a matrix drawn with i.i.d.\ entries drawn from $\mathcal D$, unless noted otherwise. 

\subsection{Setup for analysis}

Let $\bfv\in\mathbb R^n$ be a vector. We will frequently refer to the $k$th level set $\bfv_{(k)}$ of $\bfv$, which takes on the values of $\bfv$ whenever it has absolute value in $[2^k, 2^{k+1})$, and $0$ otherwise. 
\begin{dfn}[Level sets of a vector]
We define the $k$th level set $\bfv_{(k)}$ of $\bfv\in\mathbb R^n$ coordinate-wise by
\[
    \bfe_i^\top\bfv_{(k)} \coloneqq 
    \begin{cases}
        \bfe_i^\top\bfv & \text{if $\abs{\bfe_i^\top\bfv}\in [2^k, 2^{k+1})$} \\
        0 & \text{otherwise}
    \end{cases}.
\]
For $k = 0$, we set
\[
    \bfe_i^\top\bfv_{(0)} \coloneqq 
    \begin{cases}
        \bfe_i^\top\bfv & \text{if $\abs{\bfe_i^\top\bfv}\in [0, 2)$} \\
        0 & \text{otherwise}
    \end{cases}.
\]
\end{dfn}

We will repeatedly make use of the following simple lemmas about \textsf{CountSketch} and symmetric power law distributions. 

\begin{lem}[No expansion]\label{lem:no-expansion}
Let $\bfS$ be drawn as an $r\times n$ \textsf{CountSketch} matrix with random signs $\sigma:[n]\to\{\pm1\}$ and hash functions $h:[n]\to[r]$. Then for all $\bfv\in\mathbb R^n$,
\[
    \norm{\bfS\bfv}_1\leq \norm{\bfv}_1. 
\]
\end{lem}
\begin{proof}
\[
    \norm{\bfS\bfv}_1 = \sum_{i=1}^r\Abs{\sum_{j: h(j) = i}^n \sigma_i v_i}\leq \sum_{i=1}^r \sum_{j:h(j)=i}^n \abs{\sigma_i v_i} = \sum_{j=1}^n \abs{v_i} = \norm{\bfv}_1\qedhere
\]
\end{proof}

\begin{lem}\label{lem:chernoff-power-law-levels}
Let $\mathcal D$ be a symmetric power law distribution with index $p > 0$. Then, for $k$ a large enough constant depending on $\mathcal D$,
\[
    \Pr_{X\sim\mathcal D}\parens*{\abs{X}\in [2^k, 2^{k+1})} = \Theta(2^{-kp})
\]
and
\[
    \Pr_{\bfv\sim \mathcal D^n}\parens*{ \norm{\bfv_{(k)}}_0 = \Theta(n2^{-kp})} \geq 1 - \exp\parens*{\Theta(n2^{-kp})}
\]
\end{lem}
\begin{proof}
For a large enough $k$, we have that
\[
    \Pr_{X\sim\mathcal D}\parens*{\abs{X}\in [2^k, 2^{k+1})} = \overline F(2^{k}) - \overline F(2^k) = \Theta\parens*{\frac1{2^{{k}p}} - \frac1{2^{{(k+1)}p}}} = \Theta(2^{-kp}).
\]
Then in expectation,
\[
    \norm{\bfv_{(k)}}_0 = \Theta(n2^{-kp})
\]
so we conclude by Chernoff bounds.
\end{proof}

\begin{lem}\label{lem:power-law-union-bound}
Let $\mathcal D$ be a symmetric power law distribution with index $p > 0$. Then,
\[
    \Pr_{\bfA\sim\mathcal D^{n\times d}}\parens*{\norm{\bfA}_\infty \leq O\parens*{\parens*{nd/\delta}^{1/p}}}\geq 1 - \delta
\]
\end{lem}
\begin{proof}
Each entry is at most $O((nd/\delta)^{1/p})$ with probability at least $\delta/nd$, so we conclude by a union bound over the $nd$ entries. 
\end{proof}

\begin{dfn}[Truncation]
For $T>0$ and $x\in\mathbb R$, define
\[
    \mathsf{trunc}_T(x)\coloneqq \begin{cases}
        x & \text{$\abs{x}\leq T$} \\
        0 & \text{otherwise} 
    \end{cases}.
\]
For a distribution $\mathcal D$, we define $\mathsf{trunc}_T(\mathcal D)$ to be the distribution that draws $\mathsf{trunc}_T(X)$ for $X\sim\mathcal D$. 
\end{dfn}

\begin{lem}[Moments of truncated power laws]\label{lem:truncated-p-moments}
Let $\mathcal D$ be a power law distribution with index $p>0$. Let $T>0$ be sufficiently large. Then
\begin{align*}
    \E_{X\sim\mathsf{trunc}_T(\mathcal D)} \abs{X} &= \begin{cases}
        \Theta(T^{1-p}) & \text{if $p \in (0,1)$} \\
        \Theta(\log T) & \text{if $p = 1$} \\
        \Theta(1) & \text{if $p > 1$}
    \end{cases} \\
    \E_{X\sim\mathsf{trunc}_T(\mathcal D)} X^2 &= \begin{cases}
        \Theta(T^{2-p}) & \text{if $p \in (0,2)$} \\
        \Theta(\log T) & \text{if $p = 2$} \\
        \Theta(1) & \text{if $p > 2$}
    \end{cases} \\
\end{align*}
\end{lem}
\begin{proof}
The proof is deferred to Appendix \ref{section:appendix:subspace_embeddings_random}. 
\end{proof}

\begin{dfn}\label{def:AH-AL}
Let $\bfA\in\mathbb R^{n\times d}$ and let $T > 0$. Then, we write $\bfA = \bfA^H + \bfA^L$ where $\bfA^H$ is the submatrix of $\bfA$ formed by the rows containing an entry with absolute value at least $T$, and $\bfA^L$ is the rest of the rows. 
\end{dfn}

\subsection{Algorithms for \texorpdfstring{$p<1$}{p < 1}}
We first present the results that for tails that are very heavy admit $O(1)$ distortion embeddings in $\poly(d)$ dimensions for a very simple reason: when $p < 1$, then the largest entry in every vector is a good approximation of the entire $\ell_1$ mass of the vector. 

\begin{thm}\label{thm:subspace-embedding-p<1}
Let $\mathcal D$ be a symmetric power law distribution with index $p\in(0, 1)$. Let $\bfS$ be drawn as a \textsf{CountSketch} matrix with $r = O(d^2\log^2 d)$ rows. Then
\[
    \Pr_{\bfS, \bfA}\parens*{\Omega(1)\norm{\bfA\bfx}_1\leq \norm{\bfS\bfA\bfx}_1\leq \norm{\bfA\bfx}_1, \forall \bfx\in\mathbb R^d}\geq \frac{99}{100}.
\]
\end{thm}
\begin{proof}
The proof proceeds similarly to the case of $p = 1$, and is deferred to Appendix \ref{section:appendix:subspace_embeddings_random}.
\end{proof}

Thus, we focus on the regime of $p\geq 1$. 

\subsection{Algorithms for \texorpdfstring{$p = 1$}{p = 1}}

In this section, we prove the following:

\begin{thm}\label{thm:sketch-cauchy}
Let $\mathcal D$ be a symmetric power law distribution with index $p = 1$. Let $\bfS\in\mathbb R^{r\times n}$ be drawn as a \textsf{CountSketch} matrix. Then, for any $C (d\log d)^2\leq r\leq o(\sqrt n)$ for $C$ a large enough constant, we have
\[
    \Pr_{\bfA\sim\mathcal D^{n\times d}}\braces*{\frac1\kappa\norm{\bfA\bfx}_1\leq\norm{\bfS\bfA\bfx}_1\leq\norm{\bfA\bfx}_1, \forall \bfx\in\mathbb R^d} \geq \frac{99}{100}
\]
for
\[
    \kappa = O\parens*{\frac{\log n}{\log(r/d^2\log d)}}.
\]
\end{thm}

The idea is that with $r$ rows of \textsf{CountSketch}, we can preserve the top $r$ entries of $\bfA\bfx$, which has mass approximately $\Omega(n\log r)\norm{\bfx}_1$, while the rest of the entries have mass at most $O(n\log n)\norm{\bfx}_1$. We formalize this idea in the following several lemmas. 

\begin{lem}[Mass of small entries]\label{lem:p=1-small-entries}
Let $\mathcal D$ be a power law distribution with index $p = 1$ and let $\bfA\sim \mathcal D^{n\times d}$. Let $\poly(d)\leq T\leq n$ and let $\bfA = \bfA^H + \bfA^L$ as in Definition \ref{def:AH-AL}. Then,
\[
    \Pr\parens*{\norm{\bfA^L\bfx}_1\leq O(n\log T)\norm{\bfx}_1, \forall \bfx\in\mathbb R^d}\geq 0.99
\]
\end{lem}
\begin{proof}
Note that $\bfA^L$ is drawn i.i.d.\ from $\mathsf{trunc}_T(\mathcal D)$ so by Lemma \ref{lem:truncated-p-moments}, it has entries with first two moments
\begin{align*}
    \mu &\coloneqq \E_{X\sim\mathcal D}\parens*{\abs{X}\mid \abs{X}\leq T} = \Theta(\log T) \\
    \sigma^2 &\coloneqq \E_{X\sim\mathcal D}\parens*{X^2\mid \abs{X}\leq T} = \Theta(T)
\end{align*}
Then for a single column $\bfA^L\bfe_j$ for $j\in[d]$, by Bernstein's inequality,
\[
    \Pr\parens*{\norm{\bfA^L\bfe_j}_1 \geq 2\mu n}\leq \exp\parens*{-\frac12\frac{(\mu n)^2}{\sigma^2 n + \mu n T/3}}\leq \exp(-\Omega(\log T)) = \frac1{\poly(T)}.
\]
Then since $T \geq \poly(d)$, we may union bound over the $d$ columns so that $\norm{\bfA^L\bfe_j}_1 = O(\mu n) = O(n\log T)$ for all columns $j\in[d]$ with probability at least $1 - 1/\poly(d)$. Conditioned on this event, we have for all $\bfx\in\mathbb R^d$ that
\[
    \norm{\bfA^L\bfx}_1\leq \sum_{j=1}^d \abs{\bfx_j}\norm{\bfA^L\bfe_j}_1 = O(n\log T)\norm{\bfx}_1.\qedhere
\]
\end{proof}

\begin{lem}[Unique hashing of large entry rows]\label{lem:p=1-unique-hashing-rows}
Let $\mathcal D$ be a power law distribution with index $p>0$ and cdf $F$, and let $\bfA\sim \mathcal D^{n\times d}$. Let $t$ and $r$ be parameters such that $r \geq C(d\log d)^2$ for a sufficiently large constant $C$ and $r = o(\sqrt n)$, and define 
\begin{align*}
    \tau_1 &\coloneqq F^{-1}\parens*{1-\frac{C'\log d}{n}} = \Theta((n/\log d)^{1/p}) \\
    \tau_2 &\coloneqq F^{-1}\parens*{1-\frac{r/d^2\log d}{n}} = \Theta((nd^2\log d/r)^{1/p}) \\
    \mathcal R_1 &\coloneqq \braces*{i\in[n] : \exists j\in[d], \abs{\bfe_i^\top\bfA\bfe_j} > \tau_1} \\
    \mathcal R_2 &\coloneqq \braces*{i\in[n] : \exists j\in[d], \abs{\bfe_i^\top\bfA\bfe_j} > \tau_2}
\end{align*}
for a sufficiently large constant $C'$. Then if we hash each row of $\bfA$ into $O(r)$ hash buckets, with probability at least $0.95$:
\begin{itemize}
    \item Every row of $\mathcal R_1$ is hashed to a bucket with no other row from $\mathcal R_2$.
    \item For every column $j\in[d]$ and every integer $\log_2 \tau_2\leq k\leq \log_2 \tau_1$ has $\Theta(n/2^{kp})$ rows with a large entry with absolute value in $[2^k, 2^{k+1})$ that are hashed to a bucket with no other row from $\mathcal R_2$.
    \item Let $\mathcal R$ be the set of rows which are hashed with no other row from $\mathcal R_2$, which we refer to as \emph{uniquely hashed} rows. Then $\abs{\mathcal R} = \Theta(r/d\log d)$ and every one of these large entries is on a distinct row.
\end{itemize}
\end{lem}
\begin{proof}
For each $j\in[d]$, the number of expected entries in the $j$th column with absolute value at least $\tau_1$ is $C'\log d$, so by Chernoff bounds, with probability at least $1 - \exp(\Theta(\log d)) = 1 - 1/\poly(d)$, there are at most $O(\log d)$ such entries. By a union bound over the $d$ columns, this is true for all $d$ columns with probability at least $0.99$.

Similarly, for $\log_2\tau_2\leq k\leq \log_2\tau_1$, we have by Lemma \ref{lem:chernoff-power-law-levels} that 
\[
    \Pr\parens*{\norm{(\bfA\bfe_j)_{(k)}}_0 = \Theta(n/2^{kp})} \geq 1 - \exp\parens*{\Theta(n/2^{kp})}.
\]
Then summing over $k$, we have that
\begin{align*}
    \Pr\parens*{\bigcap_{k=\log_2\tau_2}^{\log_2\tau_1} \braces*{\norm{(\bfA\bfe_j)_{(k)}}_0 = \Theta(n/2^{kp})}} &\geq 1 - \sum_{k=\log_2\tau_2}^{\log_2\tau_1}\exp(\Theta(n/2^{kp})) \\
    &\geq 1 - \exp(\Theta(C'\log d)) = 1 - \frac1{\poly(d)}.
\end{align*}
By a union bound over the $d$ columns, every large entry level set of every column $\bfA\bfe_j$ has the expected number of elements, up to constant factors, simultaneously with probability at least $99/100$. Conditioned on this event, $\abs{\mathcal R_1} = O(d\log d)$. 

Note that across the $d$ columns, there are $\Theta(r/d\log d)$ rows corresponding to level sets $k$ for $\log_2 \tau_2\leq k\leq \log_2 \tau_1$. Then with $O(r)$ hash buckets, each pair of rows from $\mathcal R_1\times\mathcal R_2$ is hashed to a separate bucket with probability $O(1/\abs{\mathcal R_1\times\mathcal R_2}) = O(1/r)$, so every row in $\mathcal R_1$ is uniquely hashed with probability $0.99$ by a union bound. Furthermore, with $O(r) = \omega(r/d\log d)$ hash buckets, we have by Lemma \ref{lem:concentration-unique-hashing} that for each level set $\bfv_{(k)}$, half of the $\Theta(n/2^{kp})$ rows in the $k$th level set get hashed to a bucket with no other row from $\mathcal R$ with probability at least 
\[
    1 - 2\exp\parens*{\frac{(1/2)^2}{12}\Theta(n/2^{kp})} = 1 - \exp\parens*{\Theta(n/2^{kp})}.
\]
Then again by a union bound over the level sets and columns, every large entry level set of every column has at least half of their rows hashed with no other row from $\mathcal R$, simultaneously with probability at least $99/100$. 

The probability that any two of the large entries lie on the same row is $O(\abs{\mathcal R_2}^2/n) = o(1)$. Then by a union bound, the total success probability for the entire lemma is at least $0.95$. 
\end{proof}

We apply the lemma above to show that when we write $\bfA = \bfA^H + \bfA^L$ as in Definition \ref{def:AH-AL}, then $\norm{\bfS\bfA^H\bfx}_1 = \Omega(n\log(r/d))\norm{\bfx}_1$ for all $\bfx\in\mathbb R^d$ when we choose $T = nd^2\log d/r$. 

\begin{lem}[Mass of large entries]\label{lem:p=1-large-entry-lower-bound}
Let Let $\bfA = \bfA^H + \bfA^L$ as in Definition \ref{def:AH-AL} with $T = nd^2\log d/r$. Let $\bfS$ be a \textsf{CountSketch} matrix with $r$ rows. Then with probability at least $0.95$, 
\[
    \norm{\bfS\bfA^H\bfx}_1\geq \norm{\bfS\bfB'\bfx}_1 \geq \Omega(\norm{\bfA^H\bfx}_1) \geq \Omega(n\log(r/d^2\log d))\norm{\bfx}_1
\]
for all $\bfx\in\mathbb R^d$, where $\bfB'$ is the subset of uniquely hashed rows of $\bfA^H$ given by Lemma \ref{lem:p=1-unique-hashing-rows}. 
\end{lem}
\begin{proof}
Let $\bfB'$ be the subset of rows of $\bfA^H$ given by Lemma \ref{lem:p=1-unique-hashing-rows} that are hashed to locations without any other rows of $\bfA^H$. Recall also $\tau_1$ and $\tau_2$ from the lemma. 

We first have that $\norm{\bfS\bfB'\bfx}_1 = \Omega(\norm{\bfA^H\bfx}_1)$ since the rows containing entries larger than $\tau_1$ are perfectly hashed, while rows containing entries between $\tau_2$ and $\tau_1$ are preserved up to constant factors. 

Let $\bfB' = \bfB_{>T}' + \bfB_{\leq T}'$ where $\bfB'_{>T}$ contains the entries of $\bfB'$ that have absolute value greater than $T$ and $\bfB'_{\leq T}$ contains the rest of the entries. Note then that $\bfB'_{>T}$ has at most one nonzero entry per row, and $\bfB'_{\leq T}$ has at most $O(d\cdot r/d\log d) = O(r/\log d)$ nonzero entries and thus by Lemma \ref{lem:power-law-union-bound}, $\norm{\bfB'_{\leq T}}_\infty\leq O(r)$ with probability at least $0.99$. We condition on this event. Then for all $\bfx$, 
\begin{align*}
    \norm{\bfS\bfA^H\bfx}_1 &\geq \norm{\bfS\bfB'\bfx}_1 \\
    &\geq \norm{\bfS\bfB_{> T}'\bfx}_1 - \norm{\bfS\bfB_{\leq T}'\bfx}_1 \\
    &= \sum_{j=1}^d \abs{\bfx_j}\norm{\bfB_{>T}'\bfe_j}_1 - \norm{\bfS\bfB_{\leq T}'\bfx}_1 && \text{Since the $\bfB_{>T}'\bfe_j$ have disjoint support} \\
    &\geq \sum_{j=1}^d \abs{\bfx_j}\sum_{k=\log_2 \tau_2}^{\log_2 \tau_1}2^k\Theta(n/2^k) - \norm{\bfB_{\leq T}'\bfx}_1 && \text{Lemmas \ref{lem:p=1-unique-hashing-rows} and \ref{lem:no-expansion}} \\
    &= \Omega(n(\log_2 \tau_1 - \log_2 \tau_2))\norm{\bfx}_1 - O(r)\norm{\bfB_{\leq T}'}_\infty\norm{\bfx}_1 && \text{H\"older's inequality} \\
    &= \Omega(n\log (r/d^2\log d))\norm{\bfx}_1 - O(r^2)\norm{\bfx}_1 \\
    &= \Omega(n\log (r/d^2\log d))\norm{\bfx}_1
\end{align*}
as desired. 
\end{proof}

The last thing we need to bound is the mass contribution of the rows of $\bfA^L$ that are hashed together with the uniquely hashed rows of $\bfA^H$. We first bound the columns of the matrix $\bfS'\bfA^L$, where $\bfS'$ is a subset of hash buckets. 

\begin{lem}\label{lem:0<p<2:small-hash-bucket-columns}
Let $\mathcal D$ be a symmetric power law distribution with index $p\in(0, 2)$ with cdf $F$. Let $T\coloneqq F^{-1}(1-r/nd^2\log d) = \Theta((nd^2\log d/r)^{1/p})$, $r'<r$. Let $\bfS'$ be a subset of $r'$ rows of a $r\times n$ \textsf{CountSketch} matrix $\bfS$. Let $\bfC\sim \mathsf{trunc}_T(\mathcal D)^{n\times d}$. Then for each $j\in [d]$,
\[
    \Pr\parens*{\norm{\bfS'\bfC\bfe_j}_1 \leq O\parens*{r' + \lambda\sqrt{r'}}(d^2\log d)^{1/p - 1/2}(n/r)^{1/p}} \geq 1 - \frac1\lambda.
\]
\end{lem}
\begin{proof}
The proof is just a second moment bound and is deferred to Appendix \ref{section:appendix:subspace_embeddings_random}.
\end{proof}

Given the above bounds, the rest of the proof is just H\"older's inequality.

\begin{lem}\label{lem:p=1-small-SC1x}
Let $r \geq C (d\log d)^2$ for a large enough constant $C$. Let $\bfA = \bfA^H + \bfA^L$ as in Definition \ref{def:AH-AL} with $T = (nd^2\log d/r)^{1/p}$. Let $\bfS$ be a \textsf{CountSketch} matrix with $r$ rows. Let $\bfA^L = \bfC_1 + \bfC_2$, where $\bfC_1$ is the submatrix formed by the rows of $\bfA^L$ that are hashed together with the uniquely hashed rows of $\bfA^H$ by $\bfS$ (c.f.\ Lemma \ref{lem:p=1-unique-hashing-rows}), and $\bfC_2$ is the submatrix formed by the rest of the rows. Then with probability at least $0.99$, for all $\bfx\in\mathbb R^d$,
\[
    \norm{\bfS\bfC_1\bfx}_1 \leq O\parens*{\frac{1}{\sqrt{\log d}}\frac{n^{1/p}}{(r/d^2\log d)^{1/p-1}}}\norm{\bfx}_1.
\]
\end{lem}
\begin{proof}
Let $\bfS'$ be the submatrix of $\bfS$ formed by the set of $r'$ uniquely hashed rows from Lemma \ref{lem:p=1-unique-hashing-rows}, with $r' = O(r/d\log d)$. Setting $\lambda = 100 d$, we have that 
\[
    r' + \lambda\sqrt{r'} = O(r') = O\parens*{\frac{r}{d\log d}}
\]
so
\begin{align*}
    \Pr\parens*{\norm{\bfS\bfC_1\bfe_j}_1 \leq O\parens*{r'}\parens{d^2\log d}^{1/p-1/2}\parens*{\frac{n}{r}}^{1/p}} &= \Pr\parens*{\norm{\bfS'\bfA^L\bfe_j}_1\leq O\parens*{\frac{1}{\sqrt{\log d}}\frac{n^{1/p}}{(r/d^2\log d)^{1/p-1}}}} \\
    &\geq 1 - \frac1{100d}.
\end{align*}
By a union bound over the $d$ columns, this is true for all $k\in[d]$ with probability at least $0.99$. Conditioned on this event, we have by the triangle inequality that
\[
    \norm{\bfS\bfC_1\bfx}_1 \leq \sum_{k=1}^d \abs{\bfx_k} \sum_{i\in S}Y_{i,k} \leq O\parens*{\frac{1}{\sqrt{\log d}}\frac{n^{1/p}}{(r/d^2\log d)^{1/p-1}}}\norm{\bfx}_1
\]
as desired. 
\end{proof}

With the above lemmas in place, we prove Theorem \ref{thm:sketch-cauchy}.

\begin{proof}[Proof of Theorem \ref{thm:sketch-cauchy}]
The ``no dilation'' bound is just Lemma \ref{lem:no-expansion}. We thus focus on the ``no contraction'' bound. 

We condition on the results of Lemmas \ref{lem:p=1-small-entries}, \ref{lem:p=1-large-entry-lower-bound}, and \ref{lem:p=1-small-SC1x}. By a union bound, these all hold simultaneously with probability at least $0.9$. Then, we have for all $\bfx$ that
\[
    \frac1{\kappa} \geq \frac{\norm{\bfS\bfA\bfx}_1}{\norm{\bfA\bfx}_1}\geq \frac{\norm{\bfS\bfA^H\bfx}_1 - \norm{\bfS\bfC_1\bfx}_1}{\norm{\bfA^H\bfx}_1 + \norm{\bfA^L\bfx}_1} \geq \frac{\Omega(\norm{\bfA^H\bfx}_1 + n\log (r/d^2\log d))\norm{\bfx}_1}{O(\norm{\bfA^H\bfx}_1 + n\log(nd^2\log d/r))\norm{\bfx}_1} \geq \Omega\parens*{\frac{\log (r/d^2\log d)}{\log n}}.\qedhere
\]
\end{proof}

\subsection{Algorithms for \texorpdfstring{$p\in(1,2)$}{p in (1,2)}}

For power law distributions with index $p\in (1,2)$, we need different algorithms based on the parameter regime: when $n$ is rather large, then the distribution looks relatively flat so that sampling is approximately optimal, while when $n$ is rather small, then the variance is large enough so that \textsf{CountSketch} helps capture and preserve large values that make up a significant fraction of the mass. 

\subsubsection{Large \texorpdfstring{$n$}{n}: sampling}
When $n$ is large, we shall see that by concentration, sampling alone will give us nearly tight distortion bounds. 

We first prove concentration in the upper tail.

\begin{lem}[Upper tail concentration]\label{lem:sampling-upper-bound}
Let $\mathcal D$ be a power law distribution with index $p\in (1,2)$ and let $\bfA\sim \mathcal D^{n\times d}$. Then,
\[
    \Pr\braces*{\norm{\bfA\bfx}_1\leq  O\parens*{1+\frac{d^{1/p}\log d}{n^{1-1/p}}}\norm{\bfx}_1 n, \forall \bfx\in\mathbb R^d}\geq 0.99. 
\]
\end{lem}
\begin{proof}
By a union bound, $\norm{\bfA}_\infty\leq B = O((nd)^{1/p})$ with probability at least $0.999$. Conditioned on this event, each entry of $\bfA$ is distributed as $\mathsf{trunc}_B(\mathcal D)$. Note that for a random variable $X\sim \mathsf{trunc}_B(\mathcal D)$, we have by Lemma \ref{lem:truncated-p-moments} that
\begin{eqn*}
    \E\abs{X} &= \Theta(1) \\
    \E\abs{X}^2 &= \Theta(B^{2-p})
\end{eqn*}
Now let $\bfv\sim \mathsf{trunc}_B(\mathcal D)^n$. By the upper tail Bernstein bound,
\[
    \Pr\braces*{\norm{\bfv}_1-n\E\abs{X}\geq \lambda}\leq \exp\parens*{-\frac12\frac{\lambda^2}{B^{2-p}n + B\lambda/3}}.
\]
Then with $\lambda = \kappa n$ for
\[
    \kappa = \frac{d^{1/p}\log d}{n^{1-1/p}},
\]
we have
\[
    \lambda = \frac{d^{1/p}\log d}{n^{1-1/p}}n =  (nd)^{1/p}\log d\geq B\log d
\]
and
\[
    \lambda^2 = \frac{d^{2/p}\log^2 d}{n^{2-2/p}}n^2 = (nd)^{2/p}\log^2 d = (nd)^{\frac1p(2-p)}(nd)\log^2 d\geq B^{2-p}n\log d
\]
so we have that with probability at least $1-1/\poly(d)$,
\[
    \norm{\bfv}_1-n\E\abs{X}\leq \lambda \implies \norm{\bfv}_1\leq n\E\abs{X} + \kappa n = \Theta\parens*{1 + \frac{d^{1/p}\log d}{n^{1-1/p}}}n.
\]
By a union bound over the $d$ columns of $\bfA$, this holds simultaneously for all columns of $\bfA$ with probability at least $1-1/\poly(d)$. We condition on this event. It then follows that for every $\bfx\in\mathbb R^d$,
\[
    \norm{\bfA\bfx}\leq \norm{\bfx}_1\max_{j=1}^d \norm{\bfA\bfe_j}_1\leq \Theta\parens*{1 + \frac{d^{1/p}\log d}{n^{1-1/p}}}n. \qedhere
\]
\end{proof}

To prove concentration in the lower tail, we first need the following lemma.
\begin{lem}\label{lem:1<p<2-constant-p-stability}
Let $\mathcal D$ be a symmetric power law distribution with index $p\in (1,2)$. Let $\bfx\in\mathbb R^d$. Then
\begin{equation}\label{eqn:1<p<2-constant-p-stability}
    \Pr_{\bfv\sim\mathcal D^d}\parens*{\abs{\angle{\bfv,\bfx}} \geq \Omega(\norm{\bfx}_p)} = \Omega(1).
\end{equation}
\end{lem}
\begin{proof}
The proof is distracting from this discussion and is deferred to Appendix \ref{section:appendix:subspace_embeddings_random}.
\end{proof}

\begin{lem}[Lower tail concentration]\label{lem:sampling-lower-bound}
Let $\mathcal D$ be a symmetric power law distribution with index $p\in (1,2)$. Let $n\geq d\log d$. Then,
\[
    \Pr_{\bfA\sim\mathcal D^{n\times d}}\braces*{\norm{\bfA\bfx}_1\geq  \Omega(n\norm{\bfx}_p), \forall \bfx\in\mathbb R^d}\geq 0.99. 
\]
\end{lem}
\begin{proof}
Let $\bfx\in\mathbb R^d$ with $\norm{\bfx}_1 = 1$ and let $\bfv\sim\mathcal D^d$. Then by Lemma \ref{lem:1<p<2-constant-p-stability},
\[
    \Pr_{\bfv\sim\mathcal D^d}\parens*{\abs{\angle{\bfv,\bfx}} \geq \Omega(\norm{\bfx}_p)} = \Omega(1).
\]
Then by Chernoff bounds, at least $\Omega(n)$ of the $n$ rows of $\bfA\bfx$ are at least $\Omega(\norm{\bfx}_p)$ with probability at least $1-\exp(n) = 1 - \exp(d\log d)$. We conclude by a standard net argument. 
\end{proof}

We put the above two parts together for a sketching algorithm based on sampling.

\begin{thm}\label{thm:1<p<2-large-n}
Let $n\geq d\log d$ and let $\bfA$ be drawn as an $n\times d$ matrix of i.i.d.\ draws from a $p$-stable distribution. Let
\begin{eqn*}
    \kappa_n &= \Theta\parens*{1+\frac{d^{1/p}\log d}{n^{1-1/p}}} \\
    \kappa_r &= \Theta\parens*{1+\frac{d^{1/p}\log d}{r^{1-1/p}}}
\end{eqn*}
be the distortion upper bound from Lemma \ref{lem:sampling-upper-bound} for when the number of rows is $n$ and $r$, respectively. Let $\bfS\in\mathbb R^{r\times n}$ be the matrix that samples $r$ rows of $\bfA$, and then scales by $\kappa_n d^{1-1/p} n/r$. Then,
\[
    \Pr\braces*{\norm{\bfA\bfx}_1\leq \norm{\bfS\bfA\bfx}_1 \leq \kappa_n \kappa_r d^{2(1-1/p)}\norm{\bfA\bfx}_1, \forall \bfx\in\mathbb R^d}\geq 0.9
\]

In particular, if
\[
    n^{1-1/p}\geq d^{1/p}\log d\iff n\geq d^{\frac1{p-1}}\log^{\frac{p}{p-1}} d,
\]
then
\[
    \Pr\braces*{\norm{\bfA\bfx}_1\leq \norm{\bfS\bfA\bfx}_1 \leq O\parens*{1 + \frac{d^{1/p}}{(r/d^2)^{1-1/p}}}, \forall \bfx\in\mathbb R^d}\geq 0.9.
\]
\end{thm}
\begin{proof}
By applying lemmas \ref{lem:sampling-upper-bound} and \ref{lem:sampling-lower-bound}, we have that for all $\bfx$,
\[
    \Omega\parens*{1}\norm{\bfx}_p n\leq \norm{\bfA\bfx}_1\leq \kappa_n \norm{\bfx}_1 n.
\]
Furthermore, we can apply the lemmas to $\bfS\bfA$ as well, which gives us 
\[
    \Omega(1)\kappa_n d^{1-1/p}\norm{\bfx}_p n\leq \norm{\bfS\bfA\bfx}_1\leq \kappa_r\kappa_nd^{1-1/p}\norm{\bfx}_1 n.
\]
By H\"older's inequality, we have that for all $\bfx\in\mathbb R^d$,
\[
    \norm{\bfx}_p\leq \norm{\bfx}_1\leq d^{1-1/p}\norm{\bfx}_p.
\]
Thus,
\[
    \norm{\bfA\bfx}_1\leq \kappa_n\norm{\bfx}_1 n\leq \kappa_n d^{1-1/p}\norm{\bfx}_p n\leq \norm{\bfS\bfA\bfx}_1
\]
so the sketch does not underestimate norms. On the other hand,
\[
    \norm{\bfS\bfA\bfx}_1\leq \kappa_r\kappa_n d^{1-1/p}\norm{\bfx}_1n\leq \kappa_r\kappa_n d^{2(1-1/p)}\norm{\bfx}_p n\leq \kappa_r\kappa_n d^{2(1-1/p)}\norm{\bfA\bfx}_1
\]
so the sketch does not overestimate norms by more than $\kappa_r\kappa_n d^{2(1-1/p)}$, as claimed. 
\end{proof}

\subsubsection{Small \texorpdfstring{$n$}{n}: \textsf{CountSketch}}
In the previous section, we have handled the case when $n^{1-1/p}\geq d^{1/p}\log d$. On the other hand, when $n^{1-1/p}\leq d^{1/p}\log d$ we will instead use \textsf{CountSketch} to hash the largest entries of each column of $\bfA$. These entries are of size around $n^{1/p}$, while the entries of vectors with size smaller than this have mass at most $n$. Thus, we approximate the mass up to a factor of
\[
    \frac{n}{n^{1/p}} = n^{1-1/p}\leq d^{1/p}\log d,
\]
which is roughly what we are shooting for.

\begin{thm}\label{thm:p-1-2-countsketch}
Let $n^{1-1/p}\leq d^{1/p}\log d$. Let $\bfS$ be drawn as a \textsf{CountSketch} matrix with $r$ rows. Then
\[
    \Pr_{\bfS, \bfA}\parens*{\Omega\parens*{\frac1\kappa}\norm{\bfA\bfx}_1\leq \norm{\bfS\bfA\bfx}_1\leq \norm{\bfA\bfx}_1, \forall \bfx\in\mathbb R^d}\geq \frac{99}{100}.
\]
where
\[
    \kappa = O\parens*{\parens*{\frac{n}{(r/d^2\log d)}}^{1-1/p}} = O\parens*{\frac{d^{1/p}\log d}{(r/d^2\log d)^{1-1/p}}}.
\]
\end{thm}
\begin{proof}
The distortion upper bound again is just Lemma \ref{lem:no-expansion}. The distortion lower bound argument is similar to the one presented for the cases of $p\in(0,1]$ and thus is deferred to Appendix \ref{section:appendix:subspace_embeddings_random}.
\end{proof}

\subsection{Algorithms for \texorpdfstring{$p\geq 2$}{p >= 2}}\label{section:p>=2}
When $p\geq 2$, we show that any $m\times d$ i.i.d.\ matrix with $m \geq \poly(d)$ has with constant probability, $\norm{\bfA\bfx}_1 = \Theta(m\norm{\bfx_2})$ for all $\bfx\in\mathbb R^d$. This shows that a uniform sampling matrix with $\poly(d)$ rows works as a sketch. The following result shows this in expectation.
\begin{lem}\label{lem:p>=2-expectation}
Let $p\geq 2$ and let $\mathcal D$ be a symmetric power law with index $p$. Let $\bfx\in\mathbb R^d$. Then,
\[
    \E_{\bfv\sim\mathcal D^d}\abs{\angle{\bfv,\bfx}} = \Theta(\norm{\bfx}_2)
\]
\end{lem}
\begin{proof}
The proof is standard and is deferred to Appendix \ref{section:appendix:subspace_embeddings_random}.
\end{proof}

Our strategy then is to show that conditioned on every entry of $\bfv\sim\mathcal D^d$ being smaller than some large value $B \geq \poly(d)$, the expectation remains approximately unchanged. We then use this in a Bernstein bound to argue the result with high enough probability to union bound over a net. 

\begin{lem}\label{lem:p>=2-conditioning-expectation-same}
Let $p\geq 2$ and let $\mathcal D$ be a symmetric power law with index $p$. Let $\eps\in(0,1/4)$ and let $B > \max\{\eps^{-1},(d\sqrt{d}/\eps)^{1/p}\}$. Let $\bfv\sim\mathcal D^d$ and fix a vector $\bfx\in\mathbb R^d$. Define the events
\begin{align*}
    \mathcal E_i &\coloneqq \braces*{\abs{\bfv_i}\leq B} \\
    \mathcal E &\coloneqq \bigcap_{i=1}^d \mathcal E_i
\end{align*}
Then for $B$ large enough, 
\[
    \E\parens*{\abs{\angle{\bfv,\bfx}} \mid \mathcal E}\geq (1-O(\eps))\E\abs{\angle{\bfv,\bfx}}.
\]
\end{lem}
\begin{proof}
The proof is standard and is deferred to Appendix \ref{section:appendix:subspace_embeddings_random}.
\end{proof}

The following lemma implements the Bernstein bound and applies a standard net argument. 

\begin{lem}\label{lem:p>=2-2-norm-subspace}
Let $p\geq 2$ and let $\mathcal D$ be a symmetric power law with index $p$. Let $\bfA\sim\mathcal D^{m\times d}$ with $m\geq \Theta(\max\braces*{\eps^{-p}, (d^{3/2+1/p}\eps^{-1}\log \eps^{-1})^{p/(p-1)}})$. For $\bfx\in\mathbb R^d$, let $\mu_\bfx\coloneqq \E_{\bfv\sim\mathcal D^d}\angle*{\bfv,\bfx}$. Then,
\[
    \Pr\parens*{\norm{\bfA\bfx}_1 = (1\pm O(\eps)) m\mu_\bfx, \forall \bfx\in\mathbb R^d} \geq 0.95.
\]
\end{lem}
\begin{proof}
Note that with probability at least $0.99$, $\norm{\bfA}_\infty = O((md)^{1/p})$. Let this event be $\mathcal E$. Then conditioned on $\mathcal E$, $\bfA$ is distributed as an i.i.d.\ matrix drawn from $\mathcal D'$, where $\mathcal D'$ is the truncation of $\mathcal D$ at
\[
    B = O((md)^{1/p}) \geq \max\braces*{\eps^{-1}, (d\sqrt{d}/\eps)^{1/p}}
\]
where the bound on $B$ follows by our choice of $m$. 

\paragraph{High probability bounds.}

Now fix $\bfx\in\mathbb R^d$. By Lemmas \ref{lem:p>=2-expectation} and \ref{lem:p>=2-conditioning-expectation-same},
\begin{equation}\label{eqn:p>=2-expectation}
    \mu\coloneqq \E\parens*{\norm{\bfA\bfx}_1\mid\mathcal E} = \sum_{i=1}^m \E\parens*{\abs{\bfe_i^\top\bfA\bfx}\mid\mathcal E} = (1\pm\eps)\sum_{i=1}^m  \E\abs{\bfe_i^\top\bfA\bfx} = (1\pm\eps)m\mu_\bfx = \Theta(m\norm{\bfx}_2)
\end{equation}
and
\[
    \sigma^2\coloneqq \Var(\norm{\bfA\bfx}_1\mid \mathcal E) = \sum_{i=1}^m \Var(\abs{\bfe_i^\top\bfA\bfx}\mid \mathcal E) \leq O(m\norm{\bfx}_2^2). 
\]
Then by Bernstein bounds, we have that
\begin{align*}
    \Pr\parens*{\abs{\norm{\bfA\bfx}_1 - \mu}\geq \eps \mu\mid\mathcal E}\leq 2\exp\parens*{-\frac12\frac{(\eps\mu)^2}{\sigma^2 + \eps\mu B/3}}\leq \exp\parens*{-\Theta(d\log \eps^{-1})}
\end{align*}
where the last inequality follows by our choice of $m$. Then chaining together with Equation \ref{eqn:p>=2-expectation},
\begin{equation}\label{eqn:1+eps-approx-p>=2}
    \Pr\parens*{\norm{\bfA\bfx}_1 = (1\pm O(\eps))m \mu_\bfx \mid \mathcal E}\leq \exp\parens*{-\Theta(d\log\eps^{-1})}.
\end{equation}

\paragraph{Net argument.}

We now proceed by a standard net argument. Recall the set $\mathcal S$ and the $\eps$-net $\mathcal N$ as given in Lemma \ref{lem:net-argument-ingredients}. Now conditioned on $\mathcal E$, we have by a union bound that $\norm{\bfA\bfx}_1 = (1\pm O(\eps))m\mu_\bfx$ for every $\bfA\bfx\in\mathcal N$, with probability at least $0.99$. We condition on this event as well. Now for any $\bfy\in\mathcal S$, write $\bfy = \sum_{i=0}^\infty \bfy^{(i)}$ as given in Lemma \ref{lem:net-argument-ingredients}. Then,
\[
    \norm{\bfS\bfA\bfx}_1 = \Norm{\bfS\sum_{i=0}^\infty\bfy^{(i)}}_1\leq \sum_{i=0}^\infty\Norm{\bfS\bfy^{(i)}}_1\leq (1+O(\eps))\sum_{i=0}^\infty\Norm{\bfy^{(i)}}_1\leq (1+O(\eps))\sum_{i=0}^\infty\eps^i\leq 1+O(\eps). 
\]
We conclude by homogeneity. 
\end{proof}

Given the above lemma, our subspace embedding follows simply from uniform sampling and rescaling.

\begin{thm}\label{thm:p>=2-upper-bound}
Let $p\geq 2$ and let $\mathcal D$ be a symmetric power law with index $p$. Let $\eps\in(0,1/2)$, let
\[
    r = \Theta(\max\braces*{\eps^{-p}, (d^{3/2+1/p}\eps^{-1}\log \eps^{-1})^{p/(p-1)}})
\]
and let $\bfS\in\mathbb R^{r\times n}$ be a matrix that uniform samples $r$ rows and scales by $n/r$. Then,
\[
    \Pr\parens*{\norm{\bfS\bfA\bfx}_1 = (1\pm O(\eps))\norm{\bfA\bfx}_1, \forall\bfx\in\mathbb R^d}\geq 0.9
\]
\end{thm}
\begin{proof}
By Lemma \ref{lem:p>=2-2-norm-subspace}, we have with probability at least $0.95$ that for all $\bfx\in\mathbb R^d$,
\[
    \norm{\bfA\bfx}_1 = (1\pm O(\eps))n\mu_\bfx
\]
and with probability at least $0.95$ that for all $\bfx\in\mathbb R^d$,
\[
    \norm{\bfS\bfA\bfx}_1 = \frac{n}{r}\cdot (1\pm O(\eps))r\mu_\bfx = (1\pm O(\eps))n\mu_\bfx.
\]
Combining these two bounds, we have that with probability at least $0.9$, for all $\bfx\in\mathbb R^d$,
\[
    \norm{\bfS\bfA\bfx}_1 = (1 \pm O(\eps)) \norm{\bfA\bfx}_1.\qedhere
\]
\end{proof}

\subsection{Lower bound}

In this section, we work towards proving a lower bound for a general class of random matrices with each column drawn i.i.d.\ from a different distribution. When specialized to our i.i.d. matrices from the above, we obtain nearly tight bounds. We do not have a single general theorem, but rather a number of different possible arguments that give better bounds depending on the underlying distribution. This can be shown to approximately recover the result of \cite[Theorem 1.1]{DBLP:conf/soda/WangW19}, and additional yields new bounds, for example, the following tight results:
\begin{thm}\label{thm:lower-bound-cauchy}
Let $\log n \leq O(d)$ and let $\bfS$ be a $r\times n$ matrix such that
\[
    \Pr_{\bfA\sim\Cauchy^{n\times d}}\parens*{\norm{\bfA}_1 \leq \norm{\bfS\bfA}_1\leq \kappa\norm{\bfA}_1} \geq \frac{99}{100}
\]
Then,
\[
    \kappa = \Omega\parens*{\frac{\log n}{\log r}}
\]
\end{thm}

\begin{thm}\label{thm:lower-bound-p-stable}
Let $\bfS$ be a $r\times n$ matrix such that
\[
    \Pr_{\bfA\sim\mathcal D^{n\times d}}\parens*{\norm{\bfA}_1 \leq \norm{\bfS\bfA}_1\leq \kappa\norm{\bfA}_1} \geq \frac{99}{100}
\]
where $\mathcal D$ is a $p$-stable distribution. Then,
\[
    \kappa = \Omega\parens*{\frac{d^{1/p}}{r^{1-1/p}}}
\]
\end{thm}

\begin{dfn}\label{def:input-dist}
Let $\mathcal D_j$ for $j\in[d]$ be distributions and consider the distribution $\mathcal D_{\bfA}$ over $n\times d$ matrices $\bfA$ that draws column $j$ from $\mathcal D_j^n$. Let 
\[
    M_j\coloneqq \median_{\bfu\sim\mathcal D_j^n}\norm{\bfu}_1.
\]
We then define the distribution $\mathcal D_{\max}$ that draws entries as
\[
    \max_{j=1}^d \frac{\abs{v_j}}{M_j}, v_j\sim\mathcal D_j
\]
and let its cdf be $F_{\mathcal D_{\max}}$.
\end{dfn}

Throughout this section, let $\bfS$ be a $r\times n$ matrix such that
\[
    \Pr_{\bfA\sim\mathcal D_{\bfA}}\parens*{\norm{\bfA}_1 \leq \norm{\bfS\bfA}_1\leq \kappa\norm{\bfA}_1}\geq 1-\delta.
\]

\subsubsection{Preliminary bounds on \texorpdfstring{$\bfS$}{S}}

\begin{lem}
For every $i\in[n]$, we have that
\[
    \norm{\bfS\bfe_i}_1\leq 2\kappa\frac{1+F_{\mathcal D_{\max}}^{-1}(4\delta)}{F_{\mathcal D_{\max}}^{-1}(4\delta)}
\]
\end{lem}
\begin{proof}
Note that for every row $i$ of $\bfA$, with probability at least $4\delta$, one of the $d$ columns of $\bfe_i^\top\bfA$, say column $j\in[d]$, has absolute value at least
\[
    \frac{\abs{\bfe_i^\top\bfA\bfe_j}}{M_j}\geq F_{\mathcal D_{\max}}^{-1}(4\delta) \iff \abs{\bfe_i^\top\bfA\bfe_j}\geq F_{\mathcal D_{\max}}^{-1}(4\delta)M_j.
\]
Independently, with probability at least $1/2$, the $\ell_1$ norm of the rest of the entries of the column is at most
\[
    \sum_{i'\in[n]\setminus\{i\}}\abs{\bfe_{i'}^\top\bfA\bfe_j}\leq M_j.
\]
Then with probability $2\delta$ both of these happen simultaneously, so that
\[
    \norm{\bfA\bfe_j}_1\leq (1+F^{-1}_{\mathcal D_{\max}}(4\delta))M_j
\]
Let this event be $\mathcal E_i$. 

Now suppose for contradiction that there is some column $i\in[n]$ such that $\norm{\bfS\bfe_i}_1 > 2\kappa / F_{\mathcal D_{\max}}^{-1}(4\delta)$. We then condition on $\mathcal E_i$. Then with probability at least $1/2$, 
\begin{eqn*}
      \norm{\bfS\bfA\bfe_j}_1 &= \norm{\bfS\bfe_{i}(\bfe_{i}^\top\bfA\bfe_j) + \sum_{i'\neq i}\bfS\bfe_{i'}(\bfe_{i'}^\top\bfA\bfe_j)}_1\geq \frac12\norm{\bfS\bfe_{i}(\bfe_{i}^\top\bfA\bfe_j)}_1 \\
      &> \frac12 \parens*{2\kappa\frac{1+F_{\mathcal D_{\max}}^{-1}(4\delta)}{F_{\mathcal D_{\max}}^{-1}(4\delta)}} F_{\mathcal D_{\max}}^{-1}(4\delta)M_j = \kappa (1+F_{\mathcal D_{\max}}^{-1}(4\delta))M_j\geq \kappa \norm{\bfA\bfe_j}_1
\end{eqn*}
so $\bfS$ fails to sketch $\bfA$ with probability at least $\delta$, which is a contradiction.
\end{proof}

\begin{lem}\label{lem:num-large-cols}
Let $F_{\mathcal D_{\max}}^{-1}(4\delta) < x < 1$. Then there are at most
\[
    \frac1{1 - F_{\mathcal D_{\max}}(x)}
\]
columns of $\bfS$ with $\ell_1$ norm more than
\[
    2\kappa\frac{1+x}{x}.
\]
\end{lem}
\begin{proof}
Let
\[
    p\coloneqq \Pr_{X\sim \mathcal D_{\max}}\braces*{X\geq x} = 1 - F_{\mathcal D_{\max}}(x)
\]
and suppose for contradiction that there are more than $1/p$ columns of $\bfS$ with $\ell_1$ norm more than $2\kappa/x$. Note that for each row $i$ of these $1/p$ rows, there is a $p$ probability that one of the $d$ columns of $\bfe_i^\top\bfA$, say column $j\in[d]$, has absolute value at least
\[
     \frac{\abs{\bfe_i^\top\bfA\bfe_j}}{M_j}\geq x \iff \abs{\bfe_i^\top\bfA\bfe_j}\geq xM_j.
\]
Then the probability that one of the $1/p$ rows, say row $i$, has an entry of absolute value at least $xM_j$ is at least
\[
    1 - (1-p)^{1/p}\geq 1 - e^{-1}.
\]
Independently, with probability at least $1/2$, the $\ell_1$ norm of the rest of the entries of this column is at most
\[
    \sum_{i'\in[i]\setminus\{i\}}\abs{\bfe_{i'}^\top\bfA\bfe_j}\leq M_j.
\]
Thus with probability at least $(1 - e^{-1})/2$, both of these events happen simultaneously, so there is row $i\in[n]$ and a column $j\in[d]$ such that
\[
\begin{cases}
    \abs{\bfe_i^\top\bfA\bfe_j} &\geq xM_j \\
    \norm{\bfA\bfe_j}_1 &\leq (1+x)M_j \\
    \norm{\bfS\bfe_i} &\geq 2\kappa\frac{1+x}{x}
\end{cases}.
\]
We condition on this event. Then with probability at least $1/2$, 
\begin{eqn*}
      \norm{\bfS\bfA\bfe_j}_1 &= \norm{\bfS\bfe_{i}(\bfe_{i}^\top\bfA\bfe_j) + \sum_{i'\neq i}\bfS\bfe_{i'}(\bfe_{i'}^\top\bfA\bfe_j)}_1\geq \frac12\norm{\bfS\bfe_{i}(\bfe_{i}^\top\bfA\bfe_j)}_1 \\
      &> \frac12 2\kappa\frac{1+x}{x}\cdot x M_j = \kappa (1+x)M_j\geq \kappa \norm{\bfA\bfe_j}_1
\end{eqn*}
so $\bfS$ fails to sketch $\bfA$ with probability at least $(1-e^{-1})/4 > \delta$, which is a contradiction.
\end{proof}

\subsubsection{Distortion lower bound}\label{section:general-iid-distortion-bound}

Fix any $\bfx\in\mathbb R^d$ with $\norm{\bfx}_1 = 1$ and let $\bfv = \bfA\bfx$. Following \cite{DBLP:conf/soda/WangW19}, our strategy is to bound $\norm{\bfS\bfv}_1$ from above in terms of $\kappa$, and then derive a lower bound on $\kappa$ by bounding $\norm{\bfS\bfv}_1$ below by $\norm{\bfv}_1$. 

Note that the bound of Lemma \ref{lem:num-large-cols} is useless when
\[
    \frac1{1-F_{\mathcal D_{\max}}(x)}\geq n \iff F_{\mathcal D_{\max}}(x)\geq 1 - \frac1n \iff x \geq F^{-1}_{\mathcal D_{\max}}\parens*{1 - \frac1n}.
\]
We thus set $\bfS^H$ to be the matrix formed by taking the columns of $\bfS$ with $\ell_1$ norm at least
\[
    2\kappa\frac{1 + F_{\mathcal D_{\max}}^{-1}(1-1/n)}{F_{\mathcal D_{\max}}^{-1}(1-1/n)}
\]
and $\bfS^L$ to be the columns of $\bfS$ with $\ell_1$ norm at most this, and individually bound $\bfS^H\bfv$ and $\bfS^L\bfv$.

Now consider the distribution $\mathcal D_\bfx$ with cdf $F_\bfx$ that draws its entries as $\abs{\angle{\bfx,\bfw}}$ with $\bfw\sim\prod_{j=1}^d \mathcal D_j$. By a union bound, the largest absolute value entry in $\bfv = \bfA\bfx$ is at most $M_\bfx\coloneqq F^{-1}_\bfx(1 - 1/2n)$ with  probability at least $1/2$. Let this event be
\[
    \mathcal E\coloneqq \braces*{\norm{\bfv}_\infty\leq M_\bfx}.
\]
Throughout this section, we condition on $\mathcal E$. We also define
\[
    F_{\bfx, \land}(x)\coloneqq \frac{F_\bfx(x)\mathbbm{1}(x\leq M_{\bfx})}{\Pr(\mathcal E)}
\]
to be the conditional cdf of the capped version of $\bfv$. 

\subsubsection{Bounding the high-norm columns of \texorpdfstring{$\bfS$}{S}}

We first bound $\norm{\bfS^H\bfv}_1$. We will need the following simple lemma:
\begin{lem}\label{lem:permute-expectation}
Let $\bfu,\bfv\in\mathbb R^n$ be vectors with nonnegative entries and unit $\ell_1$ norm. Let $\bfP$ be a uniformly random permutation matrix. Then,
\[
    \E_{\bfP}\angle{\bfu, \bfP\bfv} = \frac1n.
\]
\end{lem}
\begin{proof}
We have that
\[
    \E_{\bfP}\angle{\bfu, \bfP\bfv} = \sum_{i=1}^n u_i \E_{\bfP}(\bfe_i^\top\bfP\bfv) = \sum_{i=1}^n u_i \sum_{j=1}^n \frac{v_j}{n} = \frac1n.\qedhere
\]
\end{proof}

The main result for this section then is the following:

\begin{lem}
Let
\begin{eqn*}
    L_{\max} &\coloneqq F^{-1}_{\mathcal D_{\max}}\parens*{1 - \frac1n} \\
    L_{\min} &\coloneqq F_{\mathcal D_{\max}}^{-1}\parens*{4\delta}
\end{eqn*}
Then,
\[
    \Pr\braces*{\Norm{\bfS^H\frac{\bfv}{\norm{\bfv}_1}}_1\leq 400 \frac{\kappa}{n} \sum_{k = \log_2 L_{\min}}^{\log_2 L_{\max}} \frac{1+2^{k-1}}{2^{k}\parens*{1 - F_{\mathcal D_{\max}}(2^k)}}}\geq \frac{99}{100}.
\]
\end{lem}
\begin{proof}
Note that i.i.d.\ distributions are permutation invariant. We first fix the entries of $\bfv$, which fixes $\norm{\bfv}_1$, but not the permutation of the entries. Now by Lemma \ref{lem:permute-expectation}, we have
\[
    \E_{\bfP}\Norm{\bfS^H\bfP\frac{\bfv}{\norm{\bfv}_1}}_1 = \sum_{i=1}^r \norm{\bfe_i^\top\bfS^H}_1 \E_{\bfP}\abs{\angle*{\frac{\bfe^\top\bfS^H}{\norm{\bfe_i^\top\bfS^H}_1}, \bfP\frac{\bfv}{\norm{\bfv}_1}}} \leq \sum_{i=1}^r \norm{\bfe_i^\top\bfS^H}_1\frac{1}{n} = \frac{1}{n}\sum_{j=1}^n \norm{\bfS^H\bfe_j}_1.
\]
By Lemma \ref{lem:num-large-cols}, we have that for each integer $k$ between $\log_2 L_{\min}$ and $\log_2 L_{\max}$, there are at most
\[
    \frac1{1-F_{\mathcal D_{\max}}(2^k)}
\]
columns of $\bfS$ with $\ell_1$ norm more than $2\kappa (1 + 1/ 2^k)$. Thus, there are at most $(1-F_{\mathcal D_{\max}}(2^k))^{-1}$ columns with $\ell_1$ norm in $[2\kappa (1+1/2^k), 2\kappa (1+1/2^{k-1})]$. Then, summing over the bounds over these intervals, we have that
\[
    \sum_{j=1}^n \norm{\bfS^H\bfe_j}_1\leq \sum_{k = \log_2 L_{\min}}^{\log_2 L_{\max}} 2\kappa\parens*{1+\frac1{2^{k-1}}}\frac1{1 - F_{\mathcal D_{\max}}(2^k)} = \kappa \sum_{k = \log_2 L_{\min}}^{\log_2 L_{\max}} \frac{4(1+2^{k-1})}{2^{k}\parens*{1 - F_{\mathcal D_{\max}}(2^k)}}.
\]
Chaining together the inequalities gives the bound
\[
    \E_{\bfP}\Norm{\bfS^H\bfP\frac{\bfv}{\norm{\bfv}_1}}_1\leq \frac{\kappa}{n} \sum_{k = \log_2 L_{\min}}^{\log_2 L_{\max}} \frac{4(1+2^{k-1})}{2^{k}\parens*{1 - F_{\mathcal D_{\max}}(2^k)}}.
\]
We then conclude by Markov's inequality. 
\end{proof}

\subsubsection{Bounding the low-norm columns of \texorpdfstring{$\bfS$}{S}}

The idea for bounding $\norm{\bfS^L\bfv}_1$ is that different arguments are needed for different level sets of $\bfv$, depending on how ``spiky'' it is. That is, a relatively flat level $\bfv_k$ should benefit from the sign cancellations in the product $\bfe_i^\top\bfS^L\bfv_k$, while a very spiky vector such as standard basis vectors should just apply the triangle inequality and bound only the few columns of $\bfS^L$ that it touches. This idea is formalized in the following lemma.

\begin{lem}\label{lem:SL-vector-bound}
Let $\bfS\in\mathbb R^{r\times n}$ be a fixed matrix such that $\norm{\bfS\bfe_i}_1\leq 1$ for each $i\in[n]$, and let $\bfw\in\mathbb R^n$ be a vector with entries drawn i.i.d.\ from a distribution with $\E w_i = 0$ and $\sigma^2\coloneqq \E w_i^2 <\infty$. Let $\mu\coloneqq \E\abs{w_i}$. Then,
\[
    \E\norm{\bfS\bfw}_1\leq \min\braces*{\mu n, C\sigma\sqrt{rn}}
\]
for an absolute constant $C$.
\end{lem}
\begin{proof}
For the first term of the min, we can simply use the triangle inequality to obtain
\[
    \E\norm{\bfS\bfw}_1\leq \sum_{i=1}^n \norm{\bfS\bfe_i}_1\E\abs{w_i} = \mu n.
\]

For the second term, we first apply Jensen's inequality to get 
\[
    \E \norm{\bfS\bfw}_1 = \sum_{i=1}^r \E\abs{\bfe_i^\top\bfS\bfw}\leq C\sum_{i=1}^r\parens*{\sum_{j=1}^n (\bfe_i^\top\bfS\bfe_j)^2 \E w_j^2}^{1/2} = C\sigma \sum_{i=1}^r\norm{\bfe_i^\top\bfS}_2
\]
for some absolute constant $C$. We then finish by an application of Cauchy-Schwarz, switching from row-wise sums to column-wise sums, and bounding the $\ell_2$ norm by the $\ell_1$ norm:
\[
    \sum_{i=1}^r \norm{\bfe_i^\top\bfS}_2\leq \sqrt{r}\parens*{\sum_{i=1}^r \norm{\bfe_i^\top\bfS}_2^2}^{1/2} = \sqrt{r}\parens*{\sum_{j=1}^n \norm{\bfS\bfe_j}_2^2}^{1/2}\leq \sqrt{r}\parens*{\sum_{j=1}^n \norm{\bfS\bfe_j}_1^2}^{1/2}\leq \sqrt{rn}. \qedhere
\]
\end{proof}

Now for intuition, in Lemma \ref{lem:SL-vector-bound}, we roughly think of the distribution of $w_i$ as being $v_i$ if $v_i$ belongs to a level set, and $0$ otherwise. Then if $p$ is the probability of being in a given level set, the first term is roughly $pn$ while the second term is roughly $\sqrt{rpn}$, so the first bound is tighter when $p\leq r/n$ and the second bound is tighter when $p\geq r/n$. 

This yields the following:
\begin{cor}\label{cor:general-lower-bound-low-norm-bound}
let
\[
    T\coloneqq F_{\bfx,\land}^{-1}\parens*{1 - \frac{r}{n}}
\]
and write $\bfv = \bfv_{\leq T} + \bfv_{>T}$, where $\bfv_{\leq T}$ takes the value of $\bfv$ on coordinates $i\in[n]$ where $\abs{v_i}\leq T$ and $0$ otherwise, and $\bfv_{>T}$ similarly takes the coordinates $i\in[n]$ of $\bfv$ such that $\abs{v_i}>T$ and $0$ otherwise. Then, $\bfv_{\leq T}$ is drawn i.i.d.\ from a distribution with second moment
\[
    \sigma_{\leq T}^2 \coloneqq \int_0^T x^2 f_{\bfx,\land}(x)~dx
\]
while $\bfv_{>T}$ is drawn i.i.d.\ from a distribution with expected absolute value
\[
    \mu_{>T} \coloneqq \int_T^{M_{\bfx}}x f_{\bfx,\land}(x)~dx.
\]
Applying Lemma \ref{lem:SL-vector-bound}, we obtain the bound
\[
    \E\norm{\bfS^L\bfv}_1\leq \E\norm{\bfS^L\bfv_{\leq T}}_1 + \E\norm{\bfS^L\bfv_{>T}}_1\leq C\kappa\frac{1+F^{-1}_{\mathcal D_{\max}}(1-1/n)}{F^{-1}_{\mathcal D_{\max}}(1-1/n)}(\sigma_{\leq T}\sqrt{rn} + \mu_{>T}n).
\]
\end{cor}

\subsubsection{Lower bounds for sketching i.i.d.\ \texorpdfstring{$p$}{p}-stable matrices}
We apply Corollary \ref{cor:general-lower-bound-low-norm-bound} to prove Theorem \ref{thm:lower-bound-cauchy}:

\begin{proof}[Proof of Theorem \ref{thm:lower-bound-p-stable}]
Note that when $\bfA$ is drawn as fully i.i.d.\ Cauchy variables, then
\[
    F^{-1}_{\mathcal D_{\max}}(1-1/n) = O\parens*{\frac{nd}{n\log n}} = O(d/\log n).
\]
We now apply Corollary \ref{cor:general-lower-bound-low-norm-bound} with $\bfv = \bfA\bfe_1$, a Cauchy vector. Then, $T = \Theta(n/r)$, $\sigma_{\leq T}^2 = \Theta(n/r)$, and $\mu_{>T} = \Theta(\log r)$ which yields
\[
    \E\norm{\bfS^L\bfv}_1\leq O\parens*{\kappa\frac{1+d/\log n}{d/\log n}\bracks*{\sqrt{\frac{n}{r}}\sqrt{rn} + n\log r}} = O\parens*{\kappa n\log r}.
\]
Then with constant probability, we have
\[
    \Omega(n\log n)\leq \norm{\bfv}_1\leq \norm{\bfS^L\bfv}_1\leq  O\parens*{\kappa n\log r}
\]
and thus
\[
    \kappa = \Omega\parens*{\frac{\log n}{\log r}},
\]
as desired.
\end{proof}

When we have a column drawn i.i.d.\ from a $p$-stable distribution, we have an alternative bound:
\begin{lem}\label{lem:p-stable-distortion-lower-bound-helper}
Let $\bfv$ be drawn i.i.d.\ from a $p$-stable distribution for $p\in (1,2)$. If $\bfv$ is in the column space of $\bfS$, then
\[
    2\kappa\frac{1+F_{\mathcal D_{\max}}^{-1}(1-1/n)}{F_{\mathcal D_{\max}}^{-1}(1-1/n)}r^{1-1/p}n^{1/p} = \Omega(n).
\]
\end{lem}
\begin{proof}
Then, we have that
\begin{align*}
    \Omega(n) &\leq \norm{\bfv}_1 \leq \norm{\bfS^L\bfv}_1 = \sum_{i=1}^r \abs{\bfe_i^\top \bfS^L\bfv} = \sum_{i=1}^r \norm{\bfe_i^\top \bfS^L}_p \abs{\mathcal S_{i}}
\end{align*}
By linearity of expectation, the above sum has expectation $\sum_{i=1}^r O\parens*{\norm{\bfe_i^\top \bfS^L}_p}$, and thus is at most a constant times this with probability at least $99/100$ by a Markov bound. Then, we proceed by bounding
\begin{align*}
    \sum_{i=1}^r \norm{\bfe_i^\top \bfS^L}_p &\leq r^{1-1/p}\parens*{\sum_{i=1}^r \norm{\bfe_i^\top\bfS^L}_p^p}^{1/p} \\
    &= r^{1-1/p}\parens*{\sum_{j=1}^n \norm{\bfS^L\bfe_j}_p^p}^{1/p} \leq r^{1-1/p}\parens*{\sum_{j=1}^n \norm{\bfS^L\bfe_j}_1^p}^{1/p} \\
    &\leq r^{1-1/p}\parens*{n\parens*{2\kappa\frac{1+F_{\mathcal D_{\max}}^{-1}(1-1/n)}{F_{\mathcal D_{\max}}^{-1}(1-1/n)}}^p}^{1/p} = 2\kappa\frac{1+F_{\mathcal D_{\max}}^{-1}(1-1/n)}{F_{\mathcal D_{\max}}^{-1}(1-1/n)}r^{1-1/p}n^{1/p}.
\end{align*}
\end{proof}

This gives a proof of Theorem \ref{thm:lower-bound-p-stable}.
\begin{proof}[Proof of Theorem \ref{thm:lower-bound-p-stable}]
When $\bfA$ is drawn as fully i.i.d.\ $p$-stable variables, then
\[
    F_{\mathcal D_{\max}}^{-1}(1-1/n) = \Theta\parens*{\frac{(nd)^{1/p}}{n}}
\]
so by Lemma \ref{lem:p-stable-distortion-lower-bound-helper}, the distortion bound from these columns is
\[
    \Omega(n)\leq \kappa \frac{n}{(nd)^{1/p}}r^{1-1/p}n^{1/p}\iff \kappa \geq \Omega\parens*{\frac{d^{1/p}}{r^{1-1/p}}}.
\]
\end{proof}

%% file: preliminaries_appendix.tex
\section{Missing proofs from Section \ref{sec:prelim}}\label{sec:appendix:prelim}

\begin{proof}[Proof of Lemma \ref{lem:concentration-unique-hashing}]
For each $i\in S$, sample $i$ with probability $p$ and place the result in a uniformly random hash bucket in $[r]$ if it was sampled. Let $\mathcal E_i$ denote the event where $i$ is sampled and is hashed to a bucket with no other members of $T$. Let $C_1, C_2,\dots, C_{\abs{S}}$ denote the sequence of these independent random choices and let $f(C_1, C_2, \dots, C_s)$ denote the number of hash buckets in $[r]$ that contains members $i\in S$ satisfying $\mathcal E_i$ at the end of the sampling and hashing process. Note that $f$ is $1$-Lipschitz, and that
\[
    \E f(C_1, C_2, \dots, C_{\abs{S}}) = \sum_{i\in S}\Pr\parens*{\mathcal E_i} = \abs{S}p\parens*{1 - \frac{p}{r}}^{\abs{T}}\geq p\abs{S}\parens*{1 - \frac{p\abs{T}}{r}}\geq (1-\eps)p\abs{S}.
\]

Now consider the Doob martingale
\[
    Z_k\coloneqq \E\bracks*{f_q(C_1, C_2, \dots, C_{\abs{S}})\mid C_1, C_2, \dots, C_k}. 
\]
Note that the increments $Z_k - Z_{k-1}$ conditioned on $C_1, C_2, \dots, C_{k-1}$ is simply the indicator variable of whether on choice $C_k$ we sampled an entry and placed it in a new bucket or not. Then $Z_k - Z_{k-1} = 1$ with probability at most $p$ and thus $\E_{k-1}(Z_k-Z_{k-1})^2\leq p$. Then by Freedman's inequality \cite{freedman1975tail},
\[
    \Pr\parens*{\abs{Z_{\abs{S}} - Z_0}\geq \eps Z_0}\leq 2\exp\parens*{-\frac12\frac{(\eps(1-\eps)p\abs{S})^2}{p\abs{S} + \eps(1-\eps)p\abs{S}/3}}\leq 2\exp\parens*{-\frac{\eps^2}{12}p\abs{S}}.\qedhere
\]
\end{proof}

\begin{proof}[Proof of Theorem \ref{thm:1+eps-dense-cauchy}]
We make minor modifications of Lemmas 2.10 and 2.12 in \cite{DBLP:conf/soda/WangW19}. Let $\{X_i\}_{i=1}^n$ be independent standard Cauchys.  

\paragraph{Upper bound.}

Let $\mathcal E_i\coloneqq \braces*{\abs{X_i}\leq r\log n(\log\log r)^{-1}}$. Then,
\[
    \Pr\parens*{\mathcal E_i} = 1 - \frac2\pi \arctan\parens*{\frac{r\log r}{\log\log r}}\geq 1 - \frac2\pi \frac{\log\log r}{r\log r} \gg \frac1{1+\eps}.
\]
Let $\mathcal E = \bigcap_{i=1}^r \mathcal E_i$. Then,
\[
    \E\parens*{\abs{X_i}\mid\mathcal E} = \E\parens*{\abs{X_i}\mid\mathcal E_i} = \frac1{\Pr(\mathcal E_i)}\frac1\pi\log\parens*{1 + \parens*{\frac{r\log r}{\log\log r}}^2}
\]
and thus by linearity of expectation,
\[
    \mu\coloneqq \E\parens*{\sum_{i=1}^r\abs{X_i} \mid \mathcal E} = \frac1{\Pr(\mathcal E_i)}\frac{r}\pi\log\parens*{1 + \parens*{\frac{r\log r}{\log\log r}}^2} \leq (1+\eps)\frac{2}{\pi}r\log r.
\]
Now by a Chernoff bound applied to the $\abs{X_i}(\log\log r/r\log r)\in[0,1]$ conditioned on $\mathcal E$,
\[
    \Pr\parens*{\sum_{i=1}^r \abs{X_i} \geq (1+\eps)\mu\mid \mathcal E}\leq \exp\parens*{-\frac{\eps^2 \mu}{3}\frac{\log\log r}{\log r}} = \exp\parens*{-\Theta(\eps^2\log\log r)}
\]
so
\begin{align*}
    \Pr\parens*{\sum_{i=1}^r \abs{X_i} \leq (1+\eps)\mu} &\geq \Pr\parens*{\sum_{i=1}^n \abs{X_i} \leq (1+\eps)\mu \mid \mathcal E}\Pr(\mathcal E) \\
    &\geq \parens*{1 - \exp\parens*{-\Theta(\eps^2\log\log r)}}\parens*{1 - \frac2\pi \frac{\log\log r}{r\log r}}^r \geq 1 - \frac{(3/\eps)^d}{\delta}.
\end{align*}

\paragraph{Lower bound.}

Let $T = \frac{(3/\eps)^d}{\delta}$. Note that by Taylor expansion, there is a $T'\geq 0$ such that for $t\geq T'$,
\[
    \Pr\parens*{\abs{X_i} > t}\geq\frac2\pi t^{-1} + O(t^{-3}).
\]
Now for $i\geq 0$ and $j\in[r]$, define the indicator
\[
    N_j^i \coloneqq \begin{cases}
        1 & \text{if $\abs{X_i} > (1+\eps)^i T'$} \\
        0 & \text{otherwise}
    \end{cases}
\]
and $N^i\coloneqq \sum_{j\in[r]} N_j^i$. Note that by the Taylor expansion bound,
\[
    \E N^i \geq \frac{2r}{\pi} \frac1{(1+\eps)^i T'}.
\]
Then by Chernoff bounds,
\[
    \Pr\parens*{N^i \geq (1+\eps)\E N^i}\leq \exp\parens*{-\frac{\eps^2}{3}\frac{2r}{\pi} \frac1{(1+\eps)^i T'}}
\]
Now let $i_{\max}$ be the largest $i$ such that
\[
    \exp\parens*{-\frac{\eps^2}{3}\frac{2r}{\pi} \frac1{(1+\eps)^i T'}}\leq \frac1{T}.
\]
Then by a union bound over the first $i_{\max}$ level sets, $N^i\geq (2/\pi)r(1+\eps)^i T'$ and thus with probability at least $1 - 1/T$,
\[
    \sum_{i=1}^r \abs{X_i} \geq \sum_{i=0}^{i_{\max}} \frac2\pi r(1+\eps)^i T' = \frac2\pi r \log\parens*{\frac1{T'}\frac{\eps^2}{3}\frac{2r}{\pi}\frac1{\log T}} \geq (1-\eps)\frac2\pi r\log r.
\]

\paragraph{Net argument.}

Given the above concentration results, the rest of the argument proceeds as done in \cite{DBLP:conf/soda/WangW19}, using $1$-stability of Cauchys and then a standard net argument. 
\end{proof}

%% file: subspace_embedding_appendix.tex
\section{No contraction bound}\label{section:appendix:no-contraction}

In this section, we prove a no contraction result for a generic $M$-sketch embedding with subsampling rates $p_h$ as specified in Lemma \ref{lem:essential-weight-classes} and a hash bucket size of $N_0$ for the $0$th level and $N$ for the $h$th level for $h\in[h_{\max}]$ as specified in Definition \ref{def:useful-constants-no-contraction}. This allows us to apply the results to both $M$-sketch with random and fixed boundaries, with varied branching factors and failure rates. Recall the definition of the $M$-sketch from Definition \ref{def:M-sketch}. 

\begin{thm}[No contraction]\label{thm:no-contraction-high-prob}
Let $\bfy\in\mathbb R^n$ with $\norm{\bfy}_1 = 1$. Let $\eps\in(0,1)$ and $\delta\in (0,1)$. Let $\bfS$ be drawn as an $M$-sketch matrix. Then with probability at least $1 - 6\delta$,
\[
    \norm{\bfS\bfA\bfy}_1 \geq (1 - 16\eps)\norm{\bfA\bfy}_1.
\]
\end{thm}

\subsection{Essential weight classes}

We first classify a small subset of weight classes of $\bfy$ that we need to preserve for at least a $(1-\eps)$ approximation. 

\begin{lem}[Essential weight classes]\label{lem:essential-weight-classes}
Let $\bfy\in\mathbb R^n$ with $\norm{\bfy}_1 = 1$. Let $m_{\min}$ be a minimum class size parameter, let $B$ be a branching factor parameter, and let $\eps$ be an accuracy parameter. Finally, let $p_h = p_0/B^{h-1}$ for $h\in [\log_B n]$ be sampling rates. Define
\begin{align*}
    h_{\max} &\coloneqq \log_B n \\
    q_{\max} &\coloneqq \log_2\frac{n}{\eps} \\
    m_{\min} &\coloneqq \frac{12}{\eps^2}\log\frac{4q_{\max}}{\delta} \\
    M_{\geq} &\coloneqq \log_2\frac{B}{\eps} \\
    M_{<} &\coloneqq \log_2\frac{m_{\min}}{p_0\eps}
\end{align*}
and weight classes
\begin{align*}
    \hat Q_h &\coloneqq \braces*{q\in [q_{\max}] : m_{\min}\leq p_h\abs{W_q(\bfy)} < Bm_{\min}} && h\in[h_{\max}] \\ 
    Q_h &\coloneqq \braces*{q\in \hat Q_h : q\leq M_{\geq} + \min_{q\in\hat Q_h}q, \norm{W_q(\bfy)}_1\geq \frac{\eps}{q_{\max}}} && h\in[h_{\max}] \\
    Q_{<} &\coloneqq \braces*{q : \abs{W_q(\bfy)} < m_{\min}/p_0, q\leq M_{<}, \norm{W_q}_1\geq \frac{\eps}{M_{<}}} \\
    Q^* &\coloneqq Q_{<} \cup \bigcup_{h\in[h_{\max}]}Q_h
\end{align*}
Then,
\[
    \sum_{q\in Q^*}\norm{W_q(\bfy)}_1\geq 1-6\eps
\]
\end{lem}
\begin{rem}
The $\hat Q_h$ are the weight classes for which the $h$th level is the smallest level at which we sample at least $m_{\min}$ elements of $W_q$ in expectation, so that the mass is extremely concentrated. The $Q_h$ are the weight classes that restrict $\hat Q_h$ to only as many levels as we need to preserve the mass of $\hat Q_h$ up to a $1-\eps$ factor. The set $Q_{<}$ specifies the subset of levels that are too small for concentration, but are needed to preserve the mass of $\bfy$ up to a $1-\eps$ factor. The set $Q^*$ specifies the union of these essential weight classes needed for a $1-\eps$ approximation.
\end{rem}
\begin{proof}
Note that
\[
    \sum_{q > q_{\max}}\norm{W_q}_1\leq \frac{\eps}{n}\sum_{q > q_{\max}}\abs{W_q}_1\leq \eps
\]
so we restrict our attention to $q\in[q_{\max}]$. Note that every $q\in[q_{\max}]$ belongs in either exactly one class $\hat Q_h$, or $\abs{W}_q < m_{\min}/p_0$. The total weight of weight classes with $\abs{W_q} < m_{\min}/p_0$ and $q > M_{<}$ is at most
\[
    \sum_{q > M_{<}}\abs{W_q}2^{1-q} = 2\frac{m_{\min}}{p_0}2^{-M_{<}}\sum_{q > 0}2^{-q}\leq 2\frac{m_{\min}}{p_0}\frac{p_0\eps}{m_{\min}} = 2\eps.
\]
Furthermore, let $h\in[h_{\max}]$ and let $q_h^*\coloneqq \min_{q\in\hat Q_h}q$. Then the ratio of the total weight of classes in $W_q$ with $q > M_{\geq} + q_h^*$ to $\norm{W_{q_h^*}}_1$ is at most
\begin{align*}
    \frac{1}{\norm{W_{q_h^*}}_1}\sum_{q > M_{\geq} + q_h^*}2^{1-q}\frac{Bm_{\min}}{p_h} &\leq \frac{1}{2^{-q_h^*}m_{\min}/p_h}\sum_{q > M_{\geq} + q_h^*}2^{1-q}\frac{Bm_{\min}}{p_h} \\
    &= 2B\sum_{q > M_{\geq}}2^{-q}\leq 2B 2^{-M_{\geq}}\leq 2B\frac{\eps}{B} = 2\eps.
\end{align*}
We thus have that
\[
    \sum_{h\in[h_{\max}]} \sum_{q > M_{\geq} + q_h^*}\norm{W_{q}}_1\leq \sum_{h\in[h_{\max}]} 2\eps\norm{W_{q_h^*}}_1\leq 2\eps.
\]
Furthermore, the total weight of classes in $W_q$ with $\norm{W_q}_1 < \eps/q_{\max}$ is at most
\[
    \sum_{q : \norm{W_q}_1 < \eps/q_{\max}}\norm{W_q}_1\leq q_{\max}\frac{\eps}{q_{\max}} = \eps.
\]
We conclude by combining the above bounds.
\end{proof}

\subsection{Approximate perfect hashing}

\begin{dfn}[Useful constants]\label{def:useful-constants-no-contraction}
\begin{align*}
    N_0' &\geq \frac{1}{\delta}\frac{M_{<}}{p_0 \eps}m_{\min}\parens*{1+\frac76\frac{2\log(2M_{<}/\delta)}{\eps^2}} \\
    N_0 &\coloneqq 2N_0'\log N_0' && \text{Number of hash buckets at the $0$th level} \\
    N' &\geq \frac{B}{\eps}m_{\min}\parens*{M_{\geq}+\frac76\frac{2q_{\max}\log(2q_{\max}/\delta)}{\eps^2}} \\
    N &\coloneqq 2N'\log N' && \text{Number of hash buckets}
\end{align*}
\end{dfn}
We allow the flexibility to choose the number of buckets $N_0$ and $N$ to be larger if needed. The $N_0$ and $N$ are chosen so that
\[
    \frac{M_{<}}{p_0}m_{\min}\parens*{1 + \frac76\frac{2\log(2N_0M_{<}/\delta)}{\eps^2}}\leq \delta\eps N_0
\]
and
\[
    Bm_{\min}\parens*{M_{\geq} + \frac76\frac{2q_{\max}\log(2Nq_{\max}/\delta)}{\eps^2}}\leq \eps N.
\]

\begin{lem}[Concentration of sampled mass]\label{lem:concentration-of-sampled-mass}
Suppose $p_h\abs{W_q}\geq m_{\min}$. Then with probability at least $1-\delta/q_{\max}$,
\begin{align*}
    \sum_{\bfy_i\in W_q}b_{i,h} &= (1\pm\eps)p_h\abs{W_q} \\
    \sum_{\bfy_i\in W_q}\abs{\bfy_i} b_{i,h} &= (1\pm\eps)p_h\norm{W_q}_1
\end{align*}
\end{lem}
\begin{proof}
Let
\[
    X\coloneqq\sum_{\bfy_i\in W_q}b_{i,h}.
\]
By the Chernoff bound,
\[
    \Pr\parens*{\abs{X - \E X}\geq \eps\E X}\leq 2\exp\parens*{-\frac{\eps^2\E X}{3}}\leq \frac\delta{2q_{\max}}.
\]
Similarly, let
\[
    Y\coloneqq\sum_{\bfy_i\in W_q}\abs{\bfy_i} b_{i,h}.
\]
Note that
\begin{align*}
    \E Y &= p_h\norm{W_q}_1\geq 2^{-q}p_h\abs{W_q} \\
    \abs{\bfy_i} b_{i,h} &\leq 2^{1-q} \\
    \Var(\abs{\bfy_i} b_{i,h}) &\leq p_h 2^{2-2q}
\end{align*}
so by Bernstein's inequality,
\begin{align*}
    \Pr\parens*{\abs{Y - \E Y}\geq \eps\E Y} &\leq 2\exp\parens*{-\frac12\frac{(\eps\E Y)^2}{p_h 2^{2-2q} \abs{W_q} + (\eps\E Y) 2^{1-q}/3}} \\
    &\leq 2\exp\parens*{-\frac12\frac{(\eps 2^{-q}p_h\abs{W_q})^2}{p_h 2^{2-2q} \abs{W_q} + (\eps 2^{1-q} p_h\abs{W_q}) 2^{1-q}/3}} \\
    &= 2\exp\parens*{-\frac18\frac{p_h\abs{W_q}\eps^2}{1 + \eps/3}} \leq 2\exp\parens*{-\frac{\eps^2}{12}p_h\abs{W_q}} \leq \frac\delta{2q_{\max}}.
\end{align*}
We conclude by a union bound over the two events. 
\end{proof}

The following lemma uses a standard balls and bins martingale argument (e.g., \cite{lee2016lecture}) to show that most items are hashed uniquely. 
\begin{lem}[Approximately perfect hashing]\label{lem:approx-perfect-hashing}
Let $h\in[h_{\max}]$ and let $Q \subseteq \braces*{q : p_h\abs{W_q}\geq m_{\min}}$. Let $\hat W\subset\bfy$ contain $W_Q\coloneqq \bigcup_{q\in Q}W_q$. Let $p_h\abs{\hat W}\leq \eps N$ for some $\eps\in (0,1/2)$. Then with probability at least $1 - (3/2)\abs{Q}\delta/q_{\max}$, every $W_q$ has a $W_q^* \subset W_q$ that gets sampled and placed in a hash bucket with no other members of $\hat W$, and $\abs{W_q^*}\geq (1-3\eps) p_h\abs{W_q}$ and $\norm{W_q^*}_1\geq (1-9\eps) p_h\norm{W_q}_1$. 
\end{lem}
\begin{proof}
We apply Lemma \ref{lem:concentration-unique-hashing} to see that with probability at least
\[
    1 - 2\exp\parens*{-\frac{\eps^2}{12}p_h\abs{W_q}}\leq 1 - \frac\delta{2q_{\max}},
\]
there is a set $W_q^*\subseteq W_q$ of elements that are hashed to a bucket with no other element of $\hat W$ in it and of size $\abs{W_q^*}\geq (1-\eps)^2 p_h\abs{W_q}\geq (1-3\eps)p_h\abs{W_q}$ with probability at least $1 - \delta/2q_{\max}$. We condition on this event. 

By Lemma \ref{lem:concentration-of-sampled-mass}, with probability at least $1-\delta/q_{\max}$, we sample $(1\pm\eps)p_h\abs{W_q}$ elements with mass $(1\pm\eps)p_h\norm{W_q}_1$. Note then that there are at most $4\eps p_h\abs{W_q}$ sampled elements that do not belong $W_q^*$. The mass of these elements is at most
\[
    4\eps p_h\abs{W_q}2^{1-q}\leq 8\eps p_h\norm{W_q}_1.
\]
Thus,
\[
    \norm{W_q^*}_1\geq (1-\eps)p_h\norm{W_q}_1 - 8\eps p_h\norm{W_q}_1 = (1-9\eps)p_h\norm{W_q}_1.
\]
We conclude by a union bound over the weight classes $Q$. 
\end{proof}

\subsection{Preserving weight classes}

\begin{dfn}
\begin{align*}
    \tau_0 &\coloneqq \frac{p_0\eps}{2M_{<}m_{\min}} && \text{Size of a relatively large element at $0$th level} \\
    T_0 &\coloneqq \frac67\frac{\eps \tau_0}{\log(2N_0 M_{<}/\delta)} && \text{Size of a relatively small element at $0$th level} \\
    \tau_h &\coloneqq \frac{p_h\eps}{2q_{\max}Bm_{\min}} && \text{Size of a relatively large element} \\
    T_h &\coloneqq \frac67\frac{\eps \tau_h}{\log(2Nq_{\max}/\delta)} && \text{Size of a relatively small element}
\end{align*}
\end{dfn}

\begin{dfn}[Large elements]
\begin{align*}
    Q_{<,0} &\coloneqq \braces*{q : q\leq \log_2 \frac1{T_0}} \\
    Q_{<,h} &\coloneqq \braces*{q : q\leq \log_2 \frac1{T_h}}
\end{align*}
\end{dfn}
The weight class $Q_{<,h}$ is the set of relatively large elements at the $h$th level of sampling. 

We directly recall the following Lemma 3.3 from \cite{DBLP:conf/soda/ClarksonW15}.
\begin{lem}\label{lem:max-hash-bucket-bound}
Let $h\in[h_{\max}]$, $\bar W\subset \bfy$, $T\geq \norm{\bar W}_\infty$, and $\delta'\in(0,1)$. If
\[
    N \geq \frac{6\norm{\bar W}_1}{T\log (N/\delta)},
\]
then 
\[
    \Pr\parens*{\max_{k\in[N]}\norm{L_{h,k}\cap \bar W}_1\leq \frac76 T\log(N/\delta)}\geq 1-\delta'.
\]
\end{lem}

\subsubsection{Preserving weight classes in \texorpdfstring{$Q_{h}$}{Qh}}

\begin{lem}\label{lem:large-entries-in-Qh}
Let $h\in[h_{\max}]$, $q\in Q_h$. Then
\[
    \abs{\bfy_i}\geq \tau_h = \frac{p_h\eps}{2q_{\max}Bm_{\min}}.
\]
\end{lem}
\begin{proof}
By the definition of $Q_h$, we have that $m_{\min}\leq p_h\abs{W_q}\leq B m_{\min}$ and $\norm{W_q}_1\geq \frac{\eps}{q_{\max}}$. We then have that
\[
    \abs{W_q}2^{1-q}\geq \norm{W_q}_1\geq \frac{\eps}{q_{\max}}.
\]
Then for any $\bfy_i\in W_q$,
\[
    \abs{\bfy_i}\geq 2^{-q}\geq \frac{\eps}{2\abs{W_q}q_{\max}}\geq \frac{p_h\eps}{2q_{\max}Bm_{\min}}.\qedhere
\]
\end{proof}

\begin{lem}\label{lem:large-buckets}
Let $h\in[h_{\max}]$ and let $L_h(\bfy_i)$ denote the multiset of elements in the hash bucket in the $h$th level containing $\bfy_i$. Then with probability at least $1 - 2\abs{Q_h}\delta/q_{\max}$, for all $q\in Q_h$, we sample a set $W_q^*\subseteq W_q$ such that 
\[
    \norm{W_q^*}_1\geq (1-9\eps)p_h\norm{W_q}_1
\]
and for every $\bfy_i\in W_q^*$,
\[
    \Abs{\sum_{\bfy_j\in L_h(\bfy_i)}\Lambda_{j}\bfy_j}\geq (1-\eps)\abs{\bfy_i}.
\]
\end{lem}
\begin{proof}
Let $\hat W = W_{Q_h}\cup W_{Q_{<,h}}$. Then by our choice of $N$,
\[
    \abs{\hat W} \leq \abs{W_{Q_h}} + \abs{W_{Q_{<,h}}} \leq M_{\geq}\frac{Bm_{\min}}{p_h} + \frac1{T_h} = \frac{B}{p_h}m_{\min}\parens*{M_{\geq} + \frac76\frac{2q_{\max}\log(2Nq_{\max}/\delta)}{\eps^2}}\leq \frac{\eps N}{p_h}.
\]
Then by Lemma \ref{lem:approx-perfect-hashing}, with probability at least $1 - (3/2)\abs{Q_h}/q_{\max}$, for each $q\in Q_h$, there is a set of sampled elements $W_q^*\subseteq W_q$ that get hashed to a bucket with no other members of $\hat W$, and $\norm{W_q^*}_1\geq (1-9\eps)p_h\norm{W_q}_1$. 

Note that for each $q\in Q_h$ and $\bfy_i\in W_q^*$, the absolute value of the largest element in $L_h(\bfy_i)$ not equal to $\bfy_i$ is at most $T_h$, since we have hashed the elements of $W_{Q_{<,h}}$ to other buckets. Then by Lemma \ref{lem:max-hash-bucket-bound}, the $\ell_1$ mass of elements that are at most $T_h$ in all hash buckets are at most 
\[
    \norm{L_h(\bfy_i) \setminus \{\bfy_i\}}_1\leq \frac76 T_h \log(2Nq_{\max}/\delta) = \eps \tau_h
\]
with probability at least $1-\delta/2q_{\max}$. By a union bound over $q\in Q_h$, this is true for all $\bfy_i\in W_q^*$ for $q\in Q_h$ with probability at least $1 - 2\abs{Q_h}\delta/q_{\max}$. 

Recall from Lemma \ref{lem:large-entries-in-Qh} that $\abs{\bfy_i}\geq \tau_h$ for all $\bfy_i\in W_q$ with $q\in Q_h$. Note then that the mass of this hash bucket is at least
\[
    \Abs{\sum_{\bfy_j\in L_h(\bfy_i)}\Lambda_{j}\bfy_j}\geq \abs{\bfy_i} - \norm{L_h(\bfy_i) \setminus \{\bfy_i\}}_1\geq \abs{\bfy_i} - \eps\tau_h\geq (1-\eps)\abs{\bfy_i}
\]
which is the desired bound. Thus overall, the total success probability is at least $1 - 2\abs{Q_h}\delta/q_{\max}$. 
\end{proof}

\begin{lem}\label{lem:approximation-factor-Qh}
Let $h\in[h_{\max}]$. Then with probability at least $1 - 2\abs{Q_h}\delta/q_{\max}$, we have that
\[
    \norm{\bfC^{(h)}\bfS^{(h)}\bfy}_1 \geq (1-10\eps)\sum_{q\in Q_h}\norm{W_q}_1.
\]
\end{lem}
\begin{proof}
Taking a sum over $q\in Q_h$ and $\bfy_i\in W_q^*$, we find that 
\begin{align*}
    \norm{\bfC^{(h)}\bfS^{(h)}\bfy}_1 &\geq \frac1{p_h}\sum_{q\in Q_h}\sum_{\bfy_i\in W_q}b_{i,h}\Abs{\sum_{\bfy_j\in L_h(\bfy_i)}\Lambda_{j}\bfy_j} && \text{Looking only at rows in $Q_h$} \\
    &\geq \frac1{p_h}\sum_{q\in Q_h}\sum_{\bfy_i\in W_q^*}\Abs{\sum_{\bfy_j\in L_h(\bfy_i)}\Lambda_{j}\bfy_j} && \text{Looking only at good sampled elements $W_q^*$} \\
    &\geq \frac1{p_h}\sum_{q\in Q_h}\sum_{\bfy_i\in W_q^*}(1-\eps)\abs{\bfy_i} && \text{Lemma \ref{lem:large-buckets}} \\
    &\geq (1-\eps)\frac1{p_h}\sum_{q\in Q_h}\norm{W_q^*}_1 \\
    &\geq (1-\eps)(1-9\eps)\frac1{p_h}\sum_{q\in Q_h}p_h\norm{W_q}_1 && \text{Lemma \ref{lem:large-buckets}} \\
    &\geq (1-10\eps)\sum_{q\in Q_h}\norm{W_q}_1
\end{align*}
which is the desired bound. The failure probability is the same as from Lemma \ref{lem:large-buckets}.
\end{proof}

\subsubsection{Preserving weight classes in \texorpdfstring{$Q_{<}$}{Q<}}

With essentially the exact same proofs as in the above section, we have the following analogues of Lemmas \ref{lem:large-entries-in-Qh}, \ref{lem:large-buckets}, and \ref{lem:approximation-factor-Qh}. 
\begin{lem}
Let $q\in Q_{<}$. Then
\[
    \abs{\bfy_i}\geq \tau_0 = \frac{\eps p_0}{M_{<}m_{\min}}.
\]
\end{lem}
\begin{lem}
Let $Q = \braces*{q\in Q_{<}: \abs{W_q}\geq m_{\min}}$. Let $L_0(\bfy_i)$ denote the multiset of elements in the hash bucket in the $0$th level containing $\bfy_i$. Then with probability at least $1 - 2\abs{Q}\delta/M_{<}$, for all $q\in Q$, there is a set $W_q^*\subseteq W_q$ such that $W_q^*$ is hashed to a different bucket than $W_{Q_{<,0}}\supset W_{Q_{<}}$, 
\[
    \norm{W_q^*}_1\geq (1-9\eps)p_h\norm{W_q}_1,
\]
and for every $\bfy_i\in W_q^*$,
\[
    \Abs{\sum_{\bfy_j\in L_0(\bfy_i)}\Lambda_{j}\bfy_j}\geq (1-\eps)\abs{\bfy_i}.
\]
\end{lem}

\begin{lem}\label{lem:preserve-weight-class-geq-m}
Let $Q = \braces*{q\in Q_{<}: \abs{W_q}\geq m_{\min}}$. Then with probability at least $1 - 2\abs{Q}\delta/M_{<}$,
\[
    \norm{\bfC^{(0)}\bfy}_1\geq (1-10\eps)\sum_{q\in Q}\norm{W_q}_1. 
\]
\end{lem}

It thus remains to handle the case of $\{q\in Q_{<} : \abs{W_q} < m_{\min}\}$. For these small level sets, we can perfectly hash these into separate buckets from all the entries in $Q_{<,0}$. 

\begin{lem}\label{lem:preserve-weight-class-le-m}
Let $Q = \{q\in Q_{<} : \abs{W_q} < m_{\min}\}$. Let $L_0(\bfy_i)$ denote the multiset of elements in the hash bucket in the $0$th level containing $\bfy_i$. With probability at least $1-2\delta$, every member of $W_Q$ is hashed to a different bucket than $W_{Q_{<,0}}\supset W_{Q_{<}}$, and we have for every $\bfy_i\in W_Q$ that
\[
    \Abs{\sum_{\bfy_j\in L_0(\bfy_i)}\Lambda_{j}\bfy_j}\geq (1-\eps)\abs{\bfy_i}.
\]
\end{lem}
\begin{proof}
Note that
\[
    N_0\geq \frac1\delta \abs{W_Q} \abs{W_{Q_{<,0}}}. 
\]
Then for every $(\bfy_i, \bfy_j)\in W_Q\times W_{Q_{<,0}}$, there is a $\delta/\abs{W_Q} \abs{W_{Q_{<,0}}}$ probability that $\bfy_i$ and $\bfy_j$ get hashed to the same location. By a union bound, none of these pairs are hashed to to the same location with probability at least $1-\delta$. Then by Lemma \ref{lem:max-hash-bucket-bound}, the $\ell_1$ mass of elements that are at most $T_0$ in all hash buckets are at most
\[
    \norm{L_0(\bfy_i)\setminus\{\bfy_i\}}_1\leq \frac76 T_0 \log(2N_0M_{<}/\delta) = \eps \tau_0
\]
with probability at least $1 - \delta/2M_{<}$. By a union bound over $q\in Q$, this is true for all $\bfy_i\in W_Q$ with probability at least $1 - \abs{Q}\delta/2M_{<}\geq 1 - \delta$. Then,
\[
    \Abs{\sum_{\bfy_j\in L_0(\bfy_i)}\Lambda_{j}\bfy_j}\geq \abs{\bfy_i} - \norm{L_0(\bfy_i) \setminus \{\bfy_i\}}_1\geq \abs{\bfy_i} - \eps\tau_0\geq (1-\eps)\abs{\bfy_i}
\]
which is the desired bound. Thus overall, the failure probability is $1 - 2\delta$. 
\end{proof}

\subsection{Proof of Theorem \ref{thm:no-contraction-high-prob}}
We finally gather the pieces from above. 
\begin{proof}{Proof of Theorem \ref{thm:no-contraction-high-prob}}
We union bound over the events and sum over the results of Lemmas \ref{lem:preserve-weight-class-geq-m}, \ref{lem:preserve-weight-class-le-m}, and \ref{lem:approximation-factor-Qh}, so that with probability at least $1 - 6\delta$, 
\[
    \norm{\bfS\bfA\bfy}_1 \geq (1-10\eps)\sum_{q\in Q_{<}\cup \bigcup_{h\in[h_{\max}]}Q_h} \norm{W_q}_1.
\]
We conclude by chaining this inequality together with the result of Lemma \ref{lem:essential-weight-classes}. 
\end{proof}

%% file: entrywise_embedding_appendix.tex
\section{Missing proofs from Section \ref{section:entrywise-embeddings}}\label{section:appendix:entrywise_embeddings}

\begin{proof}[Proof of Lemma \ref{lem:rademacher-dimension-l1}]
By Hoeffding bounds, we have for each $j\in[d]$ that
\[
    \Pr\parens*{\Abs{\sum_{i=1}^s \eps_i \bfe_j^\top \bfx_i} > \sqrt{\frac{s}{2}\log\frac{2d}{\delta}}}\leq 2\exp\parens*{-\frac{2(\sqrt{(s/2)\log(2d/\delta)})^2}{s}}\leq \frac{\delta}{d}
\]
Then by a union bound over the $d$ choices of $j$, with probability at least $1-\delta$, the complement event of the above holds for every $j\in[d]$. Conditioned on this event, we have that
\[
    \Norm{\sum_{i=1}^s \eps_i \bfx_i}_1\leq \sum_{j=1}^d \Abs{\sum_{i=1}^s \eps_i \bfe_j^\top \bfx_i}\leq d\sqrt{\frac12\log\frac{2d}{\delta}}\sqrt{s}
\]
as desired.
\end{proof}

\begin{proof}[Proof of Lemma \ref{lem:sketch-cauchy-matrix}]
For each $i\in[r]$ and $j\in[d]$, by the $1$-stability of Cauchy variables,
\[
    \bfe_i^\top\bfS\bfA\bfe_j \stackrel{d}{=} \norm{\bfe_i^\top\bfS}_1\mathcal C_{i,j}
\]
where $\mathcal C_{i,j}$ are drawn as standard Cauchy variables, and are independent for distinct $j$. Now note that $\abs{\mathcal C_{i,j}}\leq O(rd)$ with probability at least $1-(100rd)^{-1}$ and thus by a union bound, $\max_{i\in[r], j\in[d]} \abs{\mathcal C_{i,j}}\leq O(rd)$ with probability at least $1 - 1/400$. We condition on this event. Note then that the conditional expectation is at most
\[
    \E\abs{\mathcal C_{i,j}}\leq O(\log(rd))
\]
as shown in \cite{DBLP:journals/jacm/Indyk06}. Then,
\[
    \E\norm{\bfS\bfA}_1 = \sum_{j=1}^d\sum_{i=1}^r \E\abs{\bfe_i^\top\bfS\bfA\bfe_j} = \sum_{j=1}^d\sum_{i=1}^r \norm{\bfe_i^\top\bfS}_1 \E\abs{\mathcal C_{i,j}} = O(d\log(rd))\norm{\bfS}_1
\]
so a Markov bound and a union bound with the earlier event shows that
\[
    \Pr\parens*{\norm{\bfS\bfA}_1\leq O(d\log(rd))\norm{\bfS}_1}\geq 1 - \frac1{200}.
\]

For the lower bound, let $\hat{\mathcal C}_{i,j}$ be the truncation of $\hat{\mathcal C}_{i,j}$ at $d$, i.e., 
\[
    \hat{\mathcal C}_{i,j} = \begin{cases}
        \mathcal C_{i,j} & \text{if $\abs{\mathcal C_{i,j}}\leq d$} \\
        0 & \text{otherwise}
    \end{cases}.
\]
Note then that by \cite[Lemma 6]{DBLP:journals/jacm/Indyk06}, $\Var(\abs{\hat C_{i,j}}) = \Theta(d)$ so
\[
    \sigma^2\coloneqq \Var\parens*{\sum_{j=1}^d \norm{\bfe_i^\top\bfS}_1 \abs{\hat{\mathcal C}_{i,j}}} = \sum_{j=1}^d \Theta(d)\norm{\bfe_i^\top\bfS}_1^2 = \Theta(d^2)\sum_{j=1}^d \norm{\bfe_i^\top\bfS}_1^2
\]
and
\[
    \mu\coloneqq \E\parens*{\sum_{j=1}^d \norm{\bfe_i^\top\bfS}_1 \abs{\hat{\mathcal C}_{i,j}}} = \Theta(d\log d)\sum_{j=1}^d\norm{\bfe_i^\top\bfS}_1. 
\]
Then by Chebyshev's inequality,
\[
    \Pr\parens*{\sum_{j=1}^d \norm{\bfe_i^\top\bfS}_1 \abs{\hat{\mathcal C}_{i,j}} - \mu \leq \Theta(\log d)\sigma}\leq \frac1{400}. 
\]
Thus, with probability at least $1 - 1/400$,
\[
    \sum_{j=1}^d \norm{\bfe_i^\top\bfS}_1 \abs{\mathcal C_{i,j}} \geq \sum_{j=1}^d \norm{\bfe_i^\top\bfS}_1 \abs{\hat{\mathcal C}_{i,j}} \geq \gamma \norm{\bfe_i^\top \bfS}_1
\]
for $\gamma = \Omega(d\log d)$. Let $\mathcal E_i$ denote the above event, so that $\Pr(\mathcal E_i)\geq 1- 1/400$. Then,
\[
    \E\parens*{\sum_{i=1}^r \mathbbm{1}(\neg\mathcal E_i)\bracks*{\sum_{j=1}^d \norm{\bfe_i^\top\bfS}_1 \abs{\mathcal C_{i,j}}}}\leq \sum_{i=1}^r \frac{1}{400} \gamma\norm{\bfe_i^\top\bfS}_1
\]
so by Markov's inequality, with probability at least $1 - 1/200$,
\[
    \sum_{i=1}^r \mathbbm{1}(\neg\mathcal E_i)\bracks*{\sum_{j=1}^d \norm{\bfe_i^\top\bfS}_1 \abs{\mathcal C_{i,j}}}\leq \frac{\gamma}{2}\sum_{i=1}^r \norm{\bfe_i^\top\bfS}_1 = \frac{\gamma}{2}\norm{\bfS}_1.
\]
Then, conditioning on this event,
\[
    \norm{\bfS\bfA}_1 \geq \sum_{i=1}^r \bracks*{\sum_{j=1}^d \norm{\bfe_i^\top\bfS}_1 \abs{\mathcal C_{i,j}}}(1 - \mathbbm{1}(\neg\mathcal E_i))\geq \gamma \sum_{i=1}^r \norm{\bfe_i^\top\bfS}_1 - \frac{\gamma}2 \norm{\bfS}_1 = \frac{\gamma}{2}\norm{\bfS}_1
\]
as desired.
\end{proof}

%% file: subspace_embedding_random_appendix.tex
\section{Missing proofs from Section \ref{section:subspace_embeddings_random}}\label{section:appendix:subspace_embeddings_random}

\begin{proof}[Proof of Lemma \ref{lem:truncated-p-moments}]
Because $\Pr(\abs{X}\leq T) = \Theta(1)$ for $T$ large enough,
\[
    \E_{X\sim\mathsf{trunc}_T(\mathcal D)} \abs{X} = \Theta(1)\E_{X\sim\mathcal D} (\abs{X}\mid \abs{X}\leq T), \qquad \E_{X\sim\mathsf{trunc}_T(\mathcal D)} X^2 = \Theta(1)\E_{X\sim\mathcal D} (X^2\mid \abs{X}\leq T)
\]
By the layer cake theorem,
\begin{align*}
    \E_{X\sim\mathcal D} (\abs{X}\mid \abs{X}\leq T) &= \int_{0}^\infty \Pr\parens*{\abs{X}>x\mid \abs{X}\leq T}~dx \\
    &= \int_0^T \frac{\Pr\parens*{x < \abs{X}\leq T}}{\Pr\parens*{\abs{X}\leq T}} \\
    &= \frac{1}{\Pr\parens*{\abs{X}\leq T}}\int_0^T \Pr\parens*{\abs{X}>x} - \Pr(\abs{X}>T)~dx \\
    &= \Theta(1)\int_{0}^{T} \Theta(x^{-p})~dx
\end{align*}
and similarly,
\[
    \E_{X\sim\mathcal D} (X^2\mid \abs{X}\leq T) = \Theta(1)\int_0^T x\Pr(\abs{X}>x)~dx = \Theta(1)\int_0^T \Theta(x^{1-p})~dx.
\]
Solving the simple integrals yields the desired results.
\end{proof}

\begin{proof}[Proof of Theorem \ref{thm:subspace-embedding-p<1}]
The distortion upper bound is just Lemma \ref{lem:no-expansion}. 

\paragraph{Mass of small entries.}

Let $\bfA = \bfA^H + \bfA^L$ as in Definition \ref{def:AH-AL}, with $T = O\parens*{(nd^2\log d / r)^{1/p}}$. Then, $\bfA^L\sim \mathsf{trunc}_T(\mathcal D)^{n\times d}$ where by Lemma \ref{lem:truncated-p-moments}, the first two moments of each entry are
\[
    \mu = \Theta(T^{1-p}), \qquad \sigma = \Theta(T^{2-p}).
\]
Then by Bernstein's inequality,
\begin{align*}
    -\log \Pr\parens*{\norm{\bfA^L\bfe_j}_1 \geq 2\mu n} &\geq \frac12\frac{(\mu n)^2}{\sigma^2 n + \mu nT/3} \\
    &= \frac{\Omega((T^{1-p}n)^2)}{O(T^{2-p})n + O(T^{1-p})nT} = \Omega(nT^{-p}) = \Omega\parens*{\frac{r}{d^2\log d}} = \Omega(\log d)
\end{align*}
Thus, $\Pr(\norm{\bfA^L\bfe_j}_1 \leq 2\mu n)\geq 1 - 1/\poly(d)$ so by a union bound over the $d$ columns, this event simultaneously holds for all $d$ columns with probability at least $1 - 1/\poly(d)$. Conditioned on this event, by the triangle inequality,
\[
    \norm{\bfA^L\bfx}_1\leq O\parens*{\frac{n^{1/p}}{(\log d)^{1/p-1}}}\norm{\bfx}_1
\]
for all $\bfx\in\mathbb R^d$. 

\paragraph{Mass of large entries.}

Furthermore, let $\bfB'$ be the subset of rows of $\bfA^H$ given by Lemma \ref{lem:p=1-unique-hashing-rows} that are hashed to locations without any other rows of $\bfA^H$. Recall also $\tau_1$ and $\tau_2$ from the lemma. 

We first have that $\norm{\bfS\bfB'\bfx}_1 = \Omega(\norm{\bfA^H\bfx}_1)$ since the rows containing entries larger than $\tau_1$ are perfectly hashed, while rows containing entries between $\tau_2$ and $\tau_1$ are preserved up to constant factors. 

Let $\bfB' = \bfB_{>T}' + \bfB_{\leq T}'$ where $\bfB'_{>T}$ contains the entries of $\bfB'$ that have absolute value greater than $T$ and $\bfB'_{\leq T}$ contains the rest of the entries. Note then that $\bfB'_{>T}$ has at most one nonzero entry per row, and $\bfB'_{\leq T}$ has at most $O(d\cdot r/d\log d) = O(r/\log d)$ nonzero entries and thus by Lemma \ref{lem:power-law-union-bound}, $\norm{\bfB'_{\leq T}}_\infty\leq O(r^{1/p})$ with probability at least $0.99$. We condition on this event. Then for all $\bfx$, 
\begin{align*}
    \norm{\bfS\bfA^H\bfx}_1 &\geq \norm{\bfS\bfB'\bfx}_1 \\
    &\geq \norm{\bfS\bfB_{> T}'\bfx}_1 - \norm{\bfS\bfB_{\leq T}'\bfx}_1 \\
    &= \sum_{j=1}^d \abs{\bfx_j}\norm{\bfB_{>T}'\bfe_j}_1 - \norm{\bfS\bfB_{\leq T}'\bfx}_1 && \text{$\bfB_{>T}'\bfe_j$ have disjoint support} \\
    &\geq \sum_{j=1}^d \abs{\bfx_j}\sum_{k=\log_2 \tau_2}^{\log_2 \tau_1}2^k\Theta(n/2^{kp}) - \norm{\bfB_{\leq T}'\bfx}_1 && \text{Lemmas \ref{lem:p=1-unique-hashing-rows} and \ref{lem:no-expansion}} \\
    &= \Omega((n/\log d)^{1/p}\log d)\norm{\bfx}_1 - O(r)\norm{\bfB_{\leq T}'}_\infty\norm{\bfx}_1 && \text{H\"older's inequality} \\
    &= \Omega(n^{1/p}/(\log d)^{1/p-1})\norm{\bfx}_1 - O(r^{1+1/p})\norm{\bfx}_1 \\
    &= \Omega(n^{1/p}/(\log d)^{1/p-1})\norm{\bfx}_1.
\end{align*}

\paragraph{Conclusion.}

On the other hand, by Lemma \ref{lem:p=1-small-SC1x}, the mass of the $O(r/d\log d)$ rows that are hashed together with the rows of $\bfA^H$ have mass at most
\[
    O\parens*{\frac{1}{\sqrt{\log d}}\frac{n^{1/p}}{(r/d^2\log d)^{1/p-1}}}\norm{\bfx}_1 = o\parens*{n^{1/p}/(\log d)^{1/p-1}}\norm{\bfx}_1. 
\]
Then,
\[
    \frac1{\kappa} \geq \frac{\norm{\bfS\bfA\bfx}_1}{\norm{\bfA\bfx}_1}\geq \frac{\norm{\bfS\bfA^H\bfx}_1 - \norm{\bfS\bfC_1\bfx}_1}{\norm{\bfA^H\bfx}_1 + \norm{\bfA^L\bfx}_1} \geq \frac{\Omega(\norm{\bfA^H\bfx}_1 + n^{1/p}/(\log d)^{1/p-1})\norm{\bfx}_1}{O(\norm{\bfA^H\bfx}_1 + n^{1/p}/(\log d)^{1/p-1}))\norm{\bfx}_1} \geq \Omega\parens*{1}.\qedhere
\]
\end{proof}

\begin{proof}[Proof of Lemma \ref{lem:0<p<2:small-hash-bucket-columns}]
For a hash bucket $i\in[r]$ and $k\in [d]$, let
\[
    Y_{i,k}\coloneqq \Abs{\sum_{j : h(j) = i} \bfe_j^\top \bfC \bfe_k}
\]
where $h$ is the hash function for the \textsf{CountSketch} matrix $\bfS$. By Chernoff bounds and a union bound, there are $\Theta(n/r)$ rows $j\in[n]$ such that $h(j) = i$ for all buckets $i\in[r]$, with probability at least $1 - r\exp(-\Theta(n/r)) = 1 - o(1)$. Conditioned on this event, which is independent of the randomness of $\bfC$,
\begin{align*}
    \E Y_{i,k}^2 &= \sum_{j_1, j_2\in h^{-1}(i)\times h^{-1}(i)} \E\bracks*{(\bfe_{j_1}^\top\bfC\bfe_k)(\bfe_{j_2}^\top\bfC\bfe_k)} \\
    &= \sum_{j:h(j) = i}\E\parens*{\bfe_j^\top\bfC\bfe_k}^2 = O\parens*{\frac{n}{r}T^{2-p}} = O\parens*{(d^2\log d)^{(2-p)/p}\parens*{\frac{n}{r}}^{2/p}}
\end{align*}
by the second moment bound in Lemma \ref{lem:truncated-p-moments}. 

Now let $S$ be the subset of rows of $\bfS'$. Then for each $k\in[d]$,
\begin{align*}
    \E\bracks*{\sum_{i\in S}Y_{i,k}} &= \sum_{i\in S}\E Y_{i,k}\leq \sum_{i\in S}\sqrt{\E Y_{i,k}^2} = O\parens*{r'(d^2\log d)^{1/p-1/2}(n/r)^{1/p}} \\
    \Var\parens*{\sum_{i\in S}Y_{i,k}} &= \sum_{i\in S}\Var(Y_{i,k}) = O\parens*{r'(d^2\log d)^{(2-p)/p}\parens*{n/r}^{2/p}}.
\end{align*}
By Chebyshev's inequality,
\[
    \Pr\parens*{\sum_{i\in S} Y_{i,k}\leq \E\bracks*{\sum_{i\in S}Y_{i,k}} + \lambda\sqrt{\Var\parens*{\sum_{i\in S}Y_{i,k}}}}\geq 1 - \frac1{\lambda}
\]
which gives the desired result.
\end{proof}

\begin{proof}[Proof of Lemma \ref{lem:1<p<2-constant-p-stability}]
We compare $\mathcal D$ to a $p$-stable distribution $\mathcal D_p$. By \cite[Theorem 1.12]{Nolan}, a $p$-stable distribution is a power law with index $p$. Then, there exist constants $T$ and $c$ such that for all $t\geq T$,
\[
    \Pr_{X\sim\mathcal D_p}\parens*{cX > t}\leq \Pr_{Y\sim\mathcal D}\parens*{Y>t}.
\]
We then define the distribution $\mathcal D_p'$ which draws $Z\sim\mathcal D_p'$ as $cX$ for $X\sim\mathcal D$ if $\abs{cX} > T$, and $0$ otherwise. Note then that for $Z\sim\mathcal D_p'$ and $Y\sim\mathcal D$, $\abs{Y}$ stochastically dominates $\abs{Z}$. 

We are then in the position to apply the following theorem from probability theory.
\begin{thm}[Theorem 2, \cite{pruss1997comparisons}]\label{thm:symmetric-sum-dominance}
Let $X_1, X_2, \dots, X_d$ be independent symmetric random variables, and suppose $Y_1, Y_2, \dots, Y_d$ are also independent symmetric random variables. Assume that for every $j$ we have $\abs{Y_j}$ stochastically dominated by $\abs{X_j}$. Then
\[
    \Pr\braces*{\Abs{\sum_{j=1}^d Y_j}\geq\lambda}\leq 2\Pr\braces*{\Abs{\sum_{j=1}^d X_j}\geq\lambda}
\]
for every positive $\lambda$. 
\end{thm}

Thus, it suffices to show Equation \ref{eqn:1<p<2-constant-p-stability} for $\mathcal D_p'$ in place of $\mathcal D$. For $j\in[d]$, let $X_j\sim\mathcal D_p$ and define
\[
    \hat X_j\coloneqq \begin{cases}
        0 & \text{if $\abs{cX_j}>T$} \\
        X_j & \text{otherwise}
    \end{cases}.
\]
Note then that $X_j - \hat X_j\sim \mathcal D_p'$, so
\[
    \Pr\braces*{\Abs{\sum_{j=1}^d \bfx_j Y_j}\geq\lambda} = \Pr\braces*{\Abs{\sum_{j=1}^d \bfx_j(X_j-X_j')}\geq\lambda}.
\]
We first have by $p$-stability that
\[
    \Abs{\sum_{j=1}^d \bfx_j X_j} \stackrel{d}{=} \norm{\bfx}_p\abs{\hat X}
\]
for a $p$-stable variable $\hat X$, so there are constants $R$, $p$ such that
\[
    \Pr\parens*{\Abs{\sum_{j=1}^d \bfx_j X_j}\geq R\norm{\bfx}_p} = \Pr\parens*{\norm{\bfx}_p\abs{\hat X}\geq R\norm{\bfx}_p} = \Pr\parens*{\abs{\hat X}\geq R}\geq p.
\]
Next note that $X_j' \leq T/c = O(1)$ so
\[
    \E \Abs{\sum_{j=1}^d \bfx_j X_j'}\leq \sqrt{\E\Abs{\sum_{j=1}^d \bfx_j X_j'}^2} = \E \sqrt{\sum_{i=1}^d\sum_{j=1}^d \bfx_i\bfx_j\E[X_i'X_j']} = \sqrt{\sum_{j=1}^d \bfx_j^2\E X_j'^2} = O(\norm{\bfx}_2)
\]
by Jensen's inequality. Then by Markov's inequality, with probability at least $1 - p/2$, $\Abs{\sum_{j=1}^d \bfx_j X_j'} \leq C\norm{\bfx}_2$ for some constant $C$ that depends on $p$. Then for $\bfx$ such that $R\norm{\bfx}_p \geq 2C\norm{\bfx}_2$, we have by a union bound that
\[
    \Pr\braces*{\Abs{\sum_{j=1}^d \bfx_j(X_j-X_j')}\geq \frac{R}{2}\norm{\bfx}_p }\geq \Pr\braces*{\Abs{\sum_{j=1}^d \bfx_j X_j} - \Abs{\sum_{j=1}^d \bfx_j X_j'}\geq \frac{R}{2}\norm{\bfx}_p }\geq \frac{p}{2}.
\]
On the other hand, if $R\norm{\bfx}_p < 2C\norm{\bfx}_2$, the argument in Lemma \ref{lem:p>=2-expectation} shows that
\[
    \Pr\parens*{\Abs{\sum_{j=1}^d \bfx_j Y_j}\geq \Omega(\norm{\bfx}_2)} = \Omega(1)
\]
so the result holds under this case as well. 
\end{proof}

\begin{proof}[Proof of Theorem \ref{thm:p-1-2-countsketch}]
The distortion upper bound is just Lemma \ref{lem:no-expansion}. 

\paragraph{Mass of small entries.}

Let $\bfA = \bfA^H + \bfA^L$ as in Definition \ref{def:AH-AL}, with $T = O\parens*{(nd^2\log d / r)^{1/p}}$. 

By Lemma \ref{lem:chernoff-power-law-levels}, the sizes and mass of all level sets $\bfv_{(k)}$ with entries at most $2^k\leq T$ are concentrated around their means up to constant factors with probability at least $1 - \exp(-\Theta(n2^{-kp}))$. Thus by a union bound over $d$ columns $j$ and level sets $0\leq k\leq \log_2 T$, with probability at least
\[
    1 - d\sum_{k = 0}^{\log_2 T} \exp\parens*{-\Theta(n2^{-kp})} \geq 1 - d\exp\parens*{-\Theta(\log d)} = 1 - \frac1{\poly(r/d^2\log d)}
\]
we have for all $j\in[d]$ and $0\leq k\leq \log_2 T$ that
\[
    \norm{(\bfA\bfe_j)_{(k)}}_0 = \Theta(n2^{-kp}) \qquad \norm{(\bfA\bfe_j)_{(k)}}_1 = \Theta(n2^{k(1-p)})
\]
Then
\[
    \norm{\bfA^L\bfe_j}_1\leq \sum_{k=0}^{\log_2 T}\norm{(\bfA\bfe_j)_{(k)}}_1 = O(n).
\]

\paragraph{Mass of large entries.}

Furthermore, let $\bfB'$ be the subset of rows of $\bfA^H$ given by Lemma \ref{lem:p=1-unique-hashing-rows} that are hashed to locations without any other rows of $\bfA^H$. Recall also $\tau_1$ and $\tau_2$ from the lemma. 

We first have that $\norm{\bfS\bfB'\bfx}_1 = \Omega(\norm{\bfA^H\bfx}_1)$ since the rows containing entries larger than $\tau_1$ are perfectly hashed, while rows containing entries between $\tau_2$ and $\tau_1$ are preserved up to constant factors. 

Let $\bfB' = \bfB_{>T}' + \bfB_{\leq T}'$ where $\bfB'_{>T}$ contains the entries of $\bfB'$ that have absolute value greater than $T$ and $\bfB'_{\leq T}$ contains the rest of the entries. Note then that $\bfB'_{>T}$ has at most one nonzero entry per row, and $\bfB'_{\leq T}$ has at most $O(d\cdot r/d\log d) = O(r/\log d)$ nonzero entries and thus by Lemma \ref{lem:power-law-union-bound}, $\norm{\bfB'_{\leq T}}_\infty\leq O(r^{1/p})$ with probability at least $0.99$. We condition on this event. Then for all $\bfx$, 
\begin{align*}
    \norm{\bfS\bfA^H\bfx}_1 &\geq \norm{\bfS\bfB'\bfx}_1 \\
    &\geq \norm{\bfS\bfB_{> T}'\bfx}_1 - \norm{\bfS\bfB_{\leq T}'\bfx}_1 \\
    &= \sum_{j=1}^d \abs{\bfx_j}\norm{\bfB_{>T}'\bfe_j}_1 - \norm{\bfS\bfB_{\leq T}'\bfx}_1 && \text{$\bfB_{>T}'\bfe_j$ have disjoint support} \\
    &\geq \sum_{j=1}^d \abs{\bfx_j}\sum_{k=\log_2 \tau_2}^{\log_2 \tau_1}2^k\Theta(n/2^{kp}) - \norm{\bfB_{\leq T}'\bfx}_1 && \text{Lemmas \ref{lem:p=1-unique-hashing-rows} and \ref{lem:no-expansion}} \\
    &= \Omega\parens*{\frac{r}{d^2\log d}(nd^2\log d/r)^{1/p}}\norm{\bfx}_1 - O(r)\norm{\bfB_{\leq T}'}_\infty\norm{\bfx}_1 && \text{H\"older's inequality} \\
    &= \Omega\parens*{(r/d^2\log d)^{1-1/p}n^{1/p}}\norm{\bfx}_1 - O(r^{1+1/p})\norm{\bfx}_1 \\
    &= \Omega((r/d^2\log d)^{1-1/p}n^{1/p})\norm{\bfx}_1.
\end{align*}

\paragraph{Conclusion.}

On the other hand, by Lemma \ref{lem:p=1-small-SC1x}, the mass of the $O(r/d\log d)$ rows that are hashed together with the rows of $\bfA^H$ have mass at most
\[
    O\parens*{\frac{1}{\sqrt{\log d}}\frac{n^{1/p}}{(r/d^2\log d)^{1/p-1}}}\norm{\bfx}_1 = o\parens*{(r/d^2\log d)^{1-1/p}n^{1/p}}\norm{\bfx}_1. 
\]
Then,
\begin{align*}
    \frac1{\kappa} &\geq \frac{\norm{\bfS\bfA\bfx}_1}{\norm{\bfA\bfx}_1}\geq \frac{\norm{\bfS\bfA^H\bfx}_1 - \norm{\bfS\bfC_1\bfx}_1}{\norm{\bfA^H\bfx}_1 + \norm{\bfA^L\bfx}_1} \\
    &\geq \frac{\Omega(\norm{\bfA^H\bfx}_1 + (r/d^2\log d)^{1-1/p}n^{1/p})\norm{\bfx}_1}{O(\norm{\bfA^H\bfx}_1 + n)\norm{\bfx}_1} \geq \Omega\parens*{\parens*{\frac{(r/d^2\log d)}{n}}^{1-1/p}}.\qedhere
\end{align*}
\end{proof}

\subsection{Proofs for Section \ref{section:p>=2}}

\begin{proof}[Proof of Lemma \ref{lem:p>=2-expectation}]
For the upper bound, we have by Jensen's inequality that
\[
    \E_{\bfv\sim\mathcal D^d}\abs{\angle{\bfv,\bfx}}\leq \sqrt{\E_{\bfv\sim\mathcal D^d}\abs{\angle{\bfv,\bfx}}^2} = \sqrt{\sum_{j=1}^d \bfx_j^2\E\bfv_j^2} = O(\norm{\bfx}_2).
\]
We now focus on the lower bound.

Let $M = O(1)$ be the median of $\mathcal D$. We define $\bfw$ to be the truncation of $\bfv$ at $M$, that is, $\bfw_i = 0$ if $\abs{\bfv_i} > M$ and $\bfw_i = \bfv_i$ otherwise. Then by \cite[Lemma 6.1.2]{vershynin2018high}, 
\[
    \E\abs{\angle{\bfv,\bfx}}\geq \E\abs{\angle{\bfw,\bfx}}
\]
so it suffices to bound $\E\abs{\angle{\bfw,\bfx}}$ instead.

Note that
\[
    \E\abs{\angle{\bfw,\bfx}}^2 = \sum_{i=1}^d\sum_{j=1}^d E(\bfw_i\bfx_i\bfw_j\bfx_j) = \sum_{j=1}^d \bfx_j^2\E\bfw_j^2 = \Omega(\norm{\bfx}_2^2)
\]
and
\[
    \E\abs{\angle{\bfw,\bfx}}^4 = \sum_{j=1}^d \bfx_j^4\E\bfw_j^4 + 3\sum_{j\neq k}^d \bfx_j^2\bfx_k^2 \E(\bfw_j^2\bfw_k^2)\leq O(\norm{\bfx}_4^4 + \norm{\bfx}_2^4) = O(\norm{\bfx}_2^4)
\]
so by the Paley-Zygmund inequality,
\[
    \Pr\parens*{\abs{\angle{\bfw,\bfx}}\geq \sqrt{\lambda}\sqrt{\E\abs{\angle{\bfw,\bfx}}^2}} = \Pr\parens*{\abs{\angle{\bfw,\bfx}}^2\geq \lambda\E\abs{\angle{\bfw,\bfx}}^2}\geq (1-\lambda)^2\frac{(\E\abs{\angle{\bfw,\bfx}}^2)^2}{\E\abs{\angle{\bfw,\bfx}}^4} = \Omega(1).
\]
Thus $\abs{\angle{\bfw,\bfx}} = \Omega(\norm{\bfx}_2)$ with constant probability and thus $\E\abs{\angle{\bfw,\bfx}} = \Omega(\norm{\bfx}_2)$, as desired.
\end{proof}

\begin{proof}[Proof of Lemma \ref{lem:p>=2-conditioning-expectation-same}]
Let $X\coloneqq \angle{\bfv,\bfx}$. We have
\[
    \Pr(\neg\mathcal E_i) = B^{-p} \leq \frac{\eps}{d}
\]
so by the union bound,
\[
    \Pr(\mathcal E)\geq 1 - \sum_{i=1}^d \Pr(\neg\mathcal E_i) = 1 - d\Pr(\neg\mathcal E_1) =  1 - \eps. 
\]

For $B$ large enough, we have by the layer cake theorem that
\[
    \E_{Y\sim\mathcal D}\parens*{\abs{Y}\mid \abs{Y} > B} \leq \frac1{\Pr(\abs{Y} > B)}\int_B^\infty  O(x^p) ~dx = \frac{1}{\Omega(B^{-p})} O(B^{1-p}) = O(B)
\]
since $p\geq 2 > 1$. Then,
\begin{align*}
    \E\parens*{\abs{X}\mid \neg\mathcal E} &\leq \sum_{i=1}^d \E\parens*{\abs{\bfv_i\bfx_i}\mid \neg\mathcal E} \\
    &= \sum_{i=1}^d \E\parens*{\abs{\bfv_i\bfx_i}\mid \neg \mathcal E_i, \neg\mathcal E} \Pr\parens*{\neg \mathcal E_i\mid \neg\mathcal E} + \E\parens*{\abs{\bfv_i\bfx_i}\mid\mathcal E_i, \neg\mathcal E} \Pr\parens*{\mathcal E_i\mid \neg\mathcal E}  \\
    &\leq \sum_{i=1}^d O(B)\abs{\bfx_i} \frac1d + O(\abs{\bfx_i}) \\
    &\leq O(B+d)\norm{\bfx}_1.
\end{align*}

We then have
\[
    \E\abs{X} = \E\parens*{\abs{X}\mid \mathcal E}\Pr\parens*{\mathcal E} + \E\parens*{\abs{X}\mid \neg\mathcal E}\Pr\parens*{\neg\mathcal E}\leq \E\parens*{\abs{X}\mid \mathcal E} + O(B+d)\norm{\bfx}_1\Pr\parens*{\neg\mathcal E}.
\]
Since $\E\abs{X} = \Omega(\norm{\bfx}_2) = \Omega(\norm{\bfx}_1 / \sqrt{d})$ by Lemma \ref{lem:p>=2-expectation}, 
\[
    O(B+d)\norm{\bfx}_1\Pr\parens*{\neg\mathcal E} = O(B+d)\norm{\bfx}_1\frac1{B^p} \leq O(\eps)\E\abs{X}
\]
by our choice of $B$. We thus have
\[
    \E\abs{X} \leq \E\parens*{\abs{X}\mid \mathcal E} + O(\eps)\E\abs{X}
\]
so
\[
    \E\parens*{\abs{X}\mid \mathcal E}\geq (1-O(\eps))\E\abs{X}.\qedhere
\]
\end{proof}